\documentclass{iopart}

\usepackage{graphicx,float}
\usepackage{hyperref}
\usepackage{amsthm,amssymb}
\expandafter\let\csname equation*\endcsname\undefined
\expandafter\let\csname endequation*\endcsname\undefined
\usepackage{amsmath}
\usepackage{physics}
\usepackage{enumitem}
\usepackage{mathrsfs}
\usepackage{dsfont}
\usepackage[most]{tcolorbox}
\usepackage[sort&compress,numbers]{natbib}

\definecolor{shadecolor}{gray}{0.9}

\newcommand{\wt}{\widetilde}
\newcommand{\ol}{\overline}

\newcommand{\ve}{\varepsilon}
\newcommand{\mrm}{\mathrm}
\newcommand{\mbb}{\mathbb}
\newcommand{\mbf}{\mathbf}
\newcommand{\Bve}{\mathscr{B}^{\ve}}
\renewcommand{\op}[2]{|#1\rangle\langle #2|}
\newcommand{\supp}{\mrm{supp}}
\newcommand{\ext}{\mrm{ext}}

\renewcommand{\vec}{\mathrm{vec}}
\newcommand{\cv}{\text{cv}}

\newcommand{\cB}{\mathcal{B}}
\newcommand{\cC}{\mathcal{C}}
\newcommand{\cD}{\mathcal{D}}
\newcommand{\cE}{\mathcal{E}}
\newcommand{\cF}{\mathcal{F}}
\newcommand{\cG}{\mathcal{G}}
\newcommand{\cH}{\mathcal{H}}
\newcommand{\cI}{\mathcal{I}}

\newcommand{\cK}{\mathcal{K}}

\newcommand{\cN}{\mathcal{N}}

\newcommand{\cP}{\mathcal{P}}

\newcommand{\cS}{\mathcal{S}}

\newcommand{\cU}{\mathcal{U}}
\newcommand{\cV}{\mathcal{V}}
\newcommand{\cW}{\mathcal{W}}
\newcommand{\cX}{\mathcal{X}}
\newcommand{\cY}{\mathcal{Y}}
\newcommand{\cZ}{\mathcal{Z}}

\newcommand{\Lin}{\mathrm{L}}
\newcommand{\Trans}{\mathrm{T}}
\newcommand{\Pos}{\mathrm{Pos}}
\newcommand{\Herm}{\mathrm{Herm}}
\newcommand{\Channel}{\mathrm{C}}
\newcommand{\CPU}{\mathrm{CPU}}
\newcommand{\CPTPU}{\mathrm{CPTPU}}
\newcommand{\CU}{\mathrm{CU}}

\newcommand{\Density}{\mathrm{D}}
\newcommand{\CP}{\mathrm{CP}}
\newcommand{\Sep}{\mathrm{Sep}}

\newcommand{\PPT}{\mathrm{PPT}}

\newcommand{\id}{\textrm{id}}

\newcommand{\relint}{\mrm{Relint}}
\newcommand{\Ent}{\mrm{Ent}}

\newcommand{\EntS}{\mbb{E}\mrm{nt}}
\newcommand{\CQ}{\mrm{CQ}}

\newcommand{\I}{I}

\renewcommand{\ip}[2]{\langle #1 , #2\rangle}

\newtheorem{theorem}{Theorem}
\newtheorem{lemma}[theorem]{Lemma}

\newtheorem{definition}{Definition}

\newtheorem{corollary}{Corollary}

\newtheorem{proposition}{Proposition}
\theoremstyle{remark}
\newtheorem{remark}{Remark}
\newtheorem*{summary}{Summary of Main Results}

\newenvironment{miniproblem}[1]{
  \begin{minipage}[t]{#1\textwidth}
    \begin{center}}{
    \end{center}
\end{minipage}}

\begin{document}
\title{Cone-Restricted Information Theory}

\author{Ian George and
Eric Chitambar}

\address{Department of Electrical and Computer Engineering,
University of Illinois at Urbana-Champaign}
\ead{igeorge3@illinois.edu}

\begin{abstract}
The max-relative entropy and the conditional min-entropy it induces have become central to one-shot and zero-error quantum information theory. Both may be expressed in terms of a conic program over the positive semidefinite cone. Recently, it was shown that the same conic program over the separable cone admits an operational interpretation in terms of communicating classical information over a quantum channel. In this work, we generalize this framework to determine which results in quantum information theory rely upon the positive semidefinite cone and which can be generalized. We show that the standard asymptotic equipartition properties break down if the cone never approximates the positive semidefinite cone. This shows that such altered measures characterize different tasks than the traditional quantities even asymptotically. We present parallel results for the extended conditional min-entropy, which requires extending the notion of k-superpositive channels to superchannels. We also show that the near-equivalence of the smoothed max and Hartley entropies breaks down in this setting. However, we show for classical-quantum states, the separable cone is sufficient to re-cover the asymptotic theory, thereby drawing a strong distinction between the fully and partial quantum settings. We also present operational uses of this framework. We first show the cone restricted min-entropy of a Choi operator captures a measure of entanglement-assisted noiseless classical communication using restricted measurements. We show that quantum majorization results naturally generalize to other cones. As a novel example, we introduce a new min-entropy-like quantity that captures the quantum majorization of quantum channels in terms of bistochastic pre-processing. Lastly, we relate this framework to general conic norms and their non-additivity. Throughout this work we emphasize the introduced measures’ relationship to general convex resource theories. In particular, we look at both resource theories that capture locality and resource theories of coherence/Abelian symmetries.
\end{abstract}

\maketitle 

\tableofcontents

\section{Introduction}
\label{sec:introduction}
Quantum information theory is perhaps best set apart from its classical subtheory in the variety of resources it supplies for information processing tasks\cite{Chitambar-2019a}, such as entanglement and coherence.  A number of new information measures also emerge in the quantum setting due to both the non-commutative aspect of quantum mechanics as well as the plethora of ways that quantum resources can be transformed in the one-shot setting \cite{Tomamichel-2015,Entropy-Zoo}.  
As such, the quantum information theorist is always interested in a better understanding of the interplay between the use of a resource for an information processing task, the ability to convert between resources, the measures that capture resourcefulness, and the operational interpretations of said measures. 

At an even more abstract level, the goal of the information theorist is to characterize operational tasks in terms of (entropic) measures. What characterizations can be established depends on the theory, e.g.\ the existence of a single-letter characterization of classical capacity when restricting to classical channels is not, at least currently, known to extend for quantum information theory because it includes quantum channels and entangled states \cite{Holevo-1998a,Schumacher-1997a}. It is therefore interesting to find manners of determining where characterizations may break down or change as the physical resources considered in information theory are restricted. A refined delineation between quantum and classical information theory, which is one way to interpret many of the results of this work, may be viewed as one motivation for the altered measures we consider.

The interplay between resources, operational tasks, and entropic measures may be seen by considering \textit{quantum majorization}.  Broadly speaking, quantum majorization describes when zero-error convertibility of one quantum object to another is feasible using a given family of allowed transformations.  When the two objects belong to some relevant set, this feasibility in conversion induces a partial order on the set.  Quantum majorization is analogous to the well-known notion of vector majorization, which is said to hold when one real vector can be converted to another by some stochastic channel.  More generally, one can consider a majorization relation between classical channels defined by the convertibility of one to another under post-processing conversion, as formulated by Blackwell \cite{Blackwell-1953a,Blackwell-1953b}.
A quantum notion of "Blackwellian" majorization for quantum channels was obtained in \cite{Buscemi-2016a}.  Similarly, a notion of quantum majorization on bipartite states via local processing, $\sigma_{AB} \to \sigma_{AC}$ using $\Phi_{B \to C}$, was obtained in \cite{Gour-2018a}, and a notion of converting bipartite channels was obtained in \cite{Gour-2019a} and subsequently improved in \cite{Gour-2020a,Ji-2021a}. Beyond the obvious operational interest of these quantum versions, it is worth remarking that in each case it is shown that one element of the set majorizes the other if and only if some family of relations involving the conditional min-entropy holds.

The relation of majorization to the min-entropy is interesting for multiple reasons. First, the min-entropy has operational interpretations in terms of singlet fraction and guessing probability \cite{Konig-2009a}. As such, if a form of quantum majorization relates to conditions on the min-entropy, it also relates to conditions on the entanglement or distinguishability. Second, while the min-entropy is a one-shot quantity, when smoothed its regularization is known to converge to the conditional von Neumann entropy for identically and independently distributed (i.i.d.) inputs as the number of copies tends to infinity. This is known as the (fully quantum) asymptotic equipartition property (AEP) and is crucial for recovering asymptotic behaviour from one-shot behaviour. However, the smooth min-entropy AEP may be shown to be in a sense inherited from the smoothed max-relative divergence \cite{Tomamichel-2015}. At the same time, the max-relative divergence itself is of interest not only because it induces the min-entropy, but because the behaviour of its regularization in the limit for i.i.d.\ inputs, analogous to the smooth min-entropy AEP, quantifies the error exponent in quantum Stein's lemma to first-order.\footnote{We stress this work is not directly focused on Stein's lemma and is independent of any recent issues related to the generalized Stein's lemma often considered in resource theories \cite{Brandao-2010a,generalized-stein-gap}.} These relations along with a few others of interest are summarized in Figure \ref{fig:Entropies-and-their Uses}.

\begin{figure}
    \centering
    \includegraphics[width=0.7\linewidth]{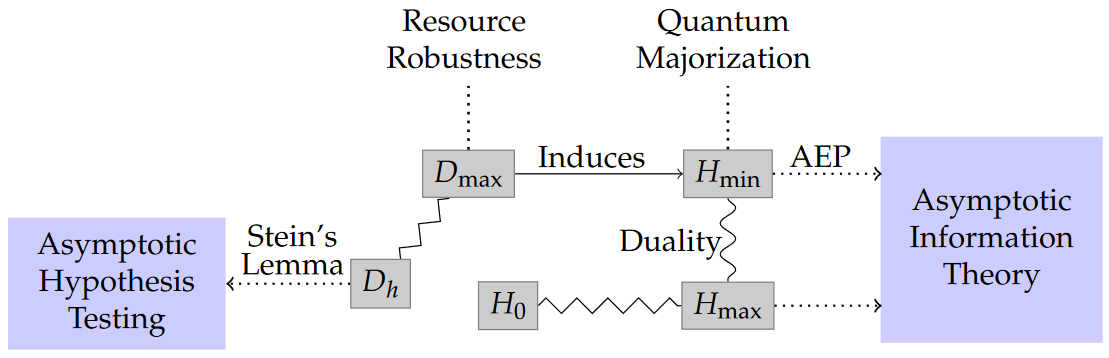}
    \caption{Diagram of common uses of the relative max-entropy, its induced entropy, and related measures. Dotted arrows denote information-theoretic applications achieved via the limit of the regularized smoothed measure. Dotted lines represent uses in resource-theoretic settings. Zig-zag line represented a near equivalence under some choice of smoothing. A wavy line represents duality of measures. The standard arrow shows what measures induce others.\label{fig:Entropies-and-their Uses}}
\end{figure}

As such, entropic measures, resource conversion, and information-theoretic tasks are strongly related at a formal level. These relations have only grown more complicated recently as a variation on the min-entropy was found to admit an operational interpretation in terms of classical communication over a quantum channel \cite{Chitambar-2021a}. This variation is obtained by taking the optimization program corresponding to the min-entropy and restricting it's feasible region to the set of separable operators. For this reason, it is called the \textit{separable min-entropy}. It follows that the separable min-entropy is always an upper bound on the standard min-entropy. Furthermore, it was shown that, unlike the min-entropy, the separable min-entropy is a non-additive measure in general. These properties call into question the behaviour of the regularized separable min-entropy for i.i.d.\ states in the limit of many copies in \textit{both} the smoothed and unsmoothed cases. As the separable min-entropy already admits an operational interpretation, it is not obvious that other cones contained in the positive cone could not induce min-entropies with operational interpretations as well.
Moreover, arguably, it is not clear the positive cone and standard min-entropy should be uniquely central in information theory given the existence of an operational interpretation of the separable min-entropy.

It is this train of thought which this work aims to resolve. First, we introduce the notion of the max-relative entropy defined over closed convex cones which we call \textit{restricted max-relative entropies}. We use these to define the \textit{restricted min-entropies}. We then use this framework to show if we consider sequences of cones that do not asymptotically become equivalent to the positive semidefinite cone, then the smoothed regularized quantities do not converge to the relative and conditional entropy respectively. These results show tasks characterized by the restricted measures do not even asymptotically become equivalent to the unrestricted case. In particular, this means these measures cannot characterize traditional fully quantum results even to first order. We refer to such results as `anti-standard AEP' results as, while they don't tell us what the regularized quantity converges to, they tell us it doesn't converge to the traditional measure. These anti-standard AEPs strongly rely on the continuously growing viewpoint of thinking of min-entropy-like quantities in terms of support functions \cite{Gour-2019a,Jencova-2021a}. Furthermore, as we will see, this approach is not only information-theoretically interesting, but deeply related to resource theories. This is largely because many resource theories can be formulated in terms of convex sets \cite{Chitambar-2016a}, although we note our results also apply to any resource theory whose underlying resource has an \textit{extensive} nature to it \cite{Vijayan-2020a}, as defined more precisely in Section \ref{sec:anti-stein-lemma}.

Moreover, by defining a restricted max-relative entropy in terms of the restricted min-entropy, we are able to show that in the previously stated regime, the restricted max-entropy's smoothed regularization won't converge to the conditional entropy even for classical states. However, we show we can construct a restricted Hartley entropy that does work properly on classical states. This exemplifies a strong gap between the restricted Hartley and restricted max-entropy that is known to not exist for the positive semidefinite cone \cite{Tomamichel-2012a}. All of these points taken together further cement the positive semidefinite cone as central to fully quantum information theory much in the same way that complete positivity being captured by the positive semidefinite cone under the Choi-Jamiolkowski isomorphism does.

Furthermore, by first considering symmetries, we show that if we only consider classical-quantum (CQ) states, then the classical-quantum cone and its ordering is sufficient for establishing known equivalences of entropic and operational definitions at least for point-to-point quantum information theory. Moreover, the separable cone is the minimal cone under the demand that the cone is invariant under local unitaries, and so in some sense captures its own privileged position. Specifically, the L\"{o}wner order is only necessary for establishing certain equivalences of entropic and operational definitions for \textit{fully quantum} information theory.

To exemplify the generality of this method, we extend the anti-standard AEP to the recently introduced extended min-entropy $H^{\ext}_{\min}(B|A)_{\Phi}$ for a bipartite channel $\Phi:A_0 B_0 \to A_1 B_1$ \cite{Gour-2019a}, which has also been recently extended to the general cone case \cite{Ji-2021a}. In doing so, we contribute to the entanglement theory for supermaps by expanding on \cite{Chen-2020a} and introducing the `flipped' extended min-entropy $H^{\ext}_{\min}(A|B)_{\Phi}$, again making use of the characterization of min-entropy-like quantities as support functions. 

Having resolved the general framework and the breakdown of asymptotic results, we turn our attention to operational interpretations of the cone-restricted theory. First, we show that every restricted min-entropy for a convex closed cone that is invariant under local unitaries is captured by an entanglement-assisted communication game with restricted measurements (See Figure \ref{fig:intro-gen-cv}). We refer to these as cone-restricted communication values as they fit into the framework of communication value \cite{Chitambar-2021a}. This resolves an open problem from \cite{Chitambar-2021a}.
\begin{figure}
    \centering
    \includegraphics[width=0.7\linewidth]{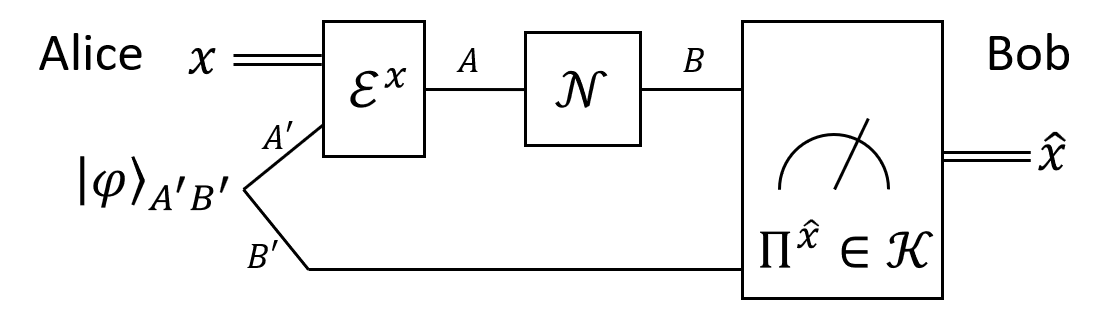}
    \caption{Depiction of the operational interpretation of cone-restricted min-entropy. Alice encodes half of an entangled state and sends it through the channel to Bob. Bob then measures the entire state, but is restricted to measurements in the cone. The unnormalized success probability that Bob gets Alice's message is given by the restricted min-entropy.}
    \label{fig:intro-gen-cv}
\end{figure}

Second, we show a notion of quantum majorization of channels in terms of (possibly restricted) sets of bistochastic pre-processing in terms of a novel restricted min-entropy evaluated on the Choi matrix. In doing so, we stress the connection between majorization results and the viewpoint that min-entropic quantities are related to support functions.

Finally, we discuss further refinements of the behaviour of the restricted max-relative entropy and its relation to general conic norms. Specifically, we show for any cone captured by (an Abelian) symmetry and operators restricted to the cone, one can recover all the standard properties of the max-relative entropy and min-entropy. We however show this does not---at least obviously---extend to the case of cones which capture a notion of entanglement. We also show in what cases one can reduce the determination of the restricted max-relative entropy and min-entropy to a cone norm---a property that always holds when the cone is the positive semidefinite cone \cite{Datta-2009a}. Regardless, we show that even when one can reduce the problem to determining the cone norm of an operator, things are not straightforward by considering the Werner states. As the total work is lengthy, we summarize the results below.

\begin{tcolorbox}[breakable,width=\columnwidth , sharp corners, frame hidden, boxrule = 0pt]
\begin{summary}
~\\[-4mm]
\begin{enumerate}[label=\arabic*.]
    \item The asymptotic equipartition properties (AEPs) for max-relative divergence and conditional min-entropy require the positive semidefinite cone. 
    \begin{itemize} 
        \item This arises from the restricted cone measures losing their proper order, i.e.\ there exist cones $\cK \subset \Pos$ such that $D_{\max}^{\cK} < D$. This also implies $\log$ is not an operator monotone under these generalized cone inequalities.
        \item The first-order characterization of fully quantum Stein's lemma does not correspond to the regularized cone measure unless the sequence of cones asymptotically contains the $n-$fold copy of the maximally coherent state for every basis of $A$. This in a sense establishes coherence as the resource that captures the possibility of a fully quantum Stein's lemma.
        \item The min-entropy AEP does not hold for a sequence of cones parameterized in terms of Schmidt rank unless the cone asymptotically contains all entangled states, i.e.\ the positive semidefinite cone. This establishes entanglement as the natural resource relevant for a fully quantum AEP, which is important for various one-shot task characterizations.
    \end{itemize}
    \item An AEP property for the extended min-entropy, which is a measure on bipartite channels, requires the positive semidefinite cone.
    \begin{itemize}
        \item This makes use of our introduction and characterization of the superchannel equivalent of $k-$superpositive maps.
    \end{itemize}
    \item The separable cone is sufficient for an AEP on CQ states for min-entropy, but the max-entropy defined to be its dual does not have an AEP. However, there exists a natural extension of the Hartley entropy that does satisfy an AEP. This establishes a large gap between the smoothed max-entropy and Hartley entropy that does not hold for the positive semidefinite cone \cite{Tomamichel-2011a}.
    \begin{itemize}
        \item That a restricted max-entropy AEP for CQ states fails may be argued to be because two registers being optimally decoupled can be achieved by only unrestricted local dynamics on the purified space.
        \item That a restricted Hartley entropy does satisfy an AEP follows from the fact there are restricted Hartley entropies that act like the Hartley entropy on CQ states, so the duality between Petz entropies and specific Sandwiched R\'{e}nyi entropies allows an AEP to be established.
    \end{itemize}
    \item The separable cone is the minimal locally-unitarily invariant cone that is sufficient for recovering standard information-theoretic results for CQ states.
    \item We establish the restricted min-entropy for all cones that are invariant under local completely-positive maps. Specifically, it characterizes a form of entanglement-assisted minimal error discrimination using restricted measurements (See. Fig \ref{fig:intro-gen-cv}). This resolves an open problem of \cite{Chitambar-2021a}.
    \item Necessary and sufficient conditions for conversion between channels via noisy pre-processing in terms of a new min-entropic-like quantity that holds for a large class of cones.
    \item The relationship between generalized conic operator norms and the generalized max-relative entropy is established. Strange properties are exemplified using the Werner states.
\end{enumerate}
\end{summary}
\end{tcolorbox}

\section{Notation \& Background}\label{sec:background}
In this section we introduce our notation and review the necessary background. Much of our notation follows that of \cite{Watrous-2018}. Finite alphabets are denoted $\Sigma,\Lambda,$ etc. The probability simplex over a finite alphabet $\Sigma$ is denoted $\cP(\Sigma)$. All results in this work consider finite-dimensional Hilbert spaces, $\mbb{C}^{|\Sigma|}$, for some finite alphabet $\Sigma$. At times it is more natural to talk in terms of the dimension, in which case we will talk of $d$-dimensional Hilbert spaces, $\mbb{C}^{d}$. In these cases we may iterate over the set $[d]:=\{1,...,d\}$. These finite-dimensional Hilbert spaces are labeled by a capital roman alphabet letter, e.g.\ $A \cong \mbb{C}^{d}$. A linear operator from $A$ to $B$ is denoted $\Lin(A,B)$. Given say $X,Y \in \Lin(A,B)$, we write $X \ll Y$ if the kernel of $X$ contains the kernel of $Y$.  The special case of endomorphisms have simplified notation, $\Lin(A) := \Lin(A,A)$. Of particular interest will be the set of Hermitian operators $\Herm(A) := \{X \in \Lin(A),\, X=X^{\ast}\}$, where ${}^{\ast}$ denotes the complex conjugate, the positive semidefinite operators $\Pos(A) := \{Y^{\ast}Y: Y \in \Lin(A)\}$, and the density matrices (quantum states) $\Density(A) := \{\rho \in \Pos(A) : \Tr[\rho] = 1 \}$. It is easy to verify that $\Pos(A)$ is the subset of $\Lin(A)$ with non-negative eigenvalues. If we relax to the positive operators with trace less than or equal to one, we denote it $D_{\leq}(A)$. To avoid confusion, at times we will introduce a subscript to denote what spaces an endomorphism are on, e.g.\ $X \in \Lin(A \otimes B)$, will be denoted $X_{AB}$.

Transformations from $\Lin(A)$ to $\Lin(B)$ are denoted $\Trans(A,B)$. To avoid confusion, at times we will introduce a subscript to denote what space the transformation acts on, e.g. $\Phi_{A \to B} \in \Trans(A,B)$, and the identity channel is denoted $\id_{A}$. We denote the set of quantum channels by $\Channel(A,B)$, where we remind the reader that $\Phi \in \Channel(A,B)$ if and only if $\Phi \in \Trans(A,B)$ is trace-preserving (TP), $\Tr[\Phi(X)] = \Tr[X]$ for all $X \in \Lin(A)$, and completely positive (CP), $(\id_{R} \otimes \Phi)(P) \in \Pos(R\otimes B)$ for all $R$ and all $P \in \Pos(R \otimes A)$. We denote the set of completely-positive (resp. trace-preserving) maps as $\mrm{CP}(A,B) \subset \mrm{Trans}(A,B)$ (resp. $\mrm{TP}(A,B) \subset \Trans(A,B)$).

We use the vec mapping $\vec: \Lin(A,B) \to A \otimes B$ which is defined by $\vec(E_{i,j}) = e_{j} \otimes e_{i}$ where $E_{i,j}$ and $e_{j}$ are standard basis operators and vectors respectively. For those more comofortable with bra-ket notation, $\vec(\op{i}{j}) = \ket{j} \otimes \ket{i}$ in the computational basis. This choice of the vec mapping is known to satisfy the identity
\begin{equation}\label{eqn:vec-map-identity}
    (X_{1}^{\Trans} \otimes X_{0})\vec(Y) = \vec(X_{0}YX_{1})
\end{equation}
where $X_{0} \in \Lin(A_{0},B_{0})$, $X_{1} \in \Lin(A_{1},B_{1})$, and $Y \in \Lin(B_{1},B_{0})$.
One reason for choosing this definition of the vec mapping is that it corresponds to the version of the Choi-Jamiolkowski Isomorphism which we will use throughout this work: Given a map $\Phi \in \Trans(A,B)$, the Choi operator of $\Phi$, is given by
\begin{align} 
J_{\Phi}^{AB} =& (\id^{A} \otimes \Phi)\left(\phi^{+}_{A\wt{A}}\right) 
= \sum_{i=1}^{|\Sigma|} \vec(A_{i})\vec(B_{i})^{\ast}  \label{eqn:choi-in-kraus-terms}\ ,
\end{align}
where $\{A_{i}\}_{i \in \Sigma},\{B_{i}\}_{i \in \Sigma} \subset \Lin(A,B)$ are the Kraus operators, i.e.\ $\Phi(X) = \sum_{i \in \Sigma} A_{i}XB_{i}^{\ast}$ (recall that in the special case that the map is completely positive, $B_{i} = A_{i}$ for all $i \in \Sigma$). Furthermore, $\phi^{+}_{A\wt{A}}$ is the unnormalized maximally entangled state operator, $\phi^{+} = \dyad{\phi^{+}}$ where $\ket{\phi^{+}}_{A\wt{A}} = \sum_{i\in[d_{A}]} \ket{ii}$. The normalized maximally entangled state will be denoted $\ket{\tau}_{A\wt{A}}$. Lastly, we recall the definition of adjoint map, $\Phi^{\ast}$ of a linear map $\Phi: A \to B$: 
$$\langle Y, \Phi(X) \rangle = \langle \Phi^{\ast}(Y) , X \rangle$$
for all $X \in \Lin(A),Y \in \Lin(B)$. One may verify that $J^{\Phi^{\ast}}_{BA} = (\mbb{F}J^{\Phi}_{AB}\mbb{F})^{\Trans}$, where $\mbb{F}$ is the swap operator between spaces $A$ and $B$ and $\Trans$ is the transpose.

\subsection{Convex Geometry for Information Theory \& Resource Theories}
Unsurprisingly, convex geometry is integral to this work, and we briefly review the topic here. While this summary is sufficient, we refer the reader to \cite{WatrousATQILecture01} for further details.

A set $S$ is \textit{convex} if for all $s_{1},s_{2} \in S$ and $\lambda \in [0,1]$, $\lambda s_{1} + (1-\lambda)s_{2}$. The \textit{extreme points} of a convex set $S$ are $s \in S$ such that $s = \lambda s_{1} +(1-\lambda)s_{2}$ if and only if for $\lambda \in (0,1)$, $s = s_{1} = s_{2}$. As an example, it is straightforward to verify that the sets $\Density(A)$ and $\Channel(A,B)$ are closed (in fact, compact) convex sets. Moreover, the extreme points of the density matrices are the pure states, $\op{\psi}{\psi}$ where $\ket{\psi} \in A$.

Given a finite-dimensional real inner product space $V$, a set $\cK \subset V$ is a \textit{cone} if for all $k \in \cK$ and $\lambda \geq 0$, $\lambda k \in \cK$. Given a cone $\cK \subset V$, its \textit{dual cone} is defined as $\cK^{\ast} := \{v: \langle v,k \rangle \,\, \forall k \in \cK\}$, where $\langle \cdot , \cdot \rangle$ is the inner product of $V$. It is straightforward to verify that given cones $\cK_{1} \subset \cK_{2} \subset V$, then $\cK_{2}^{\ast} \subset \cK_{1}^{\ast}$. Given a cone $\cK \subset V$, it induces a partial ordering on $V$ defined by $X \preceq_{\cK} Y \Leftrightarrow Y - X \in \cK$. The relative interior of a cone is given by 
\begin{equation}\label{eqn:relint-defn}
\begin{aligned}
\relint(\cK) := \{x \in \cK : \forall y \in \cK, \exists \varepsilon > 0 : (1+\ve)x-\ve y \in \cK \} \, . 
\end{aligned}
\end{equation}

A particularly nice property of most cones in quantum information theory is that they can be specified by the set of density matrices that satisfy the condition of the cone.
\begin{proposition}\label{prop:cone-to-density}(\cite[Prop. 6.4]{Watrous-2018}) Let $A$ be a finite dimensional Hilbert space. Let $\cK \subseteq \Pos(A)$. If $\cD := \cK \cap \Density(A)$ is non-empty, then $\cK = \mrm{cone}(\cD)$.
\end{proposition}

We define a linear conic program and its dual in the same manner as in \cite{WatrousATQILecture01}.
\begin{center}
    \begin{miniproblem}{0.45}
      \emph{Primal}\\[-5mm]
      \begin{equation}
      \begin{aligned}\label{eqn:conicPrimal}
        \text{max.}\quad & \ip{a}{k} \\
        \text{s.t.}\quad & \phi(k) = b \\
        & k \in \cK
      \end{aligned}
      \end{equation}
    \end{miniproblem}
    \begin{miniproblem}{0.45}
      \emph{Dual}\\[-5mm]
      \begin{equation}
      \begin{aligned}\label{eqn:conicDual}
        \text{min.}\quad & \ip{b}{y}  \\
        \text{s.t.}\quad & \phi^{\ast}(y) \succeq_{\cK^{\ast}} a \\
        & y \in W
      \end{aligned}
      \end{equation}
    \end{miniproblem}
 \end{center}
where $W$ is a real inner product space, $\phi: V \to W$ is a linear map, $a \in V$, $b \in W$, and $\phi^{\ast}$ is the adjoint map of $\phi$ and $\cK$ is assumed to be closed and convex. When the two programs achieve the same value this is referred to as strong duality.

As an example of these definitions, it is easy to verify $\Pos(A)$ is a convex cone because if a matrix has non-negative eigenvalues, it will maintain this property under scaling. It may be shown that $\relint(\Pos(A))$ is the space of positive definite operators.
Using that elements of $\Pos(A)$ have non-negative eigenvalues, one can verify that the dual cone is $\Pos(A)^{\ast} = \Pos(A)$, i.e.\ the positive semidefinite cone is self-dual. One can then consider the L\"{o}wner order which is the ordering induced by the positive cone. Given the nice properties just exemplified, we define $X \succeq Y \Leftrightarrow X \succeq_{\Pos} Y$ for $X,Y \in \Lin(A)$. This will allow us to avoid confusion with less common orderings throughout the work. Lastly we note that in this special case that the cone is $\Pos(A)$, conic programming reduces to semidefinite programming which is known to be efficient.

Finally, given a non-empty, closed convex set $\cC \subset \mbb{R}^{n}$, the support function $\supp_{\cC}:\mbb{R}^{n} \to \mbb{R}$ is defined by
\begin{align}\label{eqn:supp-func-defn}
\supp_{\cC}(x) := \sup \{ \langle x , c \rangle : c \in \cC \} \, .
\end{align}
Very roughly speaking, this can be seen as a measure of how well one can decompose $x$ into a single element of $\cC$.

\subsection{Generalized Entropies}
While the von Neumann entropy has been shown to be very successful for capturing the fundamental limits of asymptotically vanishing-error quantum information processing, there has been great progress in expanding beyond this regime by considering generalized entropies. Of particular interest for this work will be the max-relative entropy or max-divergence, which for $P,Q \in \Pos(A)$ is given by
\begin{align}
    D_{\max}(P||Q) & := \inf\{\alpha: P \leq 2^{\alpha} Q \} \ . \label{eqn:Dmax-defn}
\end{align}
This measure has various nice properties, a few of which we now note. The max-divergence is known to be the smallest generalization of the classical max-divergence that is both additive and satisfies contractivity under CPTP maps \cite[Prop. 4.3]{Tomamichel-2015}. It is known to be a limit of the Sandwiched R\'{e}nyi Divergences that have found great use recently: $D_{\max}(P||Q) = \lim_{\alpha \to \infty} \wt{D}_{\alpha}(P||Q)$ where
\begin{align}\label{eqn:SRD-defn}
    \wt{D}_{\alpha}(P||Q) = \frac{1}{\alpha -1} \log \frac{\|Q^{\frac{1-\alpha}{2\alpha}}PQ^{\frac{1-\alpha}{2\alpha}}\|_{\alpha}^{\alpha}}{\Tr(P)} \ .
\end{align}
$D_{\max}$ also induces the min-entropy $H_{\min}(A|B)$ by following the standard rule for inducing an entropy from a divergence \cite{Tomamichel-2015}:
\begin{align}\label{eqn:rule-for-inducing-entropies}
H_{\min}(A|B)_{\rho} := \sup_{\sigma_{B} \in \Density(B)} - D_{\max}(\rho_{AB}||I_{A} \otimes \sigma_{B}) \ , 
\end{align}
where $I_{A}$ is the identity on $A$.

Furthermore, $H_{\min}(A|B)$ is well-known to satisfy a particularly nice duality property: for $\ket{\psi} \in A \otimes B \otimes C$, 
\begin{align}\label{eqn:min-max-duality} H_{\min}(A|B)_{\rho_{AB}} = - H_{\max}(A|C)_{\rho_{AC}} \ ,
\end{align}
where $H_{\max}$ is the entropy induced by $\wt{D}_{\min}(P||Q) := \lim_{n \to \infty} \wt{D}_{\alpha}(P||Q)$. We pre-emptively note there is an alternative max-entropy known as the quantum Hartley entropy, denoted $H_{0}(A|B)$, but under smoothing it is known to be logarithmically (in the smoothing parameter) equivalent \cite{Tomamichel-2011a}.

\subsubsection{Recovering Asymptotic Results - Smoothed Entropies}\label{sec:recovering-asymptotic-results}
These aforementioned measures are useful for the non-asymptotic setting as already noted. One would hope that under some setting, one could recover the asymptotic results. To do this, smoothed entropies were introduced \cite{Renner2005}. In particular,
\begin{align}\label{eq:smoothed-min-ent}
    H^{\ve}_{\min}(A:B)_{\rho} := \max_{\wt{\rho} \in \Bve(\rho)} H_{\min}(A|B)_{\wt{\rho}} \ ,
\end{align}
where $\Bve(\rho) := \{\wt{\rho} : \Tr[\wt{\rho}]\leq 1, P(\rho,\wt{\rho}) \leq \ve \}$ and $P(\cdot,\cdot)$ is the purified distance. The purified distance is a metric on subnormalized states, is contractive for trace-non-increasing (TNI) CP maps, and satisfies the following inequality:
$$ \Delta(\rho,\sigma) \leq P(\rho,\sigma) \leq \sqrt{2\Delta(\rho,\sigma)} \ , $$
where $\rho,\sigma$ are sub-normalized states and $\Delta(\rho,\sigma) = \frac{1}{2}\{\|\rho-\sigma\|_{1} + |\Tr[\rho-\sigma]|\}$\cite{Tomamichel-2015}.

Using these definitions, it has been shown 
\begin{align}\label{eqn:FQAEP}
    \lim_{n\to \infty} \left[\frac{1}{n} H^{\ve}_{\min}(A^{n}|B^{n})_{\rho^{\otimes n}}\right] = H(A|B)_{\rho} \, ,
\end{align}
and the same result for $H_{\max}^{\ve}:= \min_{\wt{\rho} \in \Bve(\rho)} H_{\max}(A|B)$.
This property is known as the fully quantum asymptotic equipartition property (FQAEP, AEP for short or when clear). This is because it allows one to recover asymptotic results in terms of standard entropies from the non-asymptotic setting under the i.i.d. assumption. This in a sense extends the asymptotic equipartition property that gives rise to the typical set used to derive asymptotic results, which is why it is given this name.
The FQAEP was first shown in the partially quantum setting by Renner \cite{Renner2005} and then generalized to the fully quantum case in \cite{Tomamichel-2009a}. Since then, there have been various refinements in terms of second-order corrections and generalizing beyond i.i.d.\ structure. These AEPs have allowed for the unification of one-shot characterizations of many tasks \cite{Tomamichel-2012a,Tomamichel-2013a,Berta-2011a}, and so the AEP property for the min-entropy is of interest in and of itself.

\paragraph{AEP of One-Shot Divergences and Stein's Lemma}
Given the success of one-shot entropies, further smoothed divergence measures were introduced \cite{Datta-2009a,Wang-2012a}. In particular for our purposes, the smoothed max-relative entropy and one-shot hypothesis testing entropies:
\begin{align}
    D^{\ve}_{\max}(\rho||\sigma) := \min_{\wt{\rho} \in \Bve(\rho)} D_{\max}(\wt{\rho}||\sigma) \label{eqn:smoothed-max-rel-ent-defn}
\end{align}
\begin{equation}
\begin{aligned}\label{eqn:hyp-test-defn}
      \exp(-D_{h}^{\ve}(P||Q)) =
        \text{min.}\quad & \ip{P}{X} \\
        \text{s.t.}\quad & \ip{Q}{X} \geq 1 - \ve \\
        & 0 \leq X \leq I \ .
\end{aligned}
\end{equation}

The latter of these two is named the hypothesis testing divergence because when $P$ and $Q$ in \eqref{eqn:hyp-test-defn} are quantum states $\rho$ and $\sigma$, then it determines the minimal  error of the second kind, $\beta_{1}$, under the condition that the error of the first kind is bounded by $\ve$, $\alpha_{1} \leq \ve$, in hypothesis testing. That is, it captures asymmetric one-shot hypothesis testing exactly. Moreover, it may be used to recover the traditional (quantum) Stein's lemma, which we state for completeness.
\begin{proposition}[Stein's Lemma \cite{Chernoff-1952a,Hiai-1991a}]
For hypotheses $\rho^{\otimes n}$ and $\sigma^{\otimes n}$ respectively and tolerated error of the first kind $\ve$, denote the error of the second kind by $\beta^{\ve}_{n}$. It holds that
$$ -\lim_{n \to \infty}\left[ \frac{1}{n} \log(\beta_{n}^{\ve})\right] = D(\rho||\sigma) \ . $$
\end{proposition}
As noted in \cite{Wang-2012a}, we can recover Stein's lemma directly from the operational interpretation of $\exp(-D^{\ve}_{h}(\rho||\sigma))$ along with its AEP:
\begin{align*} \lim_{n \to \infty} \left[ \frac{1}{n} D^{\ve}_{h}(\rho^{\otimes n}||\sigma^{\otimes n}) \right] = D(\rho||\sigma) \quad \forall \ve \in (0,1) \ .
\end{align*}
In fact, a second-order of Stein's lemma follows from a second-order AEP for $D^{\ve}_{h}$ as determined in \cite{Tomamichel-2013a}.

While the smoothed max-divergence does not directly characterize hypothesis testing like the hypothesis testing divergence, it still can characterize Stein's lemma \textit{to first-order}. A simple way of seeing this is that the smooth max divergence's AEP has the same behaviour:
\begin{align}\label{eqn:Dmax-AEP} \lim_{n \to \infty} \left[ \frac{1}{n} D^{\ve}_{\max}(\rho^{\otimes n}||\sigma^{\otimes n}) \right] = D(\rho||\sigma) \quad \forall \ve \in (0,1) \ . 
\end{align}
That $D^{\ve}_{\max}$ does not characterize hypothesis testing to second-order follows from $D^{\ve}_{h} \neq D_{\max}^{\ve}$. 

An alternative way of seeing why $D^{\ve}_{\max}$ characterizes hypothesis testing to first-order follows from being able to bound the hypothesis testing divergence by the smooth max divergence. Specifically, \cite{Tomamichel-2013a} shows that for $0 < \delta < \ve < 1$ and any quantum states $\rho,\sigma$, 
\begin{align*}
 D^{\sqrt{1-\ve+\delta}}_{\max}(\rho^{\otimes n}||\sigma^{\otimes n}) + O(\log(\mrm{poly}(n)))\leq & D^{\ve-\delta}_{h}(\rho^{\otimes n}||\sigma^{\otimes n}) \\
& \hspace{0.25cm}  \leq D^{\sqrt{1-\ve}}_{\max}(\rho^{\otimes n}||\sigma^{\otimes n}) +O(\log(\mrm{poly}(n))) \ .
\end{align*}

We can in fact abstract this above point. Any smoothed-divergence-like quantity $\cD^{\ve}$ whose regularized limit converges to the relative entropy must to first-order characterize hypothesis testing. This is first because it's regularized quantity converges to the same result as the Stein's lemma. Second it is because it implies there are functions that are asymptotically sublinear in the copies of the state $n$ for bounding the distance from $D^{\ve}_{h}(\rho^{\otimes n}||\sigma^{\otimes n})$. We highlight this point as it will give us an information-theoretic operational lens for considering our cone-restricted max-divergence. 
\begin{proposition}\label{prop:stein-first-order-AEP}
Any smoothed divergence-like quantity $\cD^{\ve}(\cdot||\cdot)$ that for all quantum states $\rho,\sigma$ and $\ve \in (0,1)$ satisfies
$$\lim_{n \to \infty} \left[ \frac{1}{n} \cD^{\ve}(\rho^{\otimes n}||\sigma^{\otimes n}) \right] = D(\rho||\sigma)  $$
characterizes hypothesis testing to first-order, i.e. characterizes Stein's Lemma to first order. 
\end{proposition}

In summary, the limits given in Eqns. \eqref{eqn:FQAEP} and \eqref{eqn:Dmax-AEP} are clearly of great import to quantum information theory. One major result of this work is to establish under what settings these limits cease to hold as we modify the measure and how this relates to characterizing information-theoretic tasks and resource theories.

\subsection{Bounded Schmidt Number and Coherence Resource Theories}
The final piece we will need are the specific resources that we will consider throughout this work. In both cases we are predominantly interested in the cone relevant to the resource rather than the set of free states and so we primarily discuss the algebraic structure of the free operators in the resource theory. We refer the reader to \cite{Chitambar-2019a} for further background in resource theories.

\subsubsection{Resource Theory of Bounded Schmidt Number}
Recall that given a pure state $\ket{\psi} \in A \otimes B$, the Schmidt rank of $\ket{\psi}$, denoted $\mrm{SR}(\ket{\psi})$ is the size of the smallest finite alphabet $\Sigma$ such that there exist sets of pure state $\{\ket{\phi}_{i}\}_{i \in \Sigma} \subset A$, $\{\ket{\gamma}_{i}\}_{i\in \Sigma} \subset B$ and probability distribution $p \in \cP(\Sigma)$, such that $\ket{\psi} = \sum_{i=1}^{|\Sigma|} \sqrt{p(i)} \ket{\phi_{i}} \otimes \ket{\gamma_{i}}$. The Schmidt number is the generalization to mixed states. That is, given $\rho \in \Density(A \otimes B)$, the Schmidt number of $\rho$, denoted $\mrm{SN}(\rho)$ is the smallest $k$ such that there exists a finite alphabet $\Sigma$ such that $\rho = \sum_{i=1}^{|\Sigma|} p(i) \op{\psi_{i}}{\psi_{i}}$ where $p \in \cP(\Sigma)$ and $\mrm{SR}(\psi_{i}) \leq k$ for all $i \in \Sigma$. As we will need cones, it would be useful to have the relevant cone which we now define. Following \cite{Watrous-2018}, we refer to this cone as the `Entanglement Rank $r$ Cone'. 
\begin{definition}\label{defn:ent-rank}
An operator $X \in \Pos(A \otimes B)$ is contained in the Entanglement Rank $r$ Cone, $\mrm{Ent}_{r}(A:B)$, if and only if there exists finite alphabet $\Sigma$ and $\{Z_{i}\}_{i=1}^{|\Sigma|} \subset \Lin(A,B)$ such that
$$ X = \sum_{i \in \Sigma} \vec(Z_{i}) \vec(Z_{i})^{\ast} $$
where $\rank(Z_{i}) \leq r$ for all $i \in \Sigma$.
\end{definition}
A few remarks on the definition are in order. First, note that $\mrm{Ent}_{r}(A:B)$ is clearly a cone as this property won't change under scaling by a non-negative value. Second, this is the cone which pertains to density matrices with Schmidt number bounded by $r$. There are many ways to see this, but one is to recall that the rank of a matrix is the tensor rank of the matrix, so by linearity of the vec mapping, (the renormalized) $\vec(Z_{i})$ has Schmidt rank of at most $r$. This observation further implies $\mrm{Ent}_{r}(A:B) \supset \mrm{Ent}_{r-1}(A:B)$ for $r \geq 2$ with extreme cases $\mrm{Ent}_{n}(A:B) = \Pos(A \otimes B)$ where $n = \min\{d_{A},d_{B}\}$ and $\mrm{Ent}_{1} = \Sep(A:B)$ where we recall $X \in \Sep(A:B)$ if and only if $X = \sum_{i=1}^{|\Sigma|} P_{i} \otimes Q_{i}$ where $\{P_{i}\}_{i=1}^{|\Sigma|} \subset \Pos(A)$, $\{Q_{i}\}_{i=1}^{|\Sigma|} \subset \Pos(B)$. Therefore, we can see the entanglement rank cones parameterized by $r$ as interpolating between the separable cone and the positive semidefinite cone in terms of the amount of allowed entanglement. Lastly, we note that $\mrm{Ent}_{r}(A:B)$ is closed as follows. The set of pure states with Schmidt number $\leq r$ is known to be compact \cite{Terhal-2020a} and it is clear that $\Ent_{r}(A:B) \cap \Density$ is the convex hull of the set of pure states with Schmidt number $\leq r$. As the convex hull of a compact set is compact, $\Ent_{r}(A:B) \cap \Density$ is compact. As the cone generated by a compact set not including zero is closed, by Proposition \ref{prop:cone-to-density}, we conclude $\Ent_{r}(A:B)$ is closed.

In terms of a proper resource theory, we only need a few general properties. First recall that the maximally entangled state $\ket{\tau}_{A\wt{A}} := \frac{1}{\sqrt{d}} \sum_{i\in[d]} \ket{ii}_{A\wt{A}}$ is considered the maximal resource in bipartite entanglement theory. We will denote $\tau_{A\wt{A}} \equiv \op{\tau}{\tau}_{A\wt{A}}$. However, if we only care about its dimension, we will denote it $\tau_{d} := \tau_{A\wt{A}}$ where $A \cong \mbb{C}^{d}$. Then we can note the two following properties pertaining to entanglement rank operators, both of which can be found in \cite[Chapter 6]{Watrous-2018}:
\begin{enumerate}
    \item (Distance from $\tau_{d}$) Let $X \in \mrm{Ent}_{r}(A:B)$, then
    \begin{align}\label{eqn:distance-from-max-ent}
        \supp_{\tau_{d}}(X) \leq \frac{r}{d} \Tr[X] \ .
    \end{align}
    Moreover, this holds for any state that is local-unitarily equivalent to $\tau_{d}$ as $(U_{A} \otimes W_{A'})\vec(I_{AA'}) = \vec(WU^{T})$ and $\|WU^{T}\| = \|I_{AA'}\| = 1$.
    \item (Monotonicity Under Separable CP Maps) Given $X \in \mrm{Ent}_{r}(A_{0}:A_{1})$ and a CP map $\Xi = \sum_{i \in \Sigma} \Phi_{i} \otimes \Psi_{i}$ where $\{\Phi_{i}\}_{i \in \Sigma} \subset \mrm{CP}(A_{0},B_{0})$, $\{\Psi_{i}\}_{i \in \Sigma} \subset \mrm{CP}(A_{1},B_{1})$, $\Xi(X) \in \mrm{Ent}_{r}(B_{0},B_{1})$.
\end{enumerate}

Lastly, we recall the following result of Schmidt Number for dynamic resources (channels) \cite{Skowronek-2009a}.
\begin{proposition}\label{prop:ent-rank-channel-property}
Let $\Phi \in \mrm{CP}(A,B)$. The following are equivalent:
\begin{enumerate}
    \item $\Phi$  is $k-$superpositive
    \item $\mrm{SN}((\id_{R} \otimes \Phi)(\rho)) \leq k$ for all $\rho \in \Density(R \otimes A)$
    \item $\mrm{SN}(J_{\Phi}^{AB}) \leq k$
    \item $\Phi$ admits Kraus operators $\{A_{i}\}_{i \in \Sigma}$ such that $\rank(A_{i}) \leq k$ for all $i \in \Sigma$. 
\end{enumerate}
\end{proposition}

\paragraph*{Resource Theory of Positive Partial Transpose}
The cone of positive partial transpose (PPT) operators has a very similar framework that will allow us to derive similar results.
\begin{definition}\label{defn:PPT-operator}
An operator $X \in \PPT(A:B)$ if and only if $X \in \Pos(A \otimes B)$ and $(\id_{A} \otimes \Trans_{B})(X) \in \Pos(A \otimes B)$ where $\Trans_{B}$ is the transpose map on the $\Lin(B)$.
\end{definition}
The PPT operators satisfy the following properties, which may be found in \cite[Chapter 6]{Watrous-2018}:
\begin{enumerate}
    \item (Distance from $\tau_{d}$) Let $X \in \PPT(A:B)$, then
    \begin{align}\label{eqn:PPT-max-ent-supp}
        \supp_{\tau_{d}}(X) \leq d^{-1} \Tr[X]
    \end{align}
    \item (Preservation Under Separable CP Maps) Given $X \in \PPT(A_{0}:A_{1})$ and a CP map $\Xi = \sum_{i \in \Sigma} \Phi_{i} \otimes \Psi_{i}$ where $\{\Phi_{i}\}_{i \in \Sigma} \subset \mrm{CP}(A_{0},B_{0})$, $\{\Psi_{i}\}_{i \in \Sigma} \subset \mrm{CP}(A_{1},B_{1})$, $\Xi(X) \in \PPT(B_{0},B_{1})$.
\end{enumerate}
Moreover, the PPT cone has the relevant equivalence in terms of dynamic resources which are referred to as co-positive maps.
\begin{proposition}\label{prop:PPT-channel-property}
Let $\Phi \in \mrm{CP}(A,B)$. The following are equivalent:
\begin{enumerate}
    \item $J_{\Phi}^{AB} \in \PPT(A:B)$
    \item $(\id_{R} \otimes \Phi)(\rho) \in \PPT(R:B)$ for all $\rho \in \Density(R \otimes A)$.
    \item $T \circ \Phi \in \CP(A,B)$. 
\end{enumerate}
\end{proposition}

\subsubsection{Resource Theory of Coherence}
A resource theory with somewhat similar structure to that of entanglement theory is that of coherence theory. While a very rich theory, we will only need the following results.
\begin{definition}\label{defn:incoherent-cone}
Let $A \cong \mbb{C}^{d}$. An operator $X \in \cI(A)$,i.e.\ incoherent with respect to a preferred orthonormal basis $\{\ket{i}\}_{i \in [d]}$, if $X = \sum_{i \in [d]} \alpha_{i} \op{i}{i}$ where $\alpha_{i} \geq 0$ for all $i \in [d]$.
\end{definition}
We define the maximally coherent state for $A \cong \mbb{C}^{|\Sigma|}$ by $\ket{\varsigma_{d}}= \frac{1}{\sqrt{d}} \sum_{i \in [d]} \ket{i}$. If we wish to refer to the density matrix form, we denote it $\varsigma_{d} \equiv \op{\varsigma_{d}}{\varsigma_{d}}$.

It would be convenient for our purposes to be able to interpolate between the positive semidefinite cone and the fully incoherent cone $\cI$ in the same manner that the set of cones $\Ent_{r}$ do for entanglement. For these purposes, we introduce the following cones.
\begin{definition}\label{def:incoherence-rank}
Let $A \cong \mbb{C}^{d}$. An operator $X \in \cI_{r}(A)$, i.e. have incoherence rank $r$ with respect to a preferred orthonormal basis $\{\ket{i}\}_{i \in [d]}$, if there exists finite alphabet $\Sigma$ such that $\sum_{j \in \Sigma} \alpha_{j} \Pi_{j}$ where $\{\Pi_{j}\}_{j \in \Sigma}$ are incoherent projectors, $\Pi_{j} = \sum_{i \in S_{j}} \dyad{i}$, and $\max_{j} \rank(\Pi_{j}) \leq r$.
\end{definition}
Note these cones are clearly convex and that $\cI(A) = \cI_{1}(A)$, $\cI_{d} = \Pos(A)$, and $\cI_{r}(A) \supset \cI_{r-1}(A)$ for $r > 1$, so we get the interpolation we wanted. It is also easy to see they are closed. consider $\cI_{r} \cap \Density$, which it is straightforward to see is the convex hull of the incoherent projectors excluding the zero matrix. For $A \cong \mbb{C}^{d}$, there are a total of $2^{d}-1$ non-zero incoherent projectors as there are $2^{d}-1$ non-empty subsets of $d$ elements. As this set is finite, its closed. It's also trivially bounded, so it is compact. It follows its convex hull, which is $\cI_{r} \cap \Density$, is compact. As this set does not include the zero matrix, the cone generated by it is closed. Thus $\cI_{r}$ is closed. 

We do note the above set of cones doesn't restrict what subspaces are incoherent. We may conclude the cones restrict the strength of the coherence, not the subspace size.

\begin{proposition}(Distance from $\varsigma_{d}$)\label{prop:distance-from-maximally-coherent}
Let $A \cong \mbb{C}^{d}$, $X \in \cI_{r}(A)$. Then $\supp_{\varsigma_{d}}(X) \leq \frac{r}{d}\Tr[X]$. Likewise for $\cI_{P_{[d]}}$ where the partitioning of $[d]$ is such that $\max_{i}\rank\{\Pi_{i}\} \leq r$.
\end{proposition}
\begin{proof}
Simply using the definitions,
Note 
$$ \langle \varsigma_{d}, \Pi_{j} \rangle = d^{-1} \sum_{i \in S_{j}} |\bra{j}\ket{\varsigma_{d}}|^{2} = d_{1}^{-1}|S_{j}| \ ,$$
from which it follows
\begin{align*}
    \langle \varsigma_{d}, X \rangle = \sum_{j} \alpha_{j} \langle \varsigma_{d}, \Pi_{j} \rangle  \leq d^{-1} \sum_{j} \alpha_{j} |S_{j}|
    \leq \frac{r}{d} \Tr[X] \ ,
\end{align*}
where the bound is by assumption on the bound on the rank of $\Pi_{i}$.
\end{proof}

There are many options for the dynamic resources in coherence theory \cite{Chitambar-2016a}. We will consider the physically incoherent operations. Physically, these are the local CPTP maps $\Phi \in \Channel(A,A)$ that are marginals of channels $\Phi' \in \Channel(A \otimes B, A \otimes B)$ where $\Phi'$ takes products of incoherent states (with respect to their respective local preferred bases) to a state that remains incoherent (with respect to the same preferred bases). We define them by the algebraic structure of their Kraus operators as was determined in \cite{Chitambar-2016a}.
\begin{definition}(\cite[Prop. 1]{Chitambar-2016a})
A quantum channel $\Phi \in \Channel(\mbb{C}^{d},\mbb{C}^{d})$ is a physically incoherent operation (PIO) if all of its Kraus operators are of the form $K_{i} = U_{i}\Pi_{i} = \sum_{x \in [d]} e^{i\theta_{x}}\op{\sigma_{i}(x)}{x}\Pi_{i}$ where $\{\sigma_{i}\}_{i}$ are permutations and $\{\Pi_{i}\}_{i}$ are an orthonormal and complete set of incoherent projectors. That is, $\Pi_{i} = \sum_{j \in S_{i}} \dyad{j}$ where $\sum_{i} S_{i} = [d]$ and $S_{i} \cap S_{i'} = \emptyset$.
\end{definition}

\paragraph*{Resource Theory of (Abelian) Asymmetry}
An alternative approach to interpolating would be to simply think of sequences of cones with bounded subspace sizes. To do so, we can introduce the following definition.
\begin{definition}\label{def:incoherence-partition-rank}
Let $A \cong \mbb{C}^{d}$. Let $\mbb{P}:=\{\Pi_{i}\}_{i}$ be a set of mutually orthogonal projectors that sum to the identity on $\mbb{C}^{d}$. Then $X \in \cI_{\mbb{P}}$ if and only if $[X,\Pi_{i}] = 0$ for all $\Pi_{i} \in \mbb{P}$.
\end{definition}
\noindent $\cI_{\mbb{P}}$ is clearly a convex cone and for the trivial choice $\mbb{P} = \{I_{\mbb{C}^{d}}\}$, one will recover $\Pos(A)$. It's also closed as it is the cone generated by the convex hull of the $\{\Pi_{i} = \sum_{j \in P_{[d]}(i)} \dyad{j}\}$ which is compact.

In comparison to $\cI_{r}$, it has the disadvantage of needing to fix a set of subspaces first, but we can argue that it is in fact more operationally relevant as such cones include the cones that are closed under some symmetry. Recall that a symmetry is captured by a set of unitaries $\{U_{g}\}_{g \in G}$ where $G$ is a group and the state space invariant under $\cG(\rho) = \int U_{g}(\rho)U_{g}^{\ast} \, dg $ is clearly convex. Therefore by Proposition \ref{prop:cone-to-density}, the cone generated by these operators define a cone. Moreover, letting $k$ be the dimension of the commutant of $G$, recall any invariant state may be written as $\sum_{k} \alpha_{k} \Pi_{k}$ where $\Pi_{k}$ are the minimal projectors of the commutant. It follows by letting $\{\Pi_{k}\}_{k}$ be the relevant projectors in Definition \ref{def:incoherence-partition-rank}, one obtains the cone relevant to an Abelian symmetry. While we won't elaborate on the application of such resource theories, we note they are quite often relevant \cite{Marvian-2016a} and so it is important we see our general framework can consider restricting according to such cones.

\section{Cone-Restricted Entropies}
With all of this addressed, we now introduce our framework. At a high level the idea is to use that many one-shot entropies may be expressed as semidefinite programs and then consider sufficiently well-behaved closed convex cones that are subsets of the positive semidefinite cone. Of primary interest for this work will be the max-relative entropy and the min-entropy. We note there have been works that have had a similar approach, but these works are ultimately distinct in both form and what results are proven from them. We elaborate on this once we have introduced the restricted min-entropy.
\subsection{The Restricted Max Relative Entropy}
In the case of $P,Q \in \Pos(A)$ such that $P \ll Q$,  we can express $D_{\max}$ in terms of a semidefinite program:
$$ D_{\max}(P||Q) =  \log \left( \min \{\gamma : P \leq \gamma Q \} \right) \ , $$
where we have just used that $\gamma = 2^{\alpha}$ from the original definition \eqref{eqn:Dmax-defn}. This can be seen as the dual program of a semidefinite program by considering the standard form of \eqref{eqn:conicDual} for the choices $\cK^{\ast} \to \Pos(A)$, $W \to \mbb{R}$, $\phi^{\ast}(\gamma) = \gamma Q$, $y = \gamma$, $a = P$, and $b = 1$.
We can then immediately abstract this idea to alternative cones.
\begin{definition}
For $P,Q \in \Pos(A)$ and $\cK \subseteq \Pos(A)$, the $\cK$-restricted max-relative entropy is
$$ D^{\cK}_{\max}(P||Q) := \log \left(\inf\{\gamma : P \preceq_{\cK^{\ast}} \gamma Q \}\right) \ .$$
\end{definition}
Clearly this recovers $D_{\max}$ for the choice of $\cK = \Pos$. Moreover, we have the following straightforward relation between various cones, which implies $D_{\max}$ achieves the highest value of all choices $\cK \subset \Pos(A)$.
\begin{proposition}
Consider cones $\cK_{1} \subseteq \cK_{2} \subseteq \Pos(A)$. Let $P,Q \in \Pos(A)$. Then $D_{\max}^{\cK_{1}}(P||Q) \leq D_{\max}^{\cK_{2}}(P||Q)$.
\end{proposition}
\begin{proof}
As $\cK_{1} \subseteq \cK_{2}$, $\cK^{\ast}_{1} \supseteq \cK^{\ast}_{2}$. As $P \preceq_{\cK^{\ast}} \gamma Q \Leftrightarrow \gamma Q-P \in \cK^{\ast}$, it follows the optimal $\gamma$ for $D^{\cK_{2}}(P||Q)$ is always feasible for $D^{\cK_{1}}_{\max}(P||Q)$. This completes the proof.
\end{proof}
It also allows us to express it in terms of a conic program when operator dominance is satisfied.
\begin{proposition}
Let $\cK$ be a closed, convex cone. Let $P \ll Q$ where $P,Q \in \Pos(A)$. Then $\exp(D_{\max}^{\cK}(P||Q))$ has the following primal and dual conic programs
\begin{align}
    \emph{Primal:} &\hspace{1mm} \max\{\ip{P}{X}\;:\;\ip{Q}{X} \leq 1 \;\&\; X \in \cK \} \label{eqn:DmaxKPrimal} \\
    \emph{Dual:} &\hspace{1mm} \min\{\gamma: \gamma Q \succeq_{\cK^{\ast}} P \;\&\; \gamma \in \mbb{R} \} \label{eqn:DmaxKDual}
\end{align}
 where we have trivially relaxed the equality in the primal problem.
\end{proposition}
\begin{proof}
That they are maximizations/minimizations follows from the cone being closed and convex. The dual problem follows from the definition of $D^{\cK}_{\max}$. The primal problem follows from noting that if $\phi^{*}(\lambda) := \lambda Q$, then $\phi(M) = \langle M, Q \rangle$.
\end{proof}
It is standard to prove strong duality for conic programs. Indeed, the above programs satisfy this, but it is easiest to see by first showing that $D^{\cK}_{\max}$ will in general satisfy the property of normalization.
\begin{proposition}[Normalization] \label{prop:normalization}
Let $\cK \subseteq \Pos$ be closed and convex. Then 
$$D^{\cK}_{\max}(aP||bQ) = D^{\cK}_{\max}(P||Q) + \log(a) - \log(b) \ . $$
\end{proposition}
\begin{proof}
To see this, just consider the primal problem. The objective function and constraint are both linear. Therefore to obtain the optimum of $D^{\cK}(aP||bQ)$, one can solve for the optimizer of the conic program for $D^{\cK}(P||Q)$, $X^{\star}$. Let the optimal value of the conic program be $\alpha^{*}$. Then by linearity the optimizer of the conic program for $D^{\cK}(aP||bQ)$ is $b^{-1}X^{\star}$ which obtains an optimal value $\frac{a}{b}\alpha^{*}$. As the relative entropy is the $\log$ of the conic program, we have
\begin{align*}
D^{\cK}_{\max}(aP||bQ) = \log(\frac{a}{b}\alpha^{*}) = D^{\cK}_{\max}(P||Q) + \log(a) - \log(b) \ , 
\end{align*}
which completes the proof.
\end{proof}
Using this it is now easy to prove strong duality under the assumption that the identity is contained in the relative interior of the cone. 
\begin{proposition}[Strong Duality]
Let $P,Q \in \Pos(A)$ such that $P \ll Q$. Let $\cK \subseteq \Pos(A)$ such that $I^{A} \in \relint(\cK)$. Let $\alpha, \beta$ be the values obtained by \eqref{eqn:DmaxKPrimal}, \eqref{eqn:DmaxKDual} respectively. Then $\alpha = \beta$, i.e.\ strong duality holds.
\end{proposition}
\begin{proof}
By the normalization property, it suffices to prove strong duality in the case $\hat{P},\hat{Q} \in \Density(A)$ where $\hat{P} = \Tr[P]^{-1}P$ and likewise for $\hat{Q}$. We also assume we have not trivially relaxed the primal constraint. Then the choice $X = I^{A}$ in the primal is strictly feasible by assumption. We know $\Pos \subseteq \cK^{*}$ by definition of dual cone and the assumption $\cK \subseteq \Pos$. As $P \ll Q$, there must exist $\gamma$ such that $\gamma Q - P \in \Pos \subset \cK^{*}$. Thus the primal is strictly feasible and the dual is feasible, so strong duality holds by Slater's condition for strong duality \cite{WatrousATQILecture01}.
\end{proof}
So far in this section we have stayed abstract by talking only in terms of cones. Recall that a primary motivation for focusing on general cones is that a convex resource theory will generate a cone (Proposition \ref{prop:cone-to-density}). It is therefore best to check that the assumption that $I \in \mrm{relint}(\cK)$ is not too demanding for a resource theory. Indeed it is not. While the mathematical definition of relative interior \eqref{eqn:relint-defn} is not physical, very roughly, one may think that $I \in \mrm{relint}(\cK)$ asks that \textit{with respect to the free states} there is a ball of free states about the maximally mixed state contained in the free states.{\footnote{That the ball is with respect to the free states is crucial. For example, there are coherent states arbitrarily close to the maximally mixed state and so there is no ball of coherent operators around the maximally mixed state. However, as shown in Corollary \ref{corr:strong-duality-for-resources}, when restricting the operators to the relevant incoherent cone, there is a ball.}} The relative interior definition captures this containment by, roughly speaking, asking that you can take the \textit{difference} between a free state and the maximally mixed state and scale it so adding this difference to the maximally mixed state is (an unnormalized) free state. In general, you would expect this to be true. The maximally mixed state intuitively is resource-less, and so you would expect there is a small enough scaling such that the difference between it and any free state will perturb the maximally mixed state so that it remains a (unnormalized) free state. We now verify this for the canonical resource theories of entanglement and coherence.
\begin{corollary}\label{corr:strong-duality-for-resources}
For $\cK \in \{\mrm{Ent}_{r \in \mbb{N}}$, $\cI_{r}, \cI_{\mbb{P}}\}$, $D^{\cK}_{\max}$ satisfies strong duality.
\end{corollary}
\begin{proof}
By definition of relative interior \eqref{eqn:relint-defn}, we want $\forall Y \in \cK$, there to exist $\varepsilon > 0$ such that $(1+\ve)I - \ve Y \in \cK$. As $\cK$ is a cone, this is equivalent to $I - \frac{\ve}{1+\ve}Y \in \cK$. It is known that $\I^{AB} + \Delta \in \Sep(A:B)$ for $\|\Delta\|_{2} \leq 1$ where $\Delta \in \Herm(A \otimes B)$\cite{Gurvits-02a}. Therefore, there exists $\varepsilon > 0$ such that $I^{AB} + \frac{\varepsilon}{1+\varepsilon}\Delta \in \Sep(A:B)$ for any $\Delta \in \Herm(A \otimes B)$ by choosing it so $\|\frac{\ve}{1+\ve}\Delta\|_{2} \leq 1$. In particular this holds for $\Delta \in \cK$. This proves $\I \in \relint(\Ent_{r}(A:B))$ as $\Sep(A:B) \subseteq \mrm{Ent}_{r}$. This proves strong duality for $\cK = \mrm{Ent}_{r}$ by the previous proposition.

For the cones pertaining for incoherent operators, let $Y \in \cI$. Then $Y = \sum_{i} \alpha_{i} \dyad{i}$. If $Y = 0$, then for $\ve > 0$, $(1+\ve)I - \ve Y \in \cI$ trivially. Otherwise, let $\ve = (\max_{i}\alpha_{i})^{-1}$, so $(1+\ve)I - \ve Y \in \cI$. So by \eqref{eqn:relint-defn}, $I \in \relint(\cI)$ and by containments this holds for any resource theory of coherence that contains a full basis. This completes the proof.
\end{proof}
\begin{corollary}\label{corr:rel-int-identity}
For $\alpha > 0$, $\alpha I^{AB} \in \relint(\Ent_{r}(A:B))$ and $\alpha I^{A} \in \cI(A)$.
\end{corollary}
\begin{proof}
This follows from noting for $\alpha > 0$, there exists $\varepsilon > 0$ such that $I-\frac{\varepsilon}{\alpha(1+\varepsilon)} \Delta \in \Sep(A:B)$ for $\Delta \in \Herm(A \otimes B)$. The coherent operator result is straightforward.
\end{proof}

Lastly, following \eqref{eqn:smoothed-max-rel-ent-defn}, we define the smoothed restricted max-relative entropy:
\begin{align}\label{eq:smoothed-conic-max-rel-ent} D_{\max}^{\cK,\ve}(\rho||\sigma) = \min_{\wt{\rho} \in \Bve(\rho)} D_{\max}^{\cK}(\wt{\rho}||\sigma) \ .
\end{align}
Note this is the necessary definition as otherwise we do not recover the smooth max-relative entropy when we let $\cK = \Pos$.

\subsection{Restricted Min-Entropy}
We now use the restricted max-relative entropy to derive the restricted min-entropies. We note that all of this had been worked out for the positive semidefinite cone (see e.g.\ \cite{Tomamichel-2015}) and the separable cone \cite{Chitambar-2021a}. The contribution in this section is simply the generalized choice of cone, which has been similarly but distinctly considered in \cite{Jencova-2021a,Gour-2020a,Ji-2021a} as we discuss in the subsequent section.

Following the general rule for inducing entropies \eqref{eqn:rule-for-inducing-entropies}, it follows that for closed convex cone $\cK$ and $P \in \Pos(A \otimes B)$, we would define 
\begin{align}\label{eqn:K-min-entropy-simple}
H_{\min}^{\cK}(A|B)_{P} 
:= - \inf_{\sigma_{B} \in \Density(B)} D_{\max}^{\cK}(P_{AB}||I_{A} \otimes \sigma_{B})
\end{align}
where we have used $\sup -f(x) = - \inf f(x)$. Recalling the definition of $D_{\max}^{\cK}$ \eqref{eqn:Dmax-defn},
\begin{align*}
     & H_{\min}^{\cK}(A|B)_{P} \\
     =& - \log\left( \inf_{\sigma_{B} \in \Density(B)}\inf\{\gamma \in \mbb{R} : P \preceq_{\cK^{\ast}} \gamma I_{A} \otimes \sigma_{B} \} \right) \ .
\end{align*}
Noting that any feasible $\gamma \sigma_{B} \in \Pos(B)$ for $P \in \Pos(A \otimes B)$, we can combine the two optimizations and simplify to the following conic programs with strong duality under roughly the same condition as before.

\begin{proposition}
Let $P \in \Pos(A \otimes B)$. Consider closed, convex cone $\cK \subset \Pos(A \otimes B)$ such that $d^{-1}_{A}\I_{AB} \in \mrm{\relint}(\cK)$. Then $H^{\cK}_{\min}(A|B)_{\rho} = -\log(\alpha)$ where $\alpha$ is obtained by both the primal and dual:
\begin{center}
    \begin{miniproblem}{0.45}
      \emph{Primal}\\[-5mm]
      \begin{equation}
      \begin{aligned}\label{eqn:HminKPrimal}
        \text{maximize:}\quad & \ip{P}{X} \\
        \text{subject to:}\quad & \Tr_{A}{X} = I_{B} \\
        & X \in \cK
      \end{aligned}
      \end{equation}
    \end{miniproblem}
    \begin{miniproblem}{0.45}
      \emph{Dual}\\[-5mm]
      \begin{equation}
      \begin{aligned}\label{eqn:HminKDual}
        \text{minimize:}\quad & \Tr[Y]  \\
        \text{subject to:}\quad & I_{A} \otimes Y \succeq_{\cK^{\ast}} P \\
        & Y \in \Herm(B) \ .
      \end{aligned}
      \end{equation}
    \end{miniproblem}
 \end{center}
\end{proposition}
\begin{proof}
We get the dual from using $Y:= \gamma \sigma_{B} \in \Pos(B)$ and $\Tr[Y] = \gamma$. The primal then follows using the well known fact that the map $\phi^{\ast}(Y) = \I_{A} \otimes Y$ has the adjoint map $\Tr_{A}$. Strong duality then follows because by assumption $d_{A}^{-1}I_{AB} \in \relint(\cK)$ and it's feasible as $\Tr_{A}[d_{A}^{-1}I^{AB}] = I_{B}$. The dual is always feasible because $P \in \Pos(A \otimes B)$, there exists $\gamma \geq 0$ such that $I_{A} \otimes \gamma I_{B} \geq P$. So by Slater's criterion for strong duality, strong duality holds.
\end{proof}
\begin{corollary}
By Corollary \ref{corr:rel-int-identity} and the previous proposition, for $\cK = \mrm{Ent}_{r}(A:B)$, $H_{\min}^{\cK}(A:B)$ satisfies strong duality.
\end{corollary}

The final piece of the framework we will need is an equivalent version of the primal problem \eqref{eqn:HminKPrimal}. This was first proven in \cite{Konig-2009a} in the case $\cK = \Pos$. It was noted you could do this for other cones in \cite{Chitambar-2021a}. Our presentation follows that of \cite{Tomamichel-2015} closely.
\begin{proposition}\label{prop:min-entropy-singlet-fraction}
The conic program \eqref{eqn:HminKPrimal} for a state $P \in \Pos(A \otimes B)$ is equivalent to 
\begin{align}\label{eq:min-entropy-singlet-fraction}
d_{A} \max_{\substack{\Psi \in \mrm{CPTP}(B,A'): \\ J_{\Psi^{\ast}} \in \cK}} \langle (\id_{A} \otimes \Psi)(P), \tau^{AA'} \rangle 
\end{align}
\end{proposition}
\begin{proof}
Consider the primal problem \eqref{eqn:HminKPrimal}. Note that, $\Tr_{A}(X) = \I_{A}$ implies that $X = J_{\Phi}$ for some unital CP map $\Phi:A' \to B$. Thus, we have the optimal value is given by 
\begin{align*}
    \underset{\Phi \in \mrm{CPU}(A,B):\;J_{\Phi} \in \cK }{\max} \langle P, J_{\Phi} \rangle 
    = &\max_{\Phi \in \mrm{CPU}(A,B):\; J_{\Phi} \in \cK} d_{A} \langle (\id_{A} \otimes \Phi^{\ast})(P), \tau_{A\wt{A}} \rangle \\
    = &d_{A} \max_{\Psi \in \Channel(B,A):\; J_{\Psi^{\ast}} \in \cK} \langle (\id_{A} \otimes \Psi)(P), \tau_{A\wt{A}} \rangle \ ,
\end{align*}
where the second equality is by definition of Choi operator and re-normalizing the maximally entangled state, and the third equality is using that the adjoint of a CP unital (CPU) map is a CPTP map.
\end{proof}
Lastly, following \eqref{eq:smoothed-min-ent}, we define the smoothed restricted min-entropy:
\begin{align}\label{eq:smoothed-conic-min-entropy}
    H_{\min}^{\cK,\ve}(A|B) = \max_{\wt{\rho} \in \Bve(\rho)} H_{\min}(A|B)_{\wt{\rho}} \ .
\end{align}
Note this must be the definition if we are to recover the standard smooth min-entropy when $\cK = \Pos$.

\subsection{Other Frameworks and the Support Function}\label{sec:other-frameworks}
With the basic framework introduced, it is worth mentioning how this relates to other works as it will in particular clarify a technical distinction running throughout this work. We begin by discussing \cite{Jencova-2021a}, which, while presented in category-theoretic terms, investigates min-entropic quantities restricted to various cones. Interestingly, while the basic idea is quite similar, the approach obtains distinctly different results. In effect, the crucial difference in their framework is that the way they construct the min-entropy results in always considering $H_{\min}^{\cK}(B|A)$ for $\rho_{AB}$. One way to see this distinction could be important is the following. For the case that $\cK$ such that $J_{\Phi} \in \cK$ if and only if $J_{\Phi^{\ast}} \in \cK^{\ast}$, one can use Proposition \ref{prop:min-entropy-singlet-fraction} to generalize \cite[Eqn. 37]{Gour-2019a} to show
\begin{equation}\label{eqn:min-ent-cptp-support}
\supp_{\Channel^{\cK}(A,B)}(P) = \exp(-H_{\min}^{\cK}(B|A)_{P}) \ ,
\end{equation}
where $\Channel^{\cK}(A,B) := \{J_{\Lambda} \in \cK : \Lambda \in \Channel(A,B)\}$ (see also \cite{Gour-2020a}).
In contrast, under the same assumption on $\cK$, what Proposition \ref{prop:min-entropy-singlet-fraction} tells us is that
\begin{equation}\label{eqn:min-ent-unital-support}
 \supp_{\CPU^{\cK}(A,B)}(P) = \exp(-H_{\min}^{\cK}(A|B)_{P}) \ ,
\end{equation}
where $\CPU^{\cK} := \{J_{\Lambda} \in \cK : \Lambda \in \mrm{CPU}(A,B)\}$. 
One would not in general suspect the support on the space of unital maps and quantum channels to be the same or capture the same operational meaning. We believe this is why our results are distinct. Beyond this technical point, in effect due to the generality of the category theory, the authors did not address the transition beyond the one-shot setting, which is a primary focus of this work. In doing so, we partially resolve a suggested future direction of \cite{Jencova-2021a}.

Similarly, in \cite{Ji-2021a}, the authors do introduce the restricted min-entropy for convex cones, but again the ordering is flipped with respect to this work, and they don't consider moving from one-shot to asymptotic results. Lastly, in \cite{Matthews-2014a}, the authors considered one-shot hypothesis testing \eqref{eqn:hyp-test-defn} where the optimizer is restricted to alternative sets. While this is a similar methodology, the present work only briefly considers such a measure. 

\section{The Positive Cone Is Necessary for Fully Quantum Information Theory}\label{sec:anti-stein-lemma}
With the general framework constructed, we now use this to show that the ability to recover asymptotic results from one-shot entropies requires the L\"{o}wner order, i.e.\ the positive cone. The intuition for this is as follows:
\begin{enumerate}
\item Under specific conditions, $D^{\cK}_{\max}$ and $H^{\cK}_{\min}$ behave as support functions for the set density matrices in $\cK$.
\item For a pure state, its support must change from one when it is in the (closed) cone to strictly less than one when it is not.
\item Combining these points, we expect that for the choice of the right pure state $D_{\max}^{\cK},H^{\cK}_{\min}$ can behave sufficiently different on the same state by varying the cone.
\end{enumerate}
This intuition is depicted for $D_{\max}^{\cK}$ in Fig. \ref{fig:Geometric-Interpretation-of-Dmax}. 

Furthermore, note that this intuition also is relevant to resource theories. Specifically, since a convex resource theory has a set of free states $\cF$ which induce a cone $\cK_{\cF}$, we could consider the $\cK_{\cF}$-max relative entropy. Then this is saying under certain conditions we expect $D_{\max}$ and $D^{\cK_{\cF}}_{\max}$ to behave the same on free states but differ for resourceful states. We further this connection in subsequent sections.

\begin{figure}[H]
    \centering
    \includegraphics[width=0.7\columnwidth]{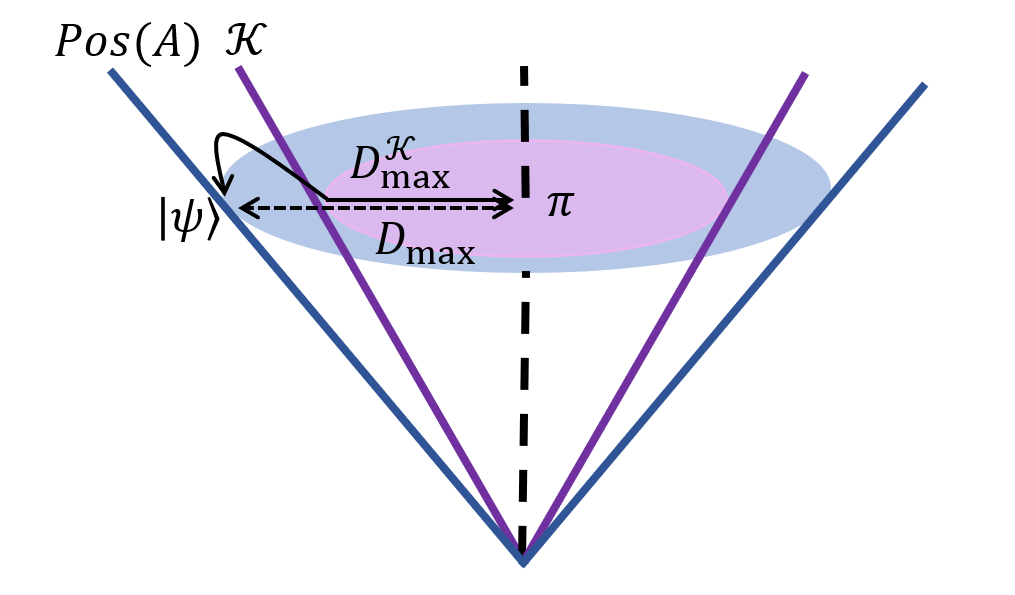}
    \caption{Geometric intuition of the restricted max-relative entropy measuring the distance from the maximally mixed state to a pure state. Viewing the path distance as representing the value of the measure, we see that as long as $\rho \in \cK$, $D_{\max}^{\cK}(\rho||\pi) = D_{\max}(\rho||\pi)$ as they are both flat lines. However, if $\rho \not \in \cK$, the path distance outside of the cone is altered. The black dotted line represents that $\alpha\I \in \relint(\cK)$ for all $\alpha > 0$.}
    \label{fig:Geometric-Interpretation-of-Dmax}
\end{figure}

\subsection{The Anti-Standard Max-Divergence AEP}
In this subsection we show sufficient conditions for breaking that the AEP for the cone-restricted max-divergence converges to the relative entropy. This allows us to give sufficient conditions for when the measure cannot characterize hypothesis testing of quantum states even to first order given Proposition \ref{prop:stein-first-order-AEP}. In doing so, we reveal necessary conditions for fully quantum hypothesis testing. At the same time, we can instead think of the intuition as arising from resource theories. Specifically, what will be crucial in this section will be that the support of pure states within the cone is bounded from a target state outside of the cone. Interpreting the cone as induced by free states, we can think of this as the requirement that the free states are sufficiently bounded away from a maximally resourceful state under the inner product. We can then think of the intuition as being that when the free states are sufficiently bounded away from a maximal resource, the induced measure can't behave like the un-restricted measure i.e.\ the max-relative entropy $D_{\max}$. This in turn means the fundamental limits of fully quantum information theory cannot be achieved in the restricted theory. We begin by capturing these ideas in the one-shot zero-error setting.

\begin{proposition}\label{prop:Dkmax-pure-state-outside-cone}
    Let $\cK \subsetneq \Pos(A)$ be a closed, convex cone such that $\I \in \relint(\cK)$. Let $\psi \in \Density(A)$ be a pure state (i.e.\ $\Tr[\psi^{2}] = 1$) and $\psi \not \in \cK$. It follows
    $$ D^{\cK}_{\max}(\psi || \pi) < D(\psi||\pi) = D_{\max}(\psi||\pi) \ , $$
    where $\pi := d_{A}^{-1}I_{A}$.
\end{proposition}
\begin{proof}
    Then 
    $$D(\psi||\pi) = \Tr[\psi \log \psi] - \Tr[\psi \log \pi] = \log(d_{A}) \ . $$
    As $\I \in \relint(\cK)$, we can use the primal problem \eqref{eqn:DmaxKPrimal}:
    \begin{align*}
        D^{\cK}_{\max}(\psi||\pi)
         =& \log(\max\{ \bra{\psi}X\ket{\psi} : \Tr[X] \leq d_{A} \, , \, X \in \cK \}) \\
        =&\log(d_{A} \max\{ \bra{\psi}\widehat{X}\ket{\psi} : \widehat{X} \in \Density \cap \cK \}) \\
        =& \log(d_{A}) + \log(\max\{\bra{\psi}\widehat{X}\ket{\psi} :  \widehat{X} \in \Density \cap \cK \}) \\
        <& \log(d_{A}) \ ,
    \end{align*}
    where the first equality is using the conic program \eqref{eqn:DmaxKPrimal} with the observations $\Tr[\psi X] = \bra{\psi}X\ket{\psi}$ and $\Tr[\pi X] = d_{A}^{-1}\Tr[X]$, the second is because the optimal $X^{\star}$ has trace of $d_{A}$ since $\cK$ is a cone, and if $X^{\star}$ is the optimizer, then $X^{\star} = d_{A} \widehat{X}$ where $\widehat{X} \in \Density \cap \cK$, the third equality is a standard property of logarithms, and the strict inequality is because, by assumption $\psi \not \in \cK$ and is a closed cone, so $\max\{\bra{\psi}\widehat{X}\ket{\psi} :  \widehat{X} \in \Density \cap \cK \} < 1$ and the logarithm must be strictly negative. In the case $\cK = \Pos$, then the above derivation shows $D_{\max}(\psi||\pi) = \log(d_{A})$, so this completes the proof.
\end{proof}
First, we again stress that if we think of $\cK$ as being induced by free states $\cF$, the above result shows the measure of a resourceful state from the maximally mixed (resource-less) state will differ between $D^{\cK}_{\max}$ and $D_{\max}$. Moreover, this is because the resourceful state's overlap with any free state is necessarily lower than overlap over the unrestricted set of density matrices, which contains the resourceful state itself. Furthermore, note that the above result shows that for pure states outside of the cone $\cK$, not only is $D^{\cK}_{\max} < D_{\max}$, but it is less than the standard relative entropy $D$. This is a strong subversion as $D < D_{\max}$ always holds, and this property may be viewed as the reason why the measure cannot characterize hypothesis testing even to first-order. Moreover, we note a standard way to prove that $D \leq D_{\max}$ is using that $\log$ is an operator monotone. It follows that the previous proposition implies, as we might expect, $\log$ does not have an equivalent property under any further restricted conic ordering. We can in fact further show that for restricted cones, we cannot satisfy the canonical requirements of a divergence.

\begin{corollary}\label{corr:DPI-breaks}
Let $\cK \subsetneq \Pos(A)$ be a closed, convex cone such that $\I \in \relint(\cK)$.  
Then $D^{\cK}_{\max}$ neither satisfies isometric invariance nor contractivity under CPTP maps (data-processing). That is, $D^{\cK}_{\max}$ is not a quantum R\'{e}nyi entropy.
\end{corollary}
\begin{proof}
As $\cK \subsetneq \Pos(A)$ is closed and convex, there must be a pure state $\psi \not \in \cK$ as cone of the convex hull of all pure states generates $\Pos(A)$. Therefore, let $\phi \in \Density(A)$ be pure and $\psi \not \in \cK$. Let $U$ be any unitary such that $U:\ket{\psi} \to \ket{\phi}$, where such a unitary will always exist.
\begin{align*} 
D_{\max}^{\cK}(\psi||I)<2\log(d) =& D^{\cK}_{\max}(\phi||I) 
= D(U\psi U^{\ast}||UIU^{\ast}) \ , 
\end{align*}
where the first inequality is by the previous proposition.
\end{proof}
Given these two previous points, we might now ask if the max-divergence AEP \eqref{eqn:Dmax-AEP} can hold for alternative cones, where, following \eqref{eqn:smoothed-max-rel-ent-defn}, $D_{\max}^{\ve,\cK}(\rho||\sigma) = \min_{\wt{\rho} \in \Bve(\rho)} D_{\max}^{\cK}(\wt{\rho}||\sigma)$. Here we first show this in our most general mathematical form where all we are concerned about is the support of the cone on some target state. The advantage of this generality is that it will allow us to show how strongly dependent Stein's lemma is on the positive cone by considering the measure. However, the general mathematical result again has meaning in terms of resource theories. In a resource theory you have some property that defines your free states $\cF$ on some Hilbert space $A$. Generally this property is well-defined on $n-$fold copies of the space. To be a physically well-defined theory, you would want that no matter how many copies you have, the free states are bounded away from the maximal resourceful state. Then taking the resource-theoretic interpretation of the cone as induced by the free states, and viewing the aforementioned `target state' as a maximally resourceful state, the following theorem says if your resource theory is stable under copies, then the AEP does not converge to the relative entropy and thus cannot characterize hypothesis testing given Stein's lemma.

\begin{theorem}\label{thm:anti-stein-lemma-support-version}
   Consider a sequence of closed convex cones $\cK^{(n)} \subset \Pos(A^{\otimes n})$. Let $\psi^{\otimes n} \not \in \cK^{(n)}$ for all $n \geq n_{0}$ for some $n_{0} \in \mbb{N}$, and $\I^{A} \in \relint(\cK^{(n')})$ for all $n' \geq n_{0}'$. Then for all $\varepsilon \in [0,1]$, $$\lim_{n \to \infty} \left \{ \frac{1}{n} D^{\ve,\cK}_{\max}(\psi^{\otimes n}||\pi^{\otimes n}) \right \}~<~D(\psi||\pi) \ $$
   if and only if $\supp_{\cK^{(n)} \cap \Density(A^{\otimes n})}(\psi^{\otimes n}) \leq f(n) = O(c^{n})$ where $c < 1$.
\end{theorem}
\begin{proof}
     We first prove that the support condition implies the traditional AEP for max-divergence breaks down for these states. We will show $\lim_{n \to \infty} \frac{1}{n} D^{\cK^{(n)}}_{\max}(\psi^{\otimes n}||\pi^{\otimes n}) < D(\psi||\pi)$. Since smoothing the conic max-relative entropy only decreases it \eqref{eq:smoothed-conic-max-rel-ent}, it suffices to prove for $\varepsilon = 0$ as it would imply it for any larger $\varepsilon$. Then we consider
    \begin{align*}
        & D^{\cK^{(n)}}_{\max}(\psi^{\otimes n}||\pi^{\otimes n})\\
        =& \log(\max\{ \bra{\psi}X\ket{\psi} : \Tr[X] \leq d^{n}_{A} \, , \, X \in \cK \}) \\
        =& \log(d^{n}_{A} \max\{ \bra{\psi^{\otimes n}}\widehat{X}\ket{\psi^{\otimes n}} : \widehat{X} \in \Density \cap \cK^{(n)} \}) \\
        =& n\log(d_{A}) + \log(\max\{\bra{\psi^{\otimes n}}\widehat{X}\ket{\psi^{\otimes n}} :  \widehat{X} \in \Density \cap \cK^{(n)} \}) \\
        =& n \log(d_{A}) + \log(f(n)) \ .
    \end{align*}
    Therefore,
    \begin{align*}
         \lim_{n \to \infty} \left[ \frac{1}{n}D^{\cK^{(n)}}_{\max}(\psi^{\otimes n}) \right]
        = \log(d_{A}) + \lim_{n \to \infty} \frac{\log(f(n))}{n} 
        = & \log(d_{A}) + \log(c)\\
        <& \log(d_{A}) = D(\psi||\pi) \ , 
    \end{align*}
    where we used the assumption $f(n) = O(c^{n})$ where $c < 1$. Note that smoothing only aggravates this problem.
    This completes the first direction.
    
    To prove the second direction, we prove the contrapositive. If $f(n) = 2^{o(n)}$, i.e.\ is sub-exponential, then $\lim_{n \to \infty} 1/n \log f(n) \to 0$, so the AEP between $\psi$ and $\pi$ would hold in this setting. Note that smoothing would preserve the $f(n)$ bound. 
    
    As $f(n) \leq 1$ for all $n$, the only option yet considered is $\lim_{n \to \infty} f(n) = 1$, in which case $\lim_{n \to \infty} 1/n \log f(n) = 0$ and the AEP holds. In the case of smoothing, if this limiting behaviour holds, the AEP holds. This completes the proof.
\end{proof}
Before applying Theorem \ref{thm:anti-stein-lemma-support-version} to specific resource theories, we note how general the demand on the overlap is. In particular, it more than captures all quantum resource theories that are \textit{extensive} in the sense that resourcefulness scales in the number of copies \cite{Vijayan-2020a}. Formally \cite[Property 1 rewritten]{Vijayan-2020a}, a resource theory is \textit{extensive} if for every pure state $\varphi \in \Density(A)$, there exists a non-negative constant $\kappa(\varphi)$ such that $\sup_{\gamma \in \cF} \langle \varphi^{\otimes m}, \gamma \rangle \leq \exp(-m \kappa(\varphi))$, where $\cF$ is the set of free states. Effectively this means either the state $\varphi$ is free, $\kappa(\varphi) = 0$, or it is exponentially quickly bounded away from the free states in the number of copies. In the language of Theorem \ref{thm:anti-stein-lemma-support-version}, an extensive resource theory is such that $\cK^{(n)} := \mrm{cone}(\cF(\cH^{\otimes n}))$ is such that $f(n)$ goes to zero at least exponentially fast for any pure state not in $\cF$. It is worth noting that the resource theories of entanglement and coherence are extensive, and thus some of the power in applying Theorem \ref{thm:anti-stein-lemma-support-version} is that it applies to general sequences of cones to handle cases beyond extensive quantum resource theories.

Having established Theorem \ref{thm:anti-stein-lemma-support-version} and addressed its generality, we now apply it to specific resource theories. This first captures a notion of the necessary cone for a max-divergence-like quantity to characterize the fully quantum Stein's lemma. It also makes it clear that the the above result applies to resource theories whose free states are properly bounded away from their maximally resourceful state.

We now show that the Stein's lemma is only characterized to first-order by cone-restricted max-divergences where $\cK$ asymptotically approximate the fully coherent cone (i.e. $\Pos(A^{\otimes n})$), or, in the bipartite setting, the cone with full Entanglement rank (i.e. $\Pos((A \otimes B)^{\otimes n})$). 
\begin{corollary}\label{corr:coherence-Stein's-lemma}
Fix a basis of Hilbert space $A$. Consider any sequence of cones $\{\cK^{(n)}(A^{\otimes n})\}$ such that there exists $n_{0} \in \mbb{N}$ such that for all $n \geq n_{0}$, $\cK^{(n)} \subseteq \cI_{f(n)}(A^{\otimes n})$ where $\cI_{f(n)}$ is defined using the same fixed basis of $A$. If there exists $n_{1} \in \mbb{N}$ and $\delta \in [0,1)$ such that for all $n \geq n_{1}$, $f(n) \leq  (\delta d_{A})^{n}$, then the sequence of measures cannot characterize Stein's lemma to first-order. In other words, in the unipartite setting, the Quantum Stein lemma  only is characterized by $D^{\cK_{n}}_{\max}$ where the sequence of cones asymptotically contains the $n-$fold copy of all maximally coherent states on $A$, i.e.\ asymptotically approximates the $n-$fold  positive semidefinite cone on $A$.
\end{corollary}
\begin{proof}
We focus on the set of cones $\cI_{r}$, the same idea holds for $\cI_{\mbb{P}}$ with appropriate choice of partition. Consider the sequence of cones $\{\cK^{(n)} = \cI_{f(n)}(A^{\otimes n})\}$ for some fixed basis of $A$. Assume there is a fixed maximally coherent state $\varsigma_{d}$ such that $\varsigma_{d} \not \in \cK^{(n)}$ for all $n \geq n' \in \mbb{N}$. By Proposition \ref{prop:distance-from-maximally-coherent}, 
\begin{align*}
\supp_{\cI_{f(n)}\cap\Density(A^{\otimes n})} \langle \varsigma_{d}^{\otimes n} , \sigma_{K} \rangle \leq \frac{f(n)}{d^{n}_{A}} \ . 
\end{align*}
As long as $\limsup f(n)/d_{A}^{n} := c<1$, we satisfy the requirements of Theorem \ref{thm:anti-stein-lemma-support-version}. Note that for any $\delta < 1$, $\lfloor(\delta d_{A})^{n}\rfloor/d_{A}^{n} < 1$ for all $n$. This shows you can in effect exponentially increase any non-trivial fraction of the total initial space and allow for that much coherence, without the regularized quantity converging to the relative entropy. Thus, the sequence of cone-restricted max-divergences converge to the relative entropy, unless you had considered the entire space, i.e.\ the positive cone. It follows by Proposition \ref{prop:stein-first-order-AEP}, that Stein's lemma is not characterized even to first-order by these measures unless the entire initial space was considered.

This shows that if a sequence of cones $\{\cK^{(n)}\}_{n}$ can be contained by incoherent cones that don't sufficiently contain their maximally coherent state, Stein's lemma cannot be characterized  by $\{\cK^{(n)}\}$. As this argument holds for any choice of basis for $A$, $\{\cK^{(n)}\}_{n}$ can only satisfy the Stein's lemma if for all choices of bases of $A$, the corresponding $n-$fold maximally coherent state is contained. This completes the proof.
\end{proof}

We now prove the equivalent for the entanglement rank cones.

\begin{corollary}
Without loss of generality, assume $d_{A} \geq d_{B}$. Consider any sequence of cones $\{\cK^{(n)}(A^{n}:B^{n})\}$ such that there exists $n_{0} \in \mbb{N}$ such that for all $n \geq n_{0}$, $\cK^{(n)} \subseteq \mrm{Ent}_{f(n)}$. If there exists $n_{1} \in \mbb{N}$ and $\delta \in [0,1)$ such that for all $n \geq n_{1}$,such that $f(n) < (\delta d_{A})^{n}$, then the regularized sequence does not converge to the relative entropy. In particular, this means it cannot characterize hypothesis testing to first order. That is, the Quantum Stein lemma cannot be characterized by a sequence of cone-restricted max-divergence measures unless they asymptotically contain all $n$-fold maximally entangled states.
\end{corollary}
\begin{proof}
For clarity, we assume $d_{A} \leq d_{B}$, but the same proof holds in the other case using $d_{B}$ instead. Consider the sequence of cones $\{\cK^{(n)} = \mrm{Ent}_{f(n)}(A^{\otimes n}: B^{\otimes n})\}$. Then by \eqref{eqn:distance-from-max-ent}, 
$$\supp_{\sigma_{K} \in \cK^{(n)} \cap \Density((A\otimes B)^{\otimes n})}(\tau_{d_{A}}) \leq \frac{f(n)}{d^{n}_{A}} \ . $$
As long as $\limsup f(n)/d_{A}^{n} := c<1$, the requirements of Theorem \ref{thm:anti-stein-lemma-support-version} are satisfied. Note that for any $\delta < 1$, $\lfloor(\delta d_{A})^{n}\rfloor/d_{A}^{n} < 1$ for all $n$. Therefore any exponential growth in entanglement that is a strict fraction of the original $d_{A}$ will break the convergence to the relative entropy. Thus by Proposition \ref{prop:stein-first-order-AEP}, it cannot characterize Stein's lemma. By our choice of using the entanglement rank cones, under these conditions it cannot contain any of $n$-fold copies of the class of maximally entangled states in $\Density(A \otimes B)$. Therefore, the cone defining the measure would need to asymptotically contain all the $n-$fold maximally entangled states to characterize Quantum Stein's lemma to first-order. This completes the proof.
\end{proof}

Note the entanglement rank cones have been chosen as they are physically interesting and capture the rate at which we increase entanglement, but they treat all maximally entangled states on equal footing. We can in fact construct a similar result to prioritize an arbitrary choice of maximally entangled.
\begin{corollary}
A sequence of cones $\{\cK^{(n)}(A^{n}:B^{n})\}$ only characterizes the Quantum Stein lemma to first-order if asymptotically it contains all states that are local-unitary equivalent to the canonical maximally entangled state.
\end{corollary}
\begin{proof}
Fix a state that is local-unitarily equivalent to the maximally entangled state and denote it 
$$\wt{\tau}_{d} := (U^{A} \otimes W^{B}) \tau_{d} (U^{A} \otimes W^{B})^{\ast} \ .$$ 
For any $n \in \mbb{N}$ define the set of density matrices $\wt{\cF}_{n} := \{\rho \in \Density(A \otimes B) : \langle \rho , \wt{\tau}_{d}^{\otimes n} \rangle \leq 1 - \delta_{n} \}$. This set is clearly convex for any $n$ by construction. We can think of a function $g_{n} := \langle \cdot, \wt{\tau}_{d}^{\otimes n} \rangle$ which is continuous for every $n$. As $\wt{\tau}_{d}$ and every density matrix is positive, $\wt{\cF}_{n}$ is the pre-image of the continuous function $g_{n}$ over the closed interval $[0,1-\delta_{n}]$ and thus $\wt{\cF}_{n}$ is closed. It's also trivially bounded, so $\wt{\cF}_{n}$ is a convex, compact set for each $n$. Define the closed convex cone $\wt{\cK}^{(n)} = \mrm{cone}(\wt{\cF}_{n})$, where that it is closed follows from the fact $\wt{\cF}_{n}$ is compact and does not include $0$. Moreover, $\cI$ is in it's relative interior so long as $\delta_{n} \leq 1 - d^{-n}$ for all $n$ as that is small enough that $\cK^{(n)} \supseteq \Ent_{1}(A:B)$ by \eqref{eqn:distance-from-max-ent}. This is because an equivalent definition of the relative interior is
$$\relint(\cK) := \{x \in \cK : \exists \ve > 0 : \cN_{\ve}(x) \cap \mrm{Aff}(\cK) \subset \cK \} \ , $$
we know there is an epsilon ball of separable states around the identity \cite{Gurvits-02a}, and we have $\cK^{(n)}$ contains the separable states, so the appropriate choice of epsilon results in the intersection being contained in $\cK^{(n)}$. Thus we have cones that satisfy constraints to apply Theorem \ref{thm:anti-stein-lemma-support-version}. Therefore as long as $\limsup_{n} \delta_{n} < 1$, the regularized quantity does not converge to the relative entropy and so by Proposition \ref{prop:stein-first-order-AEP}, Quantum Stein's lemma cannot be characterized by this sequence of cones.

Now first note this only considered a single choice of maximally entangled state, so it's much stronger than the previous entanglement rank result as, for example, the set will include the rest of the maximally entangled state basis under local unitaries $(U \otimes W)$ as they are mutually orthogonal. Finally, the above shows that if a sequence of cones $\{\cK^{(n)}\}_{n}$ can be contained by $\{\wt{\cK}^{(n)}\}_{n}$ for a choice of $\wt{\tau}_{d}$ and appropriate choices of $\delta_{n}$, then the Stein's lemma cannot be characterized by this sequence of measures. It follows for Stein's lemma to hold, that must not be the case for any choice of $\wt{\tau}_{d}$ which completes the proof.
\end{proof}

\subsection{Anti-Standard Min-Entropy AEP}
In the previous section we showed general conditions such that the convergence of the regularized cone-restricted max-divergence is not the relative entropy and thus cannot characterize hypothesis testing even to first-order. In particular, recalling the normalization property (Proposition \ref{prop:normalization}), we have shown this by showing the regularized $D_{\max}^{\cK}(\rho||\I)$ does not converge to $D(\rho||\I)$ when $\rho$ is maximally coherent and the cone is $\cI_{r}$.
Noting that $H(A)_{\rho} = -D(\rho||\I)$ and $H_{\min}(A)_{\rho} = -\log(\supp_{\Density(A)}(\rho))$ \cite{Gour-2019a}, one might suspect that we have inadvertently also shown that $H_{\min}^{\cK}(A)$ does not satisfy an AEP in a general setting. Indeed, this follows from \eqref{eqn:K-min-entropy-simple} where one lets $B = \mbb{C}$ so that the cone $\cK(A \otimes B)$ is really $\cK(A)$. The focus of this section is showing this to be the case of the restricted min-entropy in the general conditional case. While this may seem uninteresting as we have broken the unconditional case, the conditional case is fundamental in terms of entanglement theory, whereas in the unconditional case the coherence resource seems the fundamental choice. Following the same presentation as the previous section, we first show how the one-shot case breaks down for general cones, then we present sufficient conditions for an AEP to breakdown.

We note as explained in Section \ref{sec:other-frameworks}, there is an asymmetry for min-entropy that is captured in difference of the type of support function. Formally, by the same proof as for Proposition \ref{prop:min-entropy-singlet-fraction}, we have for $P \in \Pos(A \otimes B)$,
\begin{align*}
\exp(-H_{\min}^{\cK}(B|A)_{P})
 =d_{B} \underset{\Psi^{\ast} \in \CPU(B,A): J_{\Psi} \in \cK}{\max} \langle (\Psi \otimes \id_{B})(P), \tau_{d_{B}} \rangle \ ,
\end{align*}
which is less natural to work with. However, we can use the support function definition of $H_{\min}^{\cK}(B|A)$ directly, to see how the AEP function behaves. This partially resolves an open direction noted in \cite{Jencova-2021a}, though the cones we focus on are more specific.

We first present explicit examples of where there exists a gap between $H_{\min}(A|B)$,$H(A|B)$ and $H^{\cK}_{\min}$ in the same fashion as for $D_{\max}^{\cK}$. This will hold under the rather minimal assumption that $\cK$ does not contain all pure states $(I^{A} \otimes U)\ket{\tau_{d}}$ where $U$ is a unitary. It also makes clear the general idea of when the AEP would fail.
\begin{proposition}
Let $\cK \subsetneq \Pos(A \otimes B)$ be closed and convex and that $d_{A}^{-1}\I^{AB} \in \relint(\cK)$. Let $\psi:=(\I^{A} \otimes U)\tau_{d}(\I^{A} \otimes U^{\ast}) \not \in \cK \subset \Pos(A:B)$ where $U$ is a unitary. Then there exists $\rho$ such that
$$H^{\cK}_{\min}(A|B)_{\rho} > H(A|B)_{\rho} = H_{\min}(A|B)_{\rho} \ . $$
\end{proposition}
\begin{proof}
We are going to show this for $\rho = \psi$. By Proposition \ref{prop:min-entropy-singlet-fraction}, we want to maximize $\langle \id_{A} \otimes \Psi(\psi), \tau_{d} \rangle$ over CPTNI $\Psi$ such that $J_{\Psi^{*}} \in \cK$. Clearly the unique optimal choice of map is $\Psi := U^{\ast}(\cdot)U$ as this will achieve the value of $1$ and mapping $\psi$ to any other state will achieve $<1$, though possibly to arbitrary precision. We therefore focus on this map and the arbitrary precision will be resolved by the closed-ness of the cone.
The dual map of this is $U(\cdot)U^{\ast}$ whose Choi operator is $d(\I_{A} \otimes U)(\tau_{d})(\I_{A} \otimes U)^{\ast} = d\psi$.  By assumption, $\psi$ is not in the cone so neither is $d\psi$. Therefore the Choi operator of the dual map of the optimal choice is not in the cone. Since the cone is closed one also cannot be arbitrarily approaching this point. We therefore may conclude the optimal $\Xi \in \mrm{CPTP}$ such that $J_{\Xi^{\ast}} \in \cK$ is such that $\langle (\id^{A} \otimes \Xi)(\rho_{AB}), \tau_{d} \rangle < 1$. As $H^{\cK}_{\min}$ is the negative log of this optimized inner product scaled by $d_{A}$ (Proposition \ref{prop:min-entropy-singlet-fraction}), we have $H^{\cK}_{\min}(A|B)_{\psi} > \log(d_{A}) = H(A|B)_{\psi} = H_{\min}(A|B)_{\psi}$ where we have used the strict monotonicity of the logarithm and that a straightforward calculation will show $H(A|B)_{\psi} = H(A)_{\pi} = \log(d_{A})$.
\end{proof}

We now show the same property for $H_{\min}^{\cK}(B|A)$.
\begin{proposition}
Let $\cK \subsetneq \Pos(A \otimes B)$ be closed and convex and that $d_{A}^{-1}\I^{AB} \in \relint(\cK)$. Let $\psi:=(\I^{A} \otimes U)\tau_{d}(\I^{A} \otimes U^{\ast}) \not \in \cK \subset \Pos(A:B)$ where $U$ is a unitary. Then there exists $\rho$ such that $H_{\min}^{\cK}(B|A) > H(B|A) = H_{\min}(B|A).$
\end{proposition}
\begin{proof}
Let $\Psi \in \Channel(A,B)^{\cK}$ and $\cU = U^{\ast} \cdot U$. Then,
\begin{align*}
    \langle J_{\Psi}, J_{\cU} \rangle = d^{2} \langle (\id_{A} \otimes \cU^{\ast} \circ \Psi)(\tau), \tau \rangle < d^{2} \ ,
\end{align*}
where we used the definition of the Choi operator, that the identity map is self-adjoint, $\cU^{\ast} \circ \Psi$ cannot be the identity map by assumption. Noting $H_{\min}(B|A)_{J_{\id}} = H(B|A)_{J_{\id}} = -2\log(d)$ completes the proof.
\end{proof}
\begin{remark}
If we had taken \eqref{eqn:min-ent-cptp-support} as the definition of min-entropy, this would hold without the requirement on the relative interior.
\end{remark}
We note the proposition and corollary above imply that for any cone to achieve the same value as the min-entropy in even simple cases, 
it must include the set of states that are equivalent to the maximally entangled state up to a local unitary. Noting that the proof method for this was to show that the optimal channel in the min-entropy case was not available, it's clear that the gap holds whenever there does not exist an optimal $\Psi$ for $H_{\min}$ such that $J_{\Psi^{\ast}} \in \cK$. Although in general this is hard to determine, we can use this intuition quantitatively to show the AEP requires the positive semidefinite cone. We do this by first proving a general mathematical theorem and then show how this can be applied to cones induced by resource theories that are either closed under separable maps or an Abelian symmetry.

\begin{theorem}\label{thm:anti-AEP}
Consider a sequence of cones $\cK^{(n)}((A\otimes A')^{\otimes n})$ such that $J_{\Psi^{\ast}} \in \cK^{(n)}((A\otimes A')^{\otimes n})$ iff $J_{\Psi} \in \cK^{(n)}((A\otimes A')^{\otimes n})$ where $A \cong A'$. Let $\tau_{AB}$ be $\tau_{CC'}$ embedded into $A \otimes B$ where $C$ is the smaller of $A$ and $B$ Hilbert spaces. Then
\begin{align*}
    \lim_{n \to \infty} \left[ \frac{1}{n} H^{\ve,\cK}_{\min}(A|B)_{\tau_{AB}^{\otimes n}} \right] \neq H(A|B)_{\tau_{AB}}\\
\end{align*}
if and only if $\supp_{\cK^{(n)} \cap \wt{\Density}}(\tau^{\otimes n}_{AA'}) \leq f(n)=O(c^{n})$ where $c < d_{C}/d_{A}$ where and $\wt{\Density}(A^{n} \otimes {A'}^{n}) := \{\sigma \in \Density(A^{n} \otimes {A'}^{n}) : \Tr_{{A'}^{n}}(\sigma) = \pi_{A^{n}}\}$.
\end{theorem}
\begin{proof}
We first prove the support condition implies the AEP breaks down for these states. We will show $\lim_{n \to \infty} \frac{1}{n}H^{\cK^{(n)}}_{\min}(A^{n}|B^{n})_{\tau^{\otimes n}} > H(A|B)_{\tau}$. Since smoothing the min-entropy only increases it by definition \eqref{eq:smoothed-conic-min-entropy}, it suffices to prove the unsmoothed result as it would imply it for any larger $\varepsilon$. Let $\tau_{AB}$ be the maximally entangled state $\tau_{CC'}$ embedded into $A \otimes B$ where $C$ is isomorphic to the smaller of $A$ and $B$. Then we consider
    \begin{align*}
        H^{\cK^{(n)}}_{\min}(A^{n}|B^{n})_{\tau^{\otimes n}}
        & = -\log( d_{A}^{n} \, \underset{\substack{\Psi \in \Channel(B^{n},{A'}^{n}): \\ J_{\Psi^{\ast}} \in \cK}}{\max} \langle
    (\id_{A} \otimes \Psi)(\tau_{AB}^{\otimes n}), \tau_{AA'}^{\otimes n} \rangle) \\
    & \geq -n\log(d_{A}) - \log(f(n)) \ ,
    \end{align*}
    where we used that $(\id_{A} \otimes \Psi)(\tau_{AB}) \propto J_{\Psi} \in \cK^{(n)}$ and the constraint on the support to get the inequality. It follows that, 
    \begin{align*}
         \lim_{n \to \infty} \left[ \frac{1}{n}H_{\min}^{\cK^{(n)}}(A^{n}|B^{n})_{\tau^{\otimes n}} \right]
        &\geq -\log(d_{A}) - \lim_{n \to \infty} \frac{\log(f(n))}{n} \\
        & = -\log(d_{A}) - \log(c) \\
        & = -\log(cd_{A}) \\
        & > -\log(d_{C}) = H(A|B)_{\tau} \ , 
    \end{align*}
    where we used the assumption $f(n) = O(c^{n})$ in the second and third inequalties. The final equality follows from our assumption $c < d_{C}/d_{A}$ and that it is quick to verify $H(A|B) = -D(\rho_{AB}||\I_{A} \otimes \rho_{B}) = 0 -\log(d_{C})$ for $\tau_{AB}$ where we remind the reader $C$ is a copy of the smaller of $A,B$. Again, as smoothing only aggravates this problem, this completes the first direction. \\
    
    For the other direction, just note if $f(n) = 2^{o(n)}$, i.e.\ is sub-exponential, then the second term in the lower bound is zero and the AEP holds and smoothing will continue this. The only other case is the limiting behaviour is $\lim_{n \to \infty} \langle (\id_{A} \otimes \Psi)(\wt{\rho}^{n}), \tau_{AA'}^{\otimes n} \rangle = 1$, where $\wt{\rho}^{n}$ is in the smoothing ball of $\tau_{AA'}^{\otimes n}$, but if this holds, then the AEP holds by definition. 
\end{proof}
The flipped case is almost the same except the cones now depend on $A \otimes B$ rather than using a copy of $A'$.
\begin{theorem}\label{thm:anti-AEP-flipped}
Consider a sequence of cones $\cK^{(n)}((A\otimes B)^{\otimes n})$ such that $J_{\Psi} \in \cK^{(n)}((A\otimes B)^{\otimes n})$. Then
\begin{align*}
    \lim_{n \to \infty} \left[ \frac{1}{n} H^{\ve,\cK}_{\min}(B|A)_{\tau_{AB}^{\otimes n}} \right] \neq H(B|A)_{\tau_{AB}}\\
\end{align*}
if and only if $\supp_{\cK^{(n)} \cap \wt{\Density}}(\tau^{\otimes n}_{AB}) \leq f(n)=O(c^{n})$ where $c < 1$ and $\wt{\Density}(A^{n} \otimes B^{n}) := \{\sigma \in \Density(A^{n} \otimes B^{n}) : \Tr_{B^{n}} = \pi_{A^{n}}\}$.
\end{theorem}
\begin{proof}
By \eqref{eqn:min-ent-cptp-support}, it's clear that
\begin{align*}
    &  \frac{1}{n} H^{\ve,\cK}_{\min}(A|B)_{\tau_{AB}^{\otimes n}}
    = -\frac{1}{n} \log \underset{\Phi \in \Channel^{\cK}}{\max} \langle J_{\Phi}, \tau^{\otimes n} \rangle 
    = -\frac{1}{n} \left[ n \log d_{A} + \log(\supp_{\cK^{(n)} \cap \wt{\Density}}(\tau^{\otimes n}_{AB})) \right] \ ,
\end{align*}
where the last line we used that $\Tr_{B^{n}}J_{\Phi} = I^{A^{n}}$. Then if $\log(\supp_{\cK^{(n)} \cap \wt{\Density}}(\tau^{\otimes n}_{AB})) \leq nc$ for $c < 1$ then the limit is greater than $-\log(d_{A})$, otherwise it isn't. Again, smoothing only aggravates the problem. 
\end{proof}

We now show how the anti-standard min-entropy AEP applies to entanglement. To do this, 
note that the adjoint map can be written in Kraus form as $\Phi^{\ast}(\cdot) = \sum_{i} A_{i}^{\ast} \cdot A_{i}$, so by Proposition \ref{prop:ent-rank-channel-property}, $\mrm{SN}(J_{\Phi}) = \mrm{SN}(J_{\Phi^{\ast}})$. 

\begin{corollary}\label{corr:ent-rank-anti-AEP}
An AEP cannot hold in general for any sequence of cones $\{\cK^{(n)}\}$ such that $\cK^{(n)} \subseteq \mrm{Ent}_{f(n)}$ and there exists $n_{0} \in \mbb{N}$ such that $f(n) < d^{n}$ for all $n > n_{0}$ where $d = \min\{d_{A},d_{B}\}$. In particular, if $f(n) = \lfloor(d-\ve)^{n}\rfloor$, where $\ve > 0$, the AEP doesn't hold. In other words, any exponential increase in allowed entanglement as a fraction of the original local space that does not recover the entire positive semidefinite cone cannot admit an AEP.
\end{corollary}
\begin{proof}
If $\cK^{(n)} = \mrm{Ent}_{r}$ then $J_{\Phi} \in \cK^{(n)}(B:A')$ and $J_{\Phi} \in \cK^{(n)}(A:B)$. By Eqn. \eqref{eqn:PPT-max-ent-supp}, the support of $\tau^{\otimes n}$ on the space $\mrm{Ent}_{r}$ is less than or equal to $r/d_{A}^{n}$, so letting $f(n) = \lfloor d_{A}-\ve)^{n} \rfloor$ for $\ve >0$ allows Theorem 6 to apply. This completes the proof.
\end{proof}

\paragraph*{Local CP Map invariance of Free Resources Sufficient for anti-standard min-entropy AEP---}
While the above is useful for exemplifying that entanglement is the property mathematically captured by a fully quantum AEP, if we wished to interpret this in terms of resource theories, this is not as insightful. This is because normally a resource theory is defined on a space and then one considers many copies of states on that space. In other words, the cones in a generalized resource theory might be best thought of as probabilistic mixtures of many copies of some free resource, as we will soon formalize. It is a natural question to ask what property must a generalized resource theory satisfy for these induced cones to lead to an anti-standard AEP. Here we will show it suffices for the free resources to be closed under local CP maps and not contain the maximally entangled state. To do this, we will need the notion of completely free-resource-forcing channel.
\begin{definition}
Given a set of free states $\cF$ that can be well-defined for any Hilbert space, a channel $\Phi \in \Channel(A,B)$ is completely resource-breaking (CRB) if $(\id_{A} \otimes \Phi)(\rho_{RA}) \in \cF(RB)$ for all $\rho \in \Density(RA)$.
\end{definition}
First we remark that CRB channels are not the same as resource-destroying maps \cite{Liu-2017a}. First, resource-destroying maps do not consider an ancillary space. Second, a map is resource-destroying only if any free state is mapped to itself. In contrast, the map that throws out the input and prepares a fixed free state is likely to be CRB, e.g.\ if one thinks of separable states as being the free states. Having established this, we show that for resource theories that capture locality, the CRB channels play an important algebraic role in the resource theory.
\begin{proposition}\label{prop:choi-c-forcing}
Consider a resource theory with a set of free states $\cF$ that is well-defined for any Hilbert space and stable under local CP maps. Then $\Phi \in \Channel(A,B)$ is a completely resource-breaking channel if and only if $J_{\Phi} \in \cK_{\cF}$.
\end{proposition}
\begin{proof}
($\Rightarrow$) Let $\Phi \in \Channel(A,B)$ be a CRB channel. Then $J_{\Phi} = d_{A} (\id_{A} \otimes \Phi)(\tau_{AA'}) = d_{A} \sigma_{AB} \in \cK_{\cF}$ where we have used the definition of CRB channel. \\
($\Leftarrow$) Let $J_{\Phi} \in \cK_{\cF}$. Then define 
$$d_{A}\rho^{AB} := (\id_{A} \otimes \Phi)(\Phi^{+}) = J_{\Phi} \in \cK_{\cF} \ . $$
Recall that given $\sigma_{A}$, $\psi_{AA'} = d_{A}^{-1} \sqrt{\sigma_{A}}\Phi^{+}\sqrt{\sigma_{A}}$ is a purification of $\sigma_{A}$. It follows that for any $\sigma_{A}$,
$$ d_{A} (\id_{A} \otimes \Phi)(\psi_{AA'}) = d_{A}\sqrt{\sigma_{A}} J_{\Phi} \sqrt{\sigma_{A}} \in \cK_{\cF} \ , $$
where the containment follows from the assumption $\cF$ is local CP map invariant. As there is no restriction on $\sigma_{A}$, we have for any pure state $\psi_{AA'}$, the above equation holds. Therefore, by linearity, $(\id_{A} \otimes \Phi)(\rho_{AA'}) \in \cF$ for all $\rho_{AA'}$. 

To extend to arbitrary $\id_{R}$, note that for any $R$ it suffices to prove it for pure states $\psi_{RA}$ due to linearity. By isometric equivalence of purifications, $\ket{\psi}_{RA'} = (V_{A\to R} \otimes \id_{A})\ket{\psi}_{AA'}$, which again results in a local CP map. This completes the proof.
\end{proof}
As noted, in a resource theory, you normally define the free states on a single copy of the (multipartite) space and then use many copies, e.g. your resource is $\sigma^{\otimes n}$ where $\sigma \in \cF \subset \Density(A \otimes B)$. You could then in principle use probabilistic mixtures of such a state, i.e. your states are the convex hull . Thus for $n$ copies the natural set to consider is
$$ \cF_{n} := \mrm{conv}(\{\sigma^{\otimes n}:\sigma \in \cF \}) \ .$$
Then we can generate the cone from this. We now combine the above result with Theorem \ref{thm:anti-stein-lemma-support-version} to get the following, which in effect says you won't satisfy an AEP with the cone unless the maximal resource is contained in the free states.
\begin{corollary}
For a resource theory with free states $\cF \subset \Density(A \otimes B)$ stable under local CP maps, the AEP does not hold unless $\tau \in \cF$, i.e. the maximally entangled state is included.
\end{corollary}
\begin{proof}
If $\tau \not \in \cF$, there exists $\delta \in (0,1)$ such that $\max_{\sigma \in \cF} \langle \sigma , \tau \rangle \leq \delta$. Note this will also be true if you swap the spaces and take the transpose as $\tau$ is invariant under both of these operations. This is useful for our purposes as these operations take $J_{\Psi}$ to $J_{\Psi^{\ast}}$. Define $\wt{\cF}_{n} := \{{\mrm{SWAP}(\sigma^{\otimes n})}^{\Trans} : \sigma \in  \cF\}$.
It follows $$\max_{\sigma^{n} \in \wt{\cF}_{n}} \langle \sigma^{n},\tau^{\otimes n}\rangle = \max_{\nu} \int_{\sigma \in \wt{\cF}} \langle \sigma^{\otimes n}, \tau^{\otimes n} \rangle  d\nu(\sigma) \leq \delta^{n} \ , $$ 
where $\nu$ is optimized over the space of probability measures over $\wt{\cF}$ and the final equality follows by linearity. Therefore, by Proposition \ref{prop:choi-c-forcing}, $\supp_{\sigma \in \mrm{cone}(\wt{\cF}_{n}) \cap \widetilde{\Density}}\langle \sigma, \tau^{\otimes n} \rangle \leq \delta^{n}$. As $\wt{\cF}_{n}$ corresponds to swapping the spaces and taking the complex conjugate, this is the space of adjoint maps. Thus we can apply Theorem \ref{thm:anti-AEP} without needing to assume $J_{\Psi^{\ast}} \in \cK$ to break the AEP as $\delta^{n} \to 0$ since $\delta < 1$. We can do an identical argument for $\cF_{n}$ to apply Theorem \ref{thm:anti-AEP-flipped}. This completes the proof.
\end{proof}

\paragraph*{Coherence anti-standard min-entropy AEP}
While in this setting entanglement does seem to be the fundamental property, it may be of interest that we can present a coherence version of the anti-standard AEP, since we were able to present an entanglement version of breaking the cone-restricted max-divergence AEP. To do this, we do note that certain incoherent operations have algebraic structure that imply properties of their Choi matrices in terms of entanglement rank.
\begin{proposition}\label{prop:classical-to-classical-Choi}
A quantum channel $\Psi \in \Channel(A,B)$ is classical-to-classical in a preferred basis if and only if $J_{\Psi} \in \cI(A \otimes B)$. Moreover, this implies $J_{\Psi^{\ast}} \in \cI(A \otimes B)$.
\end{proposition}
Therefore, as $\cI(A \otimes B) \subset \Sep(A:B)$, we can view Corollary \ref{corr:ent-rank-anti-AEP} as capturing this case. More generally, the maximally allowed coherent subspace of a PIO is captured by the maximal rank over its complete set of incoherent projectors. The maximal rank of the PIO incoherent projectors can then be shown to be equivalent to the Schmidt number of the PIO's Choi operator as the following proposition states.
\begin{proposition}\label{prop:PIO-Entanglement-Rank}
Let $\Phi \in \Channel(A,A)$ be a PIO with Kraus operators $K_{i} = U_{i}\Pi_{i}$. Then $\mrm{SN}(J_{\Phi}) = \max_{i}\{\rank(\Pi_{i})\}$. Operationally, for a given PIO $\Phi$, the Schmidt number of $J_{\Phi}$ is the size of the largest subspace in which coherence may be preserved.
\end{proposition}
Therefore, we can view Corollary \ref{corr:ent-rank-anti-AEP} as capturing the coherence case as well. Proofs of both propositions are given in the appendix.

\section{Restricted Min- and Max-Entropy Duality}
One well known property of the standard min- and max-entropy mentioned in the background is that they satisfy a special duality: for pure states $\rho_{ABC}$, $H_{\min}(A|B) = -H_{\max}(A|C)$ \cite{Konig-2009a}. One might hope that we can generalize this notion for these generalized min-entropies, and indeed we can as we now show. As we will need to use purifications which are local unitarily invariant, we will require that the cone $\cK(A\otimes B)$ be invariant under local unitaries under both spaces. We will call such a cone \textit{physical} in the sense we haven't restricted the local frame of reference of Alice (resp. Bob) if you imagine they have access to the $A$ (resp. $B$) space of the joint state. Rather than define the restricted max-entropy as the dual of the restricted min-entropy directly, we present a generalization of the max-entropy and then verify it satisfies duality. We believe this provides better intuition further on.

\begin{definition}[Restricted Max-Entropy]\label{defn:restricted-max-entropy}
Consider $\rho_{AC} \in \Density(A\otimes C)$ and an arbitrary purification $\rho_{ABC}$ of it. Let $\cK \subseteq \Pos(A \otimes B)$ be physical. Define $\Channel^{\cK}(B,A') := \{\Phi \in \Channel(B,A') : J_{\Phi^{\ast}} \in \cK \}$. Then define
\begin{equation}
\begin{aligned}
d_{A}^{-1}\exp(H^{\cK}_{\max}(A|C))  := \max_{\substack{\Phi \in \Channel^{\cK}(B,A') \\ \omega \in \Density(C)}} \, F^{2}(\Phi(\rho_{ABC}), \tau_{AA'} \otimes \omega_{C}) \ .
\end{aligned}
\end{equation}
\end{definition}
We first note that the definition is global in the sense that to even define the measure required reference to the purified space $B$. 
This differs than in the standard min-entropy, i.e.\ the $\cK = \Pos$ case. We also note that the R.H.S.\ of the definition measures how decoupled the $AB$ space can be from the $C$ space. This aligns with the operational interpretation of $H_{\max}$ in terms of decoupling accuracy \cite{Konig-2009a}. We will return to these points once we prove this is the correct definition, which will require Uhlmann's theorem which we state first for completeness.
\begin{proposition}[Uhlmann's Theorem + Some More]\label{prop:Uhlmann+}
Let $P_0 , P_1 \in \Pos(A)$. Let $\Tr_{R}(Q_0 ) = P_0$. Then,
$$F(P_0,P_1) = \max_{Q_1: \Tr_{R}(Q_1) = P_1} F(Q_0,Q_1) \ . $$ 
Moreover, if $Q_{0}$ is a purification of $P_0$, one can restrict $Q_1$ to being a purification of $P_1$.
\end{proposition}
We now prove we have the correct definition for $H^{\cK}_{\max}$. We note that the proof is in effect the proof of duality between standard min- and max-entropy from \cite{Tomamichel-2012a} but generalized to more cones.
\begin{corollary}
For pure $\rho_{ABC}$ and $\cK \subseteq \Pos(A \otimes B)$ which is physical,
$$ H^{\cK}_{\max}(A|C) = -H_{\min}^{\cK}(A|B) \ . $$
\end{corollary}
\begin{proof}
We first need to convert $H_{\min}^{\cK}$ into appropriate form. Let $V_{\Psi}$ denote the isometric representation of a map $\Psi$. We are going to consider the maps in the optimization for $H_{\min}^{\cK}$, we note these maps will be $V_{\Psi}:B \to A'R$. By applying Proposition \ref{prop:Uhlmann+}, we have
\begin{align*}
F^{2}((\id_{A} \otimes \Psi)(\rho), \tau) = \max_{\sigma : \Tr_{R}(\sigma) = \tau} F^{2}(\rho_{AA'R},\sigma) \ , 
\end{align*}
where $\rho_{AA'R} := (\id_{A} \otimes V)(\rho)(\id_{A} \otimes V^{\ast})$. Noting that any extension of $\tau$ is of the form $\tau \otimes \omega$, using Proposition \ref{prop:min-entropy-singlet-fraction}, we may conclude:
\begin{equation}\label{eqn:isometric_Hmin} 
\begin{aligned}
 d_{A}^{-1} \exp(-H^{\cK}_{\min}(A|B)) 
:= \max_{\substack{V_{\Psi} \in \mrm{Iso}(B,A'R): \\ J_{\Psi^{\ast}} \in \cK}} \, \max_{\omega \in \Density(R)} \, F^{2}(\rho_{AA'R}, \tau_{AA'} \otimes \omega) \ .
\end{aligned}
\end{equation}

We prove this in two directions.
\begin{align*}
    \exp(H^{\cK}_{\max}(A|C))
    =& |A| \max_{\Phi \in \Channel^{\cK}(B,A')} \max_{\omega \in \Density(C)} \, F^{2}(\Phi(\rho_{ABC}), \tau_{AA'} \otimes \omega_{C}) \\
    =& |A| \max_{\cV \in \mrm{Iso}^{\cK}} \max_{\omega \in \Density(RC)} \, F^{2}(\cV(\rho_{ABC}), \tau_{AA'} \otimes \omega_{RC}) \\
    \leq& |A| \max_{\cV \in \mrm{Iso}^{\cK}} \max_{\omega \in \Density(R)} \, F^{2}(\Tr_{C} \circ \cV(\rho_{ABC}), \tau_{AA'} \otimes \omega_{R}) \\
    =& \exp(-H^{\cK}_{\min}(A|B)) \ ,
\end{align*}
where first line is definition, second line is Uhlmann's theorem and the definition
\begin{align*}
\mrm{Iso}^{\cK} := 
&\{V \in \mrm{Iso}(B,A'R) : J_{\Psi^{\ast}} \in \cK \, , \, \Psi(\cdot) := \Tr_{R} \circ (V \cdot V^\ast) \} \ ,
\end{align*}
the third line is fidelity's DPI, fourth line is Eqn. \eqref{eqn:isometric_Hmin}.
\begin{align*}
    \exp(-H^{\cK}_{\min}(A|B))
    =& \max_{\cV \in \mrm{Iso}^{\cK}} \max_{\omega \in \Density(R)} \, F^{2}(\cV(\rho_{AB}), \tau_{AA'} \otimes \omega_{R}) \\ 
    =& \max_{\cV \in \mrm{Iso}^{\cK}} \max_{\omega \in \Density(RC)} \, F^{2}(\cV(\rho_{ABC}), \tau_{AA'} \otimes \omega_{RC}) \\
    \leq& \max_{\cV \in \mrm{Iso}^{\cK}} \max_{\omega \in \Density(C)} \, F^{2}(\Tr_{R} \circ \cV(\rho_{ABC}), \tau_{AA'} \otimes \omega_{C}) \\
    =& \max_{\Phi \in \Channel^{\cK}(B,A')} \max_{\omega \in \Density(C)} \, F^{2}(\Phi(\rho_{ABC}), \tau_{AA'} \otimes \omega_{C}) \\
    =&\exp(H^{\cK}_{\max}(A|C)) \ ,
\end{align*}
where the first line is Eqn. \eqref{eqn:isometric_Hmin}, the second is Uhlmann's theorem, the third is fidelity's DPI, and the fourth is the definition.
\end{proof}
This verifies that we have chosen the correct definition of $H^{\cK}_{\max}$ if we are interested in the duality property holding. As noted earlier, this is defined on the tripartite space, but the standard $H_{\max}$ does not need to consider the purified space. In effect this is a property of the fact that without the positive semidefinite cone, one cannot consider all physically realizable dynamics on the purified space and so the purified space can no longer be ignored. For clarity, we now present the proof that makes this explicit.
\begin{proposition}
If $\cK = \Pos(A \otimes B)$, then $H^{\cK}_{\max}(A|C)_{\rho} = d_{A} \max_{\omega_{C} \in \Density(C)} F^{2}(\rho_{AC},\pi_{A} \otimes \omega_{C})$
\end{proposition}
\begin{proof}
First note that 
$$ \max_{\Phi \in \Channel(B,A)} F(\Phi(\rho_{ABC}),\tau_{AA'} \otimes \omega_{C}) \leq F(\rho_{AC},\pi_{A} \otimes \omega_{C}) \ , $$
which holds by fidelity DPI under partial trace. Next
use Proposition \ref{prop:Uhlmann+} to say that
\begin{equation*}
F(\rho_{AC}, \pi_{A} \otimes \omega_{C}) = \max_{\substack{\sigma \in \Density(AA'C) \\ \Tr_{A'}(\sigma) = \rho_{AC}}} F(\sigma, \tau_{AA'} \otimes \omega) \ .
\end{equation*}
Note that given pure $\rho_{ABC}$, any extension of $\rho_{AC}$ can be obtained via a channel $\cN$ acting on $B$. This is because for extension $\sigma_{ACD}$, one can consider the purification $\widehat{\sigma}_{ACDE}$. By the isometric equivalence of purifications, there is $V_{B \to DE}$ to take $\rho_{ABC}$ to $\widehat{\sigma}_{ACDE}$ and then just trace off $E$, which constructs such a $\cN_{B \to D}$. Therefore, for $\sigma^{\star}$ that optimizes the R.H.S. above, there exists $\cN_{B \to A'}$ such that $\cN(\rho_{ABC}) = \sigma^{\star}.$ Therefore 
$$ F(\rho_{AC}, \pi_{A} \otimes \omega_{C}) \leq \underset{\Phi \in \Channel(B,A')}{\max} F(\Phi(\rho_{ABC}),\tau_{AA'} \otimes \omega_{C}) \ . $$
Therefore for any $\omega_{C}$, $$F(\rho_{AC},\pi_{A} \otimes \omega_{C}) = \max_{\Phi \in \Channel(B,A')} F(\Phi(\rho_{ABC}),\tau_{AA'} \otimes \omega_{C}) \ ,$$
so optimizing over $\omega_{C}$ completes the proof.
\end{proof}
In effect, one might view the result that we can only consider the local spaces alone when all CPTP dynamics are considered as a notion of physicality. Indeed, one might view this as a notion somewhat similar to that of \cite[Theorem 1]{Buscemi-2014a}.  Reference \cite{Buscemi-2014a} shows that given a tripartite distribution $\rho_{ABC}$, there exist local CPTP dynamics on purified space $B$ for all global dynamics on $BC$ if and only if $\rho_{AC}$ is not entangled.\footnote{This seems to not be stated explicitly in \cite{Buscemi-2014a}, but it is equivalent as $I(A:C|B)_{\rho} = 0$ if and only if $\rho_{AC}$ is separable by the algebraic decomposition of a Markov chain.} In contrast, what the previous proposition suggests is that if and only if one considers all local dynamics on $B$ can $A$ and $C$ be made as decoupled as possible. In other words, \cite{Buscemi-2014a} considers what can happen to $B$ given $AC$ is sufficiently decoupled and the previous proposition asks how decoupled $AC$ can become given $B$. We will make this intuition rigorous by showing for restricted max-entropy there are completely decoupled states for which we cannot achieve the decoupling accuracy.

\begin{proposition}
Let the sequence of cones $\cK^{(n)}$ be such that $\Phi \in \cK(B^{n}:A^{n})$ iff $\Phi^{\ast} \in \cK(A^{n}:B^{n})$. Let $p_{XZ} = \pi_{X} \otimes \dyad{z}$. Then, 
$$ \frac{1}{n} \left[H^{\cK}_{\max}(X^{n}|Z^{n})_{p_{XZ}^{\otimes n}}\right] \neq H(X^{n}|Z^{n})_{p_{XZ}^{\otimes n}} $$
if and only if $\supp_{\cK^{(n)} \cap \Density}(\tau^{\otimes n}) \leq f(n) = O(c^{n})$ where $c < 1.$ 
\end{proposition}
\begin{proof}
Let $p_{XZ} = \pi_{X} \otimes \dyad{z}$ where $z \in \cZ$, i.e.\ the standard basis (note this could be trivial if $Z \cong \mbb{C}$). Then a purification of this is $\rho_{XX'Z} = \tau_{XX'} \otimes \dyad{z}$. Then we have 
\begin{align*}
    H^{\cK}_{\max}(X^{n}|Z^{n})_{p_{XZ}^{\otimes n}} 
    =& \log(d_{X}^{n}) +  2\underset{\Phi: J_{\Phi^{\ast}} \in \cK^{(n)}}{\max} \log[F\{(\id_{X^{n}} \otimes \Phi)(\tau_{XX'}^{\otimes n}), \tau_{XX'}^{\otimes n}\}] \\ 
    =& \log(d_{X}^{n}) + 2 \underset{\Phi: J_{\Phi^{\ast}} \in \cK^{(n)}}{\max} \log[ \langle (\id_{X} \otimes \Phi)(\tau_{XX'}^{\otimes n}), \tau_{XX'}^{\otimes n}\rangle ] \\
    \leq & n\log(d_{X}) + \log(f(n)) \ ,
\end{align*}
where the first line is just simplifying Definition \ref{defn:restricted-max-entropy} by making the optimal choice $\omega_{Z^{n}} = \dyad{z}^{\otimes n}$, the second is using for a pure state fidelity reduces to inner product, the third is by assumption. Then regularizing,
\begin{align*}
\frac{1}{n} H^{\cK}_{\max}(X^{n}|Z^{n})
= \log(d_{x}) + \log(c) 
< \log(d_{x}) 
= H(X|Z)_{\pi_{X} \otimes \dyad{z}} \ .
\end{align*}
Note that smoothing can only aggravate this inequality as $H_{\max}^{\cK,\ve}$ minimizes over the smoothed ball. This completes the if direction.

For the only if direction, it is clear that if $f(n) = 2^{o(n)}$, i.e.\ is sub-exponential,
then the correction term goes away and the AEP will hold. The only other case is $\lim_{n \to \infty} f(n) = 1$, in which case the AEP also holds.
\end{proof}
This then holds in the entanglement case. As we have seen the proof multiple times, we omit the proof.
\begin{corollary}\label{corr:HmaxK-anti-AEP}
The $H^{\cK^{(n)}}_{\max}$ does not satisfy an AEP on classical states for a sequence of cones $\{\cK^{(n)}\}$ such that for all $n \geq n_{0}$ $\cK^(n) \subset \Ent_{f(n)}(A^{n}:B^{n})$ where $f(n) \leq O(c^{n})$ where $c<1$. 
\end{corollary}
We stress the importance that $H_{\max}^{\cK}$ breaks down on classical states as we will next develop the theory necessary to show we can recover partially quantum information theory using the separable cone which means we must do this without using the restricted max entropy $H_{\max}^{\cK}$.

\section{Symmetries, Data-Processing Inequalities, and Special Cases}
To prove that the separable cone is sufficient for partially quantum information theory, it will be useful to consider the general theory of symmetries and data processing inequality for smoothed and unsmoothed restricted min-entropy, where we recall $H_{\min}^{\cK,\ve}(A|B)_{\rho} = \underset{\wt{\rho} \in \cB^{\ve}(\rho)}{\max} H^{\cK}_{\min}(A|B)_{\wt{\rho}}$.

Recalling Corollary \ref{corr:DPI-breaks} shows restricted relative entropy does not satisfy a data-processing inequality, we generalize \cite[Proposition 3]{Chitambar-2021a} in the straightforward manner to get the general conic data-processing settings.

\begin{proposition}\label{prop:DPI}
Consider $\cK \subseteq \Pos$. Given $\Phi:\Pos(A \otimes B) \to \Pos(A' \otimes B')$ such that $\Phi^{\ast}:\cK(A':B') \to \cK(A:B)$, then $D_{\max}^{\cK}(\Phi(P)||\Phi(Q)) \leq D_{\max}^{\cK}(P||Q)$ for all $P,Q \in \Pos(A \otimes B)$. Moreover, if there exists $\phi: \Pos(B) \to \Pos(B')$ that is trace-preserving such that $\Phi(\I_{A} \otimes \sigma_{B}) \preceq_{\cK{\ast}} \I_{A'} \otimes \phi(\sigma_{B})$ for all $\sigma_{B} \in \Density(B)$. Then $H^{\cK}_{\min}(A|B)_{P} \leq H^{\cK}_{\min}(A|B)_{\Phi(P)}$ for all $P \in \Pos(A \otimes B)$
\end{proposition}
\begin{proof}
We begin with the $D_{\max}^{\cK}$ case. Recall $D_{\max}^{\cK}(P||Q) = \log \inf \{ \gamma : \gamma Q - P \in \cK^{\ast} \}$. Recall by definition of dual cone, $X \in \cK^{\ast}(A:B)$ iff $\langle X , Y \rangle \geq 0$ for all $Y \in \cK(A:B)$. It follows, if $X \in \cK^{\ast}(A:B)$ and any $Y \in \cK(A':B')$, $\langle \Phi(X) , Y \rangle = \langle X , \Phi^{\ast}(Y) \rangle \geq 0$ where the last inequality is by the assumption $\Phi^{\ast}: \cK(A':B') \to \cK(A:B)$. Thus $\Phi: \cK^{\ast}(A:B) \to \cK^{\ast}(A':B')$. Consider any feasible $\gamma$ for $D_{\max}^{\cK}(P||Q)$. Then we have
\begin{align*}
    0 \preceq_{\cK}^{\ast} \Phi(\gamma Q - P) = \gamma \Phi(Q) - \Phi(P) \ .
\end{align*}
Therefore any feasible $\gamma$ for $D_{\max}^{\cK}(P||Q)$ is feasible for $D_{\max}^{\cK}(\Phi(P)||\Phi(Q))$. As we are considering an infimum, this completes the proof of the $D_{\max}^{\cK}$ case.

For the restricted min-entropy case, recall that 
\begin{align*}
 H^{\cK}_{\min}(A|B)_{\rho}
:= - \min \{\lambda \in \mbb{R}: 2^{\lambda} \I_{A} \otimes \sigma_{B} - \rho \succeq_{\cK^{\ast}} 0 \, \& \, \sigma_{B} \in \Density(B) \} \ .
\end{align*}
Consider any feasible point $(\lambda,\sigma_{B})$ for $H^{\cK}_{\min}(A|B)_{\rho}$, then we have
\begin{align*}
    0 \preceq_{\cK^{\ast}}  \Phi(2^{\lambda}\I_{A} \otimes \sigma_{B} - \rho)
    =  2^{\lambda}\Phi(\I_{A} \otimes \sigma_{B}) - \Phi(\rho) 
    \preceq_{\cK^{\ast}}  2^{\lambda} \I_{A'} \otimes
    \phi(\sigma_{B}) - \Phi(\rho),
\end{align*}
where in the second generalized inequality, we used our assumption on the existence of $\phi$.Thus, for any feasible $(\lambda,\sigma_{B})$ for $H^{\cK}_{\min}(A|B)_{\rho}$, $(\lambda,\phi(\sigma_{B}))$ is feasible $H^{\cK}_{\min}(A'|B')_{\Phi(\rho)}$. Thus the optimal $\lambda$ for determining $H^{\cK}_{\min}(A'|B')_{\Phi(\rho)}$ is at least as small as that for $H^{\cK}_{\min}(A|B)_{\Phi(\rho)}$. Taking into account the minus sign proves $H^{\cK}_{\min}(A'|B')_{\Phi(\rho)} \geq H^{\cK}_{\min}(A|B)_{\rho}$.
\end{proof}

We now want to consider the special case of states with symmetries. As introduced in the background, a `symmetry' of a state $\rho$ means $U\rho U^{\ast} = \rho$ where $U$ is some unitary. If you take all the symmetries of a given state, they form a group. Given the group $G$ of symmetries of a state $\rho$, define the projection/twirling map $\cG(\cdot) = \int U \cdot U^{\ast} dU$ where $dU$ is the Haar measure for the relevant group. Clearly $\cG(\rho) = \rho$. 

Symmetries are useful for two reasons: first, they are relevant throughout physics and determine coherence resource theories. Second, they can take complicated problems and reduce them to be more manageable. In particular, this has been useful for simplifying entanglement measures in certain cases \cite{Vollbrecht-2001a}, and as we are interested in entanglement rank cones, one might expect this to be helpful here as well.

\begin{remark}
We note a more general class of operators of similar form are mixed unitary channels where the difference is the unitaries need not make up a group. In effect, when a mixed unitary channel is not composed of a group of unitaries that are integrated according to the Haar measure, the channel won't be self-adjoint, which means moving it onto the observable doesn't restrict the relevant state space.  Hence we do not consider such channels here.
\end{remark}

We now prove a general abstract theorem and then consider specific cases.
\begin{theorem}\label{thm:optimizing-with-symmetry}
Let $\rho_{AB}$ be invariant under a projection channel $\cG$. Let $\cG$ satisfy data-processing for $H_{\min}^{\cK}(A|B)_{\rho}$ where $\cK \subset \Pos(A \otimes B)$. Then without loss of generality, when determining $H^{\cK,\ve}_{\min}(A|B)_{\rho}$, one can restrict to optimizing over $\{\wt{\rho} \in \cB^{\ve}(\rho): \cG(\wt{\rho}) = \wt{\rho}\}$.
\end{theorem}
\begin{proof}
Since $\rho$ is invariant under $\cG$, if $\wt{\rho} \in \cB^{\ve}(\rho)$, then $\cG(\rho) \in \cB^{\ve}(\rho)$ as purified distance monotonically decreases under CPTNI maps. Since, by assumption, $H^{\cK}_{\min}$ only increases under the action of $\cG$, the optimizer must be in the post-twirled set. Noting $\cG$ is a linear (and thus continuous) map and $\cB^{\ve}(\rho)$ is compact, the restricted set is compact. Thus we can restrict ourselves to this set and achieve the same value.
\end{proof}

We now present the less general but primarily relevant setting.

\begin{corollary}\label{corr:local-symmetry-restriction}
Let $\cK \subset \Pos(A \otimes B)$ be physical (closed under local unitaries). Let $\rho_{AB}$ be invariant under $\cG(\cdot) = \int (U_{A} \otimes U_{B}) \cdot (U_{A} \otimes U_{B})^{\ast} dU$, i.e.\ a local symmetry. Then one can restrict to optimizing over $\wt{\rho} \in \cB^{\ve}(\rho)$ such that $\cG(\wt{\rho}) = \wt{\rho}$ when calculating $H^{\cK,\ve}_{\min}(A|B)$. Moreover, the optimizer for the primal problem $\exp(-H^{\cK}_{\min}(A|B))$ will be invariant under $\cG$.
\end{corollary}
\begin{proof}
Given the previous theorem, for the first point to hold, it suffices to prove $\cG$ of the promised form satisfies data-processing. Since $\cK$ is closed under local unitaries, $\cG: \cK \to \cK$. Moreover 
\begin{align*}
\cG(\I_{A} \otimes \sigma_{B}) = \int (U_{A}\I_{A}U^{\ast}_{A} \otimes U_{B}\sigma U_{B}^{\ast}) \, dU 
= \I_{A} \otimes \int U_{B}\sigma U_{B}^{\ast} \, dU \ , 
\end{align*}
so the conditions of Proposition \ref{prop:DPI} hold using the map $\phi(\cdot) = \int U_{B} \cdot U_{B}^{\ast} \, dU$.

To prove the second point, just note that, given Eqn. \ref{eqn:HminKPrimal}, one needs to consider 
$$ \max\{ \langle \wt{\rho}, X \rangle : \Tr_{A}(X) = \I_{B}, \, X \in \cK \} \ . $$
As just addressed, without loss of generality, $\wt{\rho}$ is invariant under $\cG$. As $\cG$ is self-adjoint, we can move it onto $X$. Noting that $\Tr_{A} \circ \cG (X) = \I_{B}$ if $\Tr_{A}(X) = \I_{B}$, we can conclude the twirling preserves feasibility, so we can without loss of generality restrict to optimizers of the form $\cG(X)$. This completes the proof.
\end{proof}

Note the above could also be applied for multiple symmetries if $\rho$ is invariant under multiple projection maps, though in this case the commutant of the group is known to in general become more complicated. We also quickly prove that this is not special to cones as something similar holds for Sandwiched R\'{e}nyi divergences \eqref{eqn:SRD-defn}.

\begin{proposition}\label{prop:SRE}
Let $\alpha \in (1/2) \cup (1,\infty)$. Let $\rho$ be invariant under local symmetry $\cG(\cdot) = \int (U_{A} \otimes U_{B}) \cdot (U_{A} \otimes U_{B}) \, dU$ such that $\int U_{B} X U_{B} \,  dU = \int U_{B} X U_{B} dU_{B}$ for all $X \in \Lin(B)$.\footnote{That is, the Haar measure $dU$ acts like the Haar measure $dU_{B}$ when only acting on the $B$ space. 
This is always satisfied for a local symmetry by definition.} Then 
$$ H^{\uparrow}_{\alpha}(A|B)_{\rho} = \sup_{q \in P(\Lambda)} -\wt{D}_{\alpha}(\rho|| \I_{A} \otimes \sum_{\lambda} q(\lambda)\Pi_{\lambda}) \ , $$
where $\{\Pi_{\lambda}\}_{\lambda \in \Lambda}$ is the minimal projectors for the commutant of the symmetry captured by $\int U_{B} \cdot U_{B} \, dU_{B}$.
\end{proposition}
\begin{proof}
Using data-processing,
\begin{align*}
    -\wt{D}_{\alpha}(\rho||\I_{A} \otimes \sigma) \leq -\wt{D}_{\alpha}(\rho||\I_{A} \otimes \int U_{B}\sigma U_{B}^{\ast} \, dU) \ .
\end{align*}
As $H^{\uparrow}_{\alpha}(A|B)_{\rho} = \max_{\sigma_{B}} - \wt{D}_{\alpha}(\rho||\I_{A} \otimes \sigma)$, the above tells us we can restrict the optimization to $\sigma$ that are invariant under $\int U_{B} \cdot U_{B} \, dU_{B}$. Let $\{\Pi_{\lambda}\}_{\lambda \in \Lambda}$ be the set of minimal projectors for the commutant of the group corresponding to the $\int U_{B} \cdot U_{B} \, dU_{B}$ symmetry. Then optimizing over the restricted set of states is equivalent to optimizing over $q \in P(\Lambda)$ as we can decompose it into these minimal projectors. This completes the proof.
\end{proof}
It is worth noting in the R\'{e}nyi entropy case this doesn't seem to simplify the calculation as we have
\begin{align*}
     \wt{D}_{\alpha}(\rho||\I_{A} \otimes \sum_{\lambda} q_{\lambda} \I_{A} \otimes \Pi_{\lambda})
    = \frac{1}{\alpha-1} \log \left\| \sum_{\lambda,\lambda'} (q_{\lambda}q_{\lambda'})^{\frac{1-\alpha}{\alpha}} \rho_{\lambda,\lambda'} \right\|^{\alpha}_{\alpha} 
\end{align*}
where we've defined $\rho_{\lambda,\lambda'}:=\I_{A} \otimes \Pi_{\lambda} \rho \I_{A} \otimes \Pi_{\lambda'}$, and noting these operators are not mutually orthogonal in general tells us we can't distribute the $\alpha/2$ as we'd hope if we expand the norm.

\subsection{Specific Cases}
We now use the above results in two ways. First, we will emphasize that for $\cK = \Pos$, Corollary \ref{corr:local-symmetry-restriction} generalizes restricting to classical registers as shown in \cite{Tomamichel-2015}. Second, we will use it to show that for CQ states, the separable min-entropy $H^{\Sep}_{\min}$ is the same as $H_{\min}$.

\paragraph*{Symmetries for the Positive Semidefinite Cone}
Corollary \ref{corr:local-symmetry-restriction} applied to $\cK = \Pos$ can be seen as a generalization of the ability to restrict classical registers in optimizations as proven in \cite{Tomamichel-2015}. This is because states that are partially classical are invariant under the pinching channel for some pre-determined computational basis. We will now show that invariance under any pinching channel $\cP$, which corresponds to some symmetry, includes classical registers as a special case. Often the pinching channel is expressed $\cP(\cdot) = \sum_{x \in \Sigma} P_{x} \cdot P_{x}$ where $P_{x}$ are mutually orthogonal projectors that sum to identity. That this is a symmetry is not obvious in this form, so we quickly verify this.
\begin{proposition}Any pinching channel can be written as a mixed unitary channel where the unitaries form a group, i.e.\ a pinching channel is a specific type of symmetry.
\end{proposition}
\begin{proof}
Suppose $|\Sigma|=r$.  Define the unitary operators
\[U_k=\sum_{x=0}^{r-1}e^{ixk/(2\pi r)} \Pi_x\qquad \text{for $k=0,\cdots,r-1$}.\]
Then
\begin{align}
   \frac{1}{r} \sum_{k=0}^{r-1}U_k\rho U_k^\dagger&=\frac{1}{r}\sum_{x,x'=0}^{r-1}\sum_{k=0}^{r-1}e^{i(x-x')k/(2\pi r)}\Pi_x\rho \Pi_{x'}\notag\}
    =\sum_{x=0}^{r-1}\Pi_x\rho\Pi_x.
\end{align}
This is clearly a group as $k = 0$ results in the identity operator, $U_{k}U_{k'} = U_{k+k'}$, and $U_{r-k}$ is the inverse of $U_{k}$ for any $k \in \{0,...,r-1\}$.
\end{proof}

\begin{corollary}
If any register of $\rho_{ABCD}$ satisfies a pinching invariance, then we can restrict to smoothed states in its ball that satisfy said invariance.
\end{corollary}

\paragraph*{AEP for Separable Min-Entropy for CQ states}

We now prove an AEP for CQ states. we start by showing that $H^{\Sep,\ve}_{\min} = H^{\ve}_{\min}$ for CQ states. This basically follows from Corollary \ref{corr:local-symmetry-restriction}. We note a non-smoothed version of this result was observed in \cite{Chitambar-2021a}.
\begin{corollary}\label{corr:restrict-to-cq-smooth-min}
    Let $\rho_{XB} \in \Density(XB)$, $\ve \in [0,1)$. Then $H_{\min}^{\ve,\Sep}(X|B)_{p_{XB}} = H^{\ve}_{\min}(X|B)_{p_{XB}}$. The same thing holds for $\rho_{AY}$.
\end{corollary}
\begin{proof}
We start with the unsmoothed case for clarity. In general, $\exp(-H_{\min}(A|B)_{\rho}) \geq \exp(-H_{\min}^{\Sep}(A|B)_{\rho})$ for $\rho \in \Density_{\leq}$ as can be seen from the primal problem (Eqn. \ref{eqn:HminKPrimal}), so if we can show the optimizer of the former's conic program is contained in the feasible set of the latter's conic program, we have equality. Consider $\rho_{XA}$. As it is invariant under pinching on the $X$ register, it is invariant under $\cG(\cdot) = \int (U_{w} \otimes \I_{B}) \cdot (U_{w} \otimes \I_{B})^{\ast} \, , dU $ where $U_{w}$ are the appropriate pinching operators. As shown in Corollary \ref{corr:local-symmetry-restriction}, the optimizer $X$ will be invariant under $\cG$. Therefore, for $\cK = \Pos$, the optimizer is CQ and thus contained in $\Sep$. This holds for any CQ. As Corollary \ref{corr:local-symmetry-restriction} shows if $\rho_{XA}$ is CQ then the optimizer $\wt{\rho}$ is CQ, regardless of the cone, $\wt{\rho}$ is the same as it's optimizer $X$. This completes the proof.
\end{proof}

We then get the following corollary that the separable cone is sufficient for the min-entropy AEP for CQ states.
\begin{corollary}\label{corr:sep-min-ent-AEP}
For $\ve \in [0,1)$, 
$$\lim_{n \to \infty}\left[\frac{1}{n}H_{\min}^{\Sep,\ve}(X^{n}|B^{n})_{\rho_{XB}^{\otimes n}}\right] = H(X|B)_{\rho_{XB}} $$
where $\rho_{XB}$ is a classical-quantum state.
\end{corollary}
\begin{proof}
By the previous corollary, $H^{\Sep,\ve}_{\min}(X^{n}|B^{n})_{\rho_{XB}} = H^{\ve}_{\min}(X^{n}|B^{n})_{\rho_{XB}}$ for all $n \in \mbb{N}$, where we have used $\rho_{XB}^{\otimes n}$ is still a CQ state. Therefore the limit holds by the standard min-entropy fully quantum AEP \eqref{eqn:FQAEP}. 
\end{proof}

\section{Separable Cone Sufficient for Classical-Quantum Information Theory}\label{sec:separable-cone-for-CQ-IT}
Corollary \ref{corr:sep-min-ent-AEP} tells us that there exists a partially quantum setting in which the separable cone is sufficient to recover the asymptotic result. We now extend this point to a sufficient number of common entropic quantities to convince ourselves that we have recovered at least point-to-point classical-quantum information theory.  

Corollary \ref{corr:HmaxK-anti-AEP} shows that the separable max-entropy cannot satisfy an AEP even for classical states, we must ask ourselves if there is an entropic quantity that can serve a similar function. Indeed, one might recall the quantum Hartley entropy $H_{0}(A|B)$ which was proposed prior to $H_{\max}$. We now define a restricted Hartley entropy that will do this. 

The Hartley entropy may be written as 
\begin{align*}
\exp(H_{0,\max}(A|B)) = \sup_{\sigma \in \Density(B)} \Tr(\Pi_{\rho}(\I_{A} \otimes \sigma_{B}))
= \|\Tr_{A}(\Pi_{\rho})\|_{\infty}
\ ,
\end{align*}
where $\Pi_{\rho}$ is defined as the projector onto the support of $\rho_{AB}$.
Note that $\Pi_{\rho}$ may be viewed as the projection onto the positive eigenspace of $\rho \succeq 0$. We could then hope we could to replace the L\"{o}wner ordering with the generalized cone ordering $\succeq_{\cK}$ (or the dual cone), i.e.\ the projector onto the support under the generalized cone ordering. This does not seem to work, so we take the minimal idea that does: the minimal projector in the cone $\cK$ that dominates $\Pi_{\rho}$. We define this as the set 
$$\{\Pi^{\cK} \succeq \Pi_{\rho} \} := \{\Pi \in \mrm{Proj} \cap \cK : \Pi \succeq \Pi_{\rho} \} \ . $$
Note in the case $\cK = \Pos$, this set in general contains projectors other than $\Pi_{\rho}$. 

We then use this set to define the conic Hartley entropy.
\begin{align}\label{eqn:restricted-H0max}
\exp(H^{\cK}_{0,\max}(A|B)) := \underset{\Pi \in \{\Pi^{\cK} \succeq \Pi_{\rho} \}}{\min} \| \tr_{A}(\Pi_{AB})\|_{\infty} \ . 
\end{align}
This is the definition we want as we now argue by showing that it satisfies desirable properties.
To do this, we first introduce the $CQ$ cone, which we in effect used in Corollary \ref{corr:restrict-to-cq-smooth-min}.
\begin{definition}\label{def:CQ-Cone}
Fix a basis of $A$, $\{\ket{x}\}_{\cX}$. Define $\cG(\cdot) := \sum_{x \in \cX} \dyad{x} \cdot \dyad{x}$. The CQ cone is
$$ \cK := \{ P \in \Pos(A \otimes B) : (\cG \otimes \id_{B})(P) = P \} \ . $$
\end{definition}
This cone is clearly generated by the CQ quantum states, which is why we call it the CQ cone. Moreover, note that by definition of restricted Hartley entropy it is always an upper bound on $H_{0,\max}^{\Sep}$ as every projector in the CQ cone is in the separable cone.\footnote{More generally if $\cK_{1} \subseteq \cK_{2}, H_{0,\max}^{\cK_{1}} \geq H_{0,\max}^{\cK_{2}}$.} This will be a useful tool in proving an AEP on CQ states over the restricted cone. We now verify we have chosen a good choice of definition for restricted Hartley entropy along with similar properties.

\begin{proposition}\label{prop:Hartley-equivs}
\begin{enumerate}
    \item Given $\rho_{AB}$, then $H^{\Pos}_{0,\max}(A|B)_{\rho} = H_{0,\max}(A|B)_{\rho}$.
    \item Given CQ state $\rho_{XB}$, $$H^{\CQ}_{0,\max}(X|B) = H^{\Sep}_{0,\max}(X|B)_{\rho} = H_{0,\max}(X|B)_{\rho} \ , $$ where the CQ cone is defined in the basis of $\rho_{X}$. 
\end{enumerate}
\end{proposition}
\begin{proof}
We group these together because they have the same proof structure. We argue from the exponential of the entropy and then since the logarithm is a monotonic function, it holds for the entropies themselves. For Item 1, let $\cK = \Pos$. Then $\Pi_{\rho}$ is always feasible and strictly satisfies the constraint, so anything else feasible must decompose as $\Pi_{\rho} + X$ where $X \in \Pos$, which can only be worse by linearity of the partial trace and convexity of the norm. For Item 2, let $\rho_{XB}$. Then $\Pi_{\rho} \in \{\Pi^{\CQ} \succeq \Pi_{\rho} \} \subset \{\Pi^{\Sep} \succeq \Pi_{\rho}\}$ and by the same argument as above, it is also optimal.
\end{proof}
Having verified that we recover Hartley entropy when we set $\cK = \Pos$ or when the state is CQ and the cone contains the CQ states, we move on to show that it satisfies an asymptotic equipartition property. The converse will be effectively trivial. The achievability will rely on the fact that for CQ states, $H_{0,\max}^{\CQ}(X|B) = H_{0,\max}(X|B)$, so we can use a preserved duality relation \cite[Duality Relation 3]{Tomamichel-2015} that then allows for an approximate equivalence with the unrestricted max-entropy \cite{Tomamichel-2011a} that satisfies AEP \cite{Tomamichel-2015}. We split this into lemmas for clarity.

We begin by defining an intermediary entropy.
Let 
\begin{align}\label{eq:down-arrow-min-ent}
H^{\downarrow}_{\min}(A|B)_{\rho} := -D_{\max}(\rho_{AB}||\I_{A} \otimes \rho_{B}) \ . 
\end{align}
This entropy is known to satisfy the duality relation \cite{Tomamichel-2015,Tomamichel-2011a}:
\begin{align}\label{eq:petz-renyi-dual} H_{0,\max}(A|B)_{\rho} = -H_{\min}^{\downarrow}(A|C)_{\rho} \ , 
\end{align}
where $\rho_{ABC}$ is pure. Next we define the smoothed restricted Hartley entropy and the smoothed $H^{\downarrow}_{\min}$:
\begin{align} H_{0,\max}^{\cK,\ve}(A|B)_{\rho} &:= \min_{\wt{\rho} \in \Bve(\rho)} H_{0,\max}^{\cK}(A|B)_{\wt{\rho}}  \label{eq:smooth-hartley} \\
H_{\min}^{\downarrow,\ve}(A|B)_{\rho} &:= \max_{\wt{\rho} \in \Bve(\rho)} H_{\min}^{\downarrow,\ve}(A|B)_{\wt{\rho}} \nonumber
\end{align}
It is known that when $\cK = \Pos$, both quantities are invariant under local embeddings even when smoothed \cite[Lemma 21]{Tomamichel-2011a}, which we will use soon. Moreover, it follows $H^{\CQ,\ve}_{0,\max}(X|B)$ is invariant under local embeddings by the previous proposition.

Next we prove an unsurprising but special property of the CQ-Hartley entropy.
\begin{proposition}
Let $\rho_{AB}$ be an arbitrary state. Then
$$ H_{0,\max}^{\CQ}(A|B)_{\rho} =  H_{0,\max}(A|B)_{(\cG \otimes \id)(\rho)} \ ,$$
where $\cG$ is dephasing according to the $\CQ$ cone.
\end{proposition}
\begin{proof}
Define the map $\Phi_{x'}(\cdot) := (\bra{x'} \otimes I_{B}) \cdot (\ket{x'} \otimes I_{B})$ for all $x' \in \cX$. Define $p(x)\rho_{B}^{x} := \Phi_{x}(\rho_{AB})$, where if $p(x) = 0$ then let $\rho_{B}^{x}$ be the zero matrix. Let $\Pi_{\rho^{x}}$ be the projector onto the support of $\rho_{B}^{x}$. Then $\Pi_{\rho^{x}} - \rho^{x}_{B} \succeq 0$. Define $\wt{\Pi}_{x}$ as any projector defined such that $\Pi_{\rho^{x}} \succ \wt{\Pi}_{x}$. Then by definition of support, $\wt{\Pi}_{x} - \rho_{B}^{x} \not \succeq 0$. As $\Phi_{x'}(\sum_{x} \dyad{x} \otimes \wt{\Pi}_{x} - \rho_{AB}) = \dyad{x'} \otimes (\wt{\Pi}_{x'} - \rho_{B}^{x'}) \not \succeq 0$ and $\Phi_{x'}$ is always a CP map, we can conclude we cannot remove elements from the $\Pi_{x}$. However, $\sum_{x} \dyad{x} \otimes \Pi_{\rho^{x}}$ is feasible. Thus, this is the minimal projector for $\CQ$ cone as removing anything from it will become infeasible. Moreover, $\sum_{x} \dyad{x} \otimes \Pi_{\rho^{x}}$ is the projector onto the support of $\sum_{x} \Phi_{x}(\rho_{AB}) = (\cG \otimes \id)(\rho)$. This completes the proof.
\end{proof}

We now convert this into a lemma about duality. We note we could in principle avoid stating this explicitly, but it makes it clear that why the AEP holds is that there is a duality between the restricted cone and an unrestricted cone.
\begin{lemma}
Let $\rho_{XBC}$ be pure where it is classical on $X$ and $\ve \in [0,1)$. Then
$$ H_{0,\max}^{\CQ,\ve}(X|B)_{\rho} = -H_{\min}^{\downarrow,\ve}(X|C)_{\rho} $$
\end{lemma}
\begin{proof}
We effectively follow the proof of duality for smooth entropies from \cite{Tomamichel-2015} just with slightly altered entropies.

As already noted, in the relevant setting both spaces are invariant under local embeddings, so we may assume the spaces $B$ and $C$ are sufficiently large to hold the purifications of the optimal smoothed states. This means we can restrict to pure states and we will use $\Bve_{\ast}(\rho) := \{\tau \in \Bve(\rho): \rank(\tau) = 1 \}$ to denote such a restriction. Therefore, we have
\begin{align*}
    H_{0,\max}^{\CQ,\ve}(X|B)_{\rho} &= \min_{\wt{\rho} \in \Bve_{\ast}(\rho_{XBC})} H_{0,\max}(X|B)_{(\cG \otimes \id_{BC})(\wt{\rho})} \\
    &=  \min_{\wt{\rho} \in \Bve_{\ast}(\rho_{XBC})} - H_{\min}^{\downarrow}(X|C)_{(\cG \otimes \id_{BC})(\wt{\rho})} \\
    &\geq \min_{\wt{\rho} \in \Bve(\rho_{XC})} - H_{\min}^{\downarrow}(X|C)_{(\cG \otimes \id_{C})(\wt{\rho})} \\
    &= - H_{\min}^{\downarrow,\ve}(X|C)_{\rho} \ ,
\end{align*}
where the second equality uses the duality \eqref{eq:petz-renyi-dual}, the inequality is removing the pure state restriction, and the last equality is pushing the negative through the optimization. 

For the other direction, note by DPI of max-relative entropy and the purified distance, $H^{\downarrow,\ve}_{\min}(X|C)_{\rho}$ is always achieved by a CQ state, so we can apply $\cG$ without loss of generality. Therefore it's the same idea in backwards, but we provide it for completeness.
\begin{align*}
    - H_{\min}^{\downarrow,\ve}(X|C)_{\rho} &= - \max_{\wt{\rho} \in \Bve_{\ast}(\rho_{XBC})} H_{\min}^{\downarrow}(X|C)_{(\cG \otimes \id_{BC})(\wt{\rho})} \\
    &= \min_{\wt{\rho} \in \Bve_{\ast}(\rho_{XBC})} H_{0,\max}(X|B)_{(\cG \otimes \id)(\rho)} \\
    &= \min_{\wt{\rho} \in \Bve_{\ast}(\rho_{XBC})} H_{0,\max}^{\CQ}(X|B)_{(\cG \otimes \id)(\rho)}  \\
    &\geq \min_{\wt{\rho} \in \Bve(\rho_{XB})}
    H_{0,\max}^{\CQ}(X|B)_{(\cG \otimes \id)(\rho)} \ .
\end{align*}
This completes the proof.
\end{proof}

Next we can use this to get bounds on $H_{0,\max}^{\CQ,\ve}$ in terms of $H_{\max}$ via this duality. This is a direct consequence of \cite[Lemma 18]{Tomamichel-2011a}.
\begin{proposition}\label{prop:CQ-Hartley-bound}
Let $\ve > 0$, $\ve' \geq 0$ and $\rho_{XB} \in \Density(X \otimes B)$. Then, 
$$ H_{\max}^{\ve}(X|B)_{\rho} + \log(c) \geq H_{0,\max}^{\CQ,\ve+\ve'}(X|B)_{\rho} \geq H_{\max}^{\ve+\ve'}(X|B)_{\rho} \ , $$
where $c = 2/\ve^{2} + 1/(1-\ve')$.
\end{proposition}
\begin{proof}
By \cite[Lemma 18]{Tomamichel-2011a}, 
\begin{align*}
     -H_{\min}^{\ve}(X|B)_{\rho} + \log(c) \geq -H_{\min}^{\ve}(X|B)_{\rho}  \geq -H_{\min}^{\ve+\ve'}(X|B)_{\rho} \ .
\end{align*}
By the previous lemma, $H_{0,\max}^{\CQ,\ve}(X|B)_{\rho} = -H_{\min}^{\downarrow,\ve}(X|C)_{\rho}$.
Using the duality relation $H_{\min}^{\ve}(A|B) = -H_{\max}^{\ve}(A|C)$ \cite{Konig-2009a,Tomamichel-2015} completes the proof.
\end{proof}

Finally we present the theorem.
\begin{theorem}\label{thm:Sep-Hartley-AEP}
Let $\ve \in (0,1)$. Let $\rho_{XB}$ and let the CQ cone be defined according to $X$. Consider any cone $\CQ \subseteq \cK \subseteq \Pos$.
\begin{align*}
    \lim_{n \to \infty} \left[ \frac{1}{n} H_{0,\max}^{\cK,\ve}(X^{n}|B^{n})_{\rho^{\otimes n}} \right] = H(X|B)_{\rho_{XB}}
\end{align*}
\end{theorem}
\begin{proof}
(Converse) By \cite[Lemma 18]{Tomamichel-2011a}, 
$H_{\max}^{\ve}(X^{n}|B^{n})_{\rho^{\otimes n}} \leq H_{0,\max}^{\ve}(X^{n}|B^{n})_{\rho^{\otimes n}}$ for all $n \in \mbb{N}$. Moreover, for $\cK \subset \Pos$, $H_{0,\max}^{\ve}(X^{n}|B^{n})_{\rho^{\otimes n}} \leq H_{0,\max}^{\cK,\ve}(X^{n}|B^{n})_{\rho^{\otimes n}}$ for all $n \in \mbb{N}$. Then using the converse for max entropy \cite{Tomamichel-2015}, we have 
\begin{align*}
H(X|B) \leq \lim_{n \to \infty}\left[\frac{1}{n} H^{\ve}_{\max}(X^{n}|B^{n})_{\rho^{\otimes n}} \right]
\leq \lim_{n \to \infty}\left[\frac{1}{n}H^{\cK,\ve}_{0,\max}(X^{n}|B^{n})_{\rho^{\otimes n}} \right] \ .
\end{align*}
(Direct) For any cone $\CQ \subseteq \cK \subseteq \Pos$, using Proposition \ref{prop:CQ-Hartley-bound} with the choice of $\ve' = 0$, 
\begin{align*}
H_{0,\max}^{\cK,\ve}(X^{n}|B^{n})_{\rho^{\otimes n}} 
&\leq H_{0,\max}^{\CQ,\ve}(X^{n}|B^{n})_{\rho^{\otimes n}}
\leq H^{\ve}_{\max}(X^{n}|B^{n})_{\rho^{\otimes n}} + \log(c) 
\end{align*}
for any $n \in \mbb{N}$. Then using the direct part of the AEP \cite{Tomamichel-2009a,Tomamichel-2015}, we have 
\begin{align*}
\lim_{n \to \infty} \left[ \frac{1}{n}H_{0,\max}^{\cK,\ve}(X^{n}|B^{n})_{\rho^{\otimes n}} \right]
\leq &  \lim_{n \to \infty} \left[ H_{0,\max}^{\CQ,\ve}(X^{n}|B^{n})_{\rho^{\otimes n}} \right] \\
\leq & \lim_{n \to \infty} \left[ \frac{1}{n}H^{\ve}_{\max}(X^{n}|B^{n})_{\rho^{\otimes n}}\right] + \lim_{n \to \infty} \log(c)/n \\
\leq & H(X|B) \ ,
\end{align*}
where in the last line we used the correction term goes to zero and the direct part of the AEP. This completes the proof.
\end{proof}

\begin{corollary}
As we have chosen to define them so as to recover the proper measure when $\cK = \Pos$, the equivalence of $H_{0,\max}^{\cK}$ and $H_{max}^{\cK}$ asymptotically does not hold when one considers a sequence of cones $\{\cK^{(n)}\}$ such that $\cK^{(n)} \subseteq \Ent_{f(n)}(A^{n}:B^{n})$ where $f(n) \leq O(c^{n})$ where $c < 1$. Moreover, this implies they aren't logarithmically equivalent up to smoothing.
\end{corollary}
\begin{proof}
Using Corollary \ref{corr:HmaxK-anti-AEP} and Theorem \ref{thm:Sep-Hartley-AEP}, we have for $p_{XZ} = \pi_{X} \otimes \dyad{z}$ 
\begin{align*}
\lim_{n \to \infty}\left[\frac{1}{n} H^{\cK^{(n)}}_{\max}(X^{n}|Z^{n})\right] <& H(X|Z)_{\pi_{X} \otimes \dyad{z}}
= \lim_{n \to \infty}\left[\frac{1}{n} H^{\cK^{(n)}}_{0,\max}(X^{n}|Z^{n})\right] \ ,
\end{align*}
where in the last equality we have used that this is achieved by the separable cone and the positive cone and thus must be achieved by all cones $\Sep \subset \cK \subset \Pos$.
The moreover statement then follows because if they were, then the strict inequality could not hold.
\end{proof}
What the above tells us is there is an interesting thing that happens. In proper quantum information theory it is often sufficient to restrict ourselves to Sandwiched R\'{e}nyi divergences, which include $H_{\max},H_{\min}$ because we can approximate $H_{0,\max}$. In restricted theories, specifically the separable theory, we cannot do this.

In case it is not clear, we convince ourselves the reason this worked out is because $H_{0,\max}^{\CQ}$ can characterize the optimal zero-error compression for classical data with quantum side information \cite{Renes-2012a}. We present the fully classical case for clarity, although with quantum side information must also hold by the equivalence. It suffices to show achievability because a restricted theory can only do worse. Consider $p_{XY}$ where Alice (sender) has $X$ and Bob (receiver) has $Y$. If Bob has $y \in \cY$, he only needs to distinguish between $x \in \cX_{|y}$ where $\cX_{|y}$ is the support of $p(X|Y=y)$, i.e.\ the values $X$ can take conditioned on $Y = y$. 
Then if there is no error, clearly Alice needs an alphabet $\cZ$ such that $|\cZ| = \max_{y} |\cX_{|y}|$ so long as they have a pre-agreed upon codebook. Then Alice only needs to send $\log|\cZ|$ bits. So we verify that this measure captures that. As explained, the optimizer is $\Pi_{p_{XY}}$ if one writes $p_{XY}$ as a density matrix, because for a CQ state, $H_{0,\max}^{\CQ} = H_{0,\max}$. Then a direct calculation determines:
\begin{align*}
\Pi_{p_{XY}} =& \sum_{x,y} \mathds{1}\{p(x,y) \geq 0\} E_{x,x} \otimes E_{y,y} \\
=& \sum_{y} E_{y,y} \otimes \sum_{x} \mathds{1}\{p(x,y) \geq 0\} E_{x,x} \\
=& \sum_{y} E_{y,y} \otimes |\cX_{|y}| E_{x,x} \ ,
\end{align*}
where $\mathds{1}\{A\}$ is the indicator function for event $A$. It follows from this that
$$\|\Tr_{X}(\Pi_{p_{XY}})\|_{\infty} = \|\sum_{y} |\cX_{|y}| E_{y,y}\|_{\infty} = \max_{y} |\cX_{|y}| \ ,$$
which is exactly what we wanted. 

\paragraph*{Generalizing to Other Measures}
To fully justify the claim that a framework for classical-quantum states only requires the CQ or separable cone, we briefly touch on other useful entropic measures in the restricted framework.

The information spectrum method of Han and Verd\'{u} \cite{Han-1993a,Han-2003a} is a common tool for classical information theory and was generalized to the quantum setting in \cite{Tomamichel-2013a}. Much like the quantum Hartley entropy, it is defined in terms of a projector onto a positive eigenspace, albeit a more elaborate one. Given $\rho \in \Density_{\leq}(A)$, $Q \in \Pos(A)$, the standard $\ve$-information spectrum \cite{Tomamichel-2013a} is defined as
$$D^{\ve}_{s}(\rho||Q) := \sup\{\gamma \in \mbb{R}\,:\, \langle \rho , \{\rho \preceq 2^{\gamma}Q\} \rangle \leq \ve \} \ ,$$
where $\{\rho \preceq 2^{\gamma}Q\}$ is the projector onto the positive eigenspace of $(2^{\gamma}Q - \rho)$. Note in general this is distinct from the projector onto its support which is why it has distinct notation.

To extend this to generalized cones, we follow what we did for the Hartley entropy and define the set of projectors contained in the cone that also dominate this relevant projector:
\begin{align*}
 \{\Pi^{\cK} \succeq \{\rho \preceq 2^{\gamma}Q\}\}  := \{\Pi \in \mrm{Proj}\, \cap \, \cK : \Pi \succeq \{\rho \preceq 2^{\gamma}Q\} \} \ .
\end{align*}

The $\ve,\cK$-information spectrum divergence may then be defined as
\begin{align*}
    D_{s}^{\ve,\cK}(\rho||Q) = \sup\{\gamma \in \mbb{R} | \underset{\Pi \in \{\Pi^{\cK} \succeq \{\rho \preceq 2^{\gamma}Q\}\}}{\min} \langle \rho , \Pi \rangle \leq \ve \} \ . 
\end{align*}
This is a min-max and so clearly is not particularly feasible to deal with in general. Just as for the Hartley entropy (Proposition \ref{prop:Hartley-equivs}), in the case $\cK = \Pos$, the optimal $\Pi$ is the projector we want to recover as by definition of the set we are minimizing over, any other projector must be of the form $X + \{2^{\gamma}\sigma \succeq \rho\}$ where $X \geq 0$, so we have
\begin{align*}
\langle \rho , X + \{2^{\gamma}\sigma \succeq \rho\} \rangle = \langle \rho, \{2^{\gamma}\sigma \succeq \rho\} \rangle + \langle \rho, X \rangle 
\geq  \langle \rho, \{2^{\gamma}\sigma \succeq \rho\} \rangle \ . 
\end{align*}
By the same argument as above, if $\cK = \Sep$ and \textit{both} $\rho$ and $Q$ are classical-quantum, then the optimizer is the same as in the unrestricted case. Moreover it is easy to show the unrestricted case reduces to the classical notion of information spectrum if $\rho$ and $Q$ classical. Therefore, the $\ve,\Sep$-information spectrum divergence at least mathematically captures the partially quantum setting.

Given the information spectrum divergence's close relationship to the hypothesis testing divergence, we note the $\ve,\cK$-hypothesis testing divergence may be defined using \eqref{eqn:hyp-test-defn} with the additional restriction that $X \in \cK$. This restricted entropy was studied in great detail in \cite{Matthews-2014a}. For our focus on arguing the separable cone captures partially quantum information theory, we state the following unsurprising proposition that in this setting an AEP holds for the cone-restricted hypothesis-testing divergence, which is one way of proving the quantum Stein lemma holds for CQ states
with the separable cone.
\begin{proposition}\label{prop:sep-one-shot}
$D^{\ve,\Sep}_{h}(P||Q) = D^{\ve}_{h}(P||Q)$ if both $P,Q$ are CQ with respect to the same partition. Moreover, in this setting an AEP holds.
\end{proposition}
\begin{proof}
Following the same idea as in the proof of Corollary  $\ref{corr:local-symmetry-restriction}$, since $P$ or $Q$ is CQ, without loss of generality the optimal decision function $X$ is CQ and thus in the separable cone. Since many copies of a CQ state is still CQ, the AEP will hold. 
\end{proof}

While one cannot prove that all partially quantum information theory is captured by the separable cone without a, presumably never-ending, exhaustive set of proofs for each operational setting, we have shown that the separable cone is sufficient for the common entropic measures on classical-quantum states and even is able to recover the asymptotic results via AEPs.

As a final remark, we note that for both partially quantum and fully classical information theory, one does not even need the separable cone. Rather, under the assumption that classical information theory has a preferred basis, it is clear that the cone that is incoherent in the preferred bases is sufficient for partially quantum and classical information theory as one would expect. This is ultimately formalized in Theorem \ref{thm:information-theory-with-symmetries}. Regardless, neither of these cones will be closed under local unitaries, and so we may conclude that the separable cone is the minimal cone closed under local basis changes. We summarize Sections \ref{sec:anti-stein-lemma}-\ref{sec:separable-cone-for-CQ-IT} in in Fig. \ref{fig:Cone-AEP-Summary}

\begin{figure}[H]
    \centering
    \includegraphics[width=0.6\columnwidth]{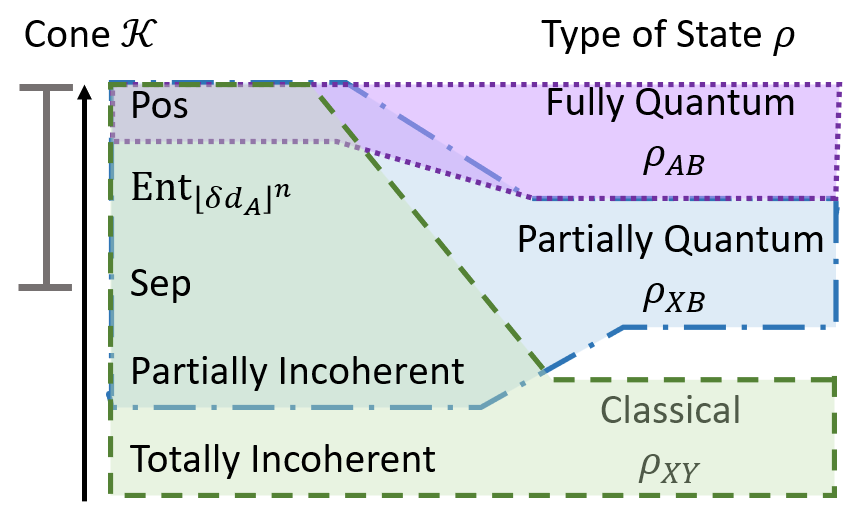}
    \caption{Diagram summarizing under what cones the restricted measures recover the asymptotic results for various types of bipartite quantum states. The grey bar denotes the space of local unitary-invariant cones. $\mrm{Ent}_{\lfloor \delta d_{A} \rfloor^{n}}$ represents that even if the entanglement rank cone grows exponentially in a fraction of $d_{A}$ captured by $\delta < 1$, then the AEP won't hold (Theorem \ref{thm:anti-AEP}).}
    \label{fig:Cone-AEP-Summary}
\end{figure}

\section{An Anti-Standard AEP for Extended Conditional Min-Entropy}\label{sec:extended-min-entropy}
Before moving on to the operational results of this framework, we show that this general approach can be extended for superchannels. Specifically, we consider the recent extension of the min-entropy known as the extended conditional min-entropy \cite{Gour-2019a}. It is defined in such a manner as to mirror \eqref{eqn:min-ent-cptp-support} but for superchannels, where, physically, a superchannel takes one quantum channel to another. To explain this in detail, we now summarize the mathematical structure for such objects largely worked out in \cite{Chiribella-2008a,Gour-2019a}. For the reader unfamiliar with this literature, we suggest noting the strong parallels to the quantum channel acting on a quantum state and a superchannel acting on a channel. For an extended min-entropy, it is this parallel that we ultimately utilize.

Consider the sets of transformations $\Trans(A_{0},A_{1})$ and $\Trans(B_{0},B_{1})$. A supermap is a linear map $\Theta: \Trans(A_{0},A_{1}) \to \Trans(B_{0},B_{1})$. We denote the space of supermaps by $\mbb{L}(A,B)$ where in this section we let $A := (A_{0},A_{1})$ and $B:=(B_{0},B_{1})$ for notational simplicity. Let $\mathds{1}_{R}$ be the identity supermap (i.e.\ $\mathds{1}_{R}(\Psi) = \Psi$ for all $\Psi \in \Trans(R_{0},R_{1})$). A superchannel $\Theta \in \mbb{L}(A,B)$ is a supermap that is completely CP-preserving (CPP), $(\mathds{1}_{R} \otimes \Theta)(\Psi) \in \CP(R_{0} \otimes B_{0},R_{1} \otimes B_{1})$ for all $\Psi \in \CP(R_{0} \otimes A_{0},R_{1} \otimes A_{1})$ and all choices of $(R_{0},R_{1})$, and TP-preserving, $\Theta(\Psi) \in \mrm{TP}(B_{0},B_{1})$ for all $\Psi \in \mrm{TP}(A_{0},A_{1})$. We denote the space of superchannels by $\mbb{S}(A,B)$.

With some work, one can show that the space of supermaps enjoy their own version of the Choi-Jamiolkowski isomorphism in terms of the original. First, without loss of generality, let $A_{0} \cong \mbb{C}^{d}$, $A_{1} \cong \mbb{C}^{d'}$. Consider 
$$ \cE_{a_{0}a_{1}}^{A_{0} \to A_{1}}(\rho):=  \bra{i}\rho\ket{j} \, \ket{k}\bra{l}^{A_{1}} \ ,$$
where $a_{0} \equiv (i,j) \in [d]^{\times 2}, a_{1} \equiv (k,l) \in [d']^{\times 2}$. Then $\{\cE_{a_{0}a_{1}}^{A_{0}\to A_{1}}\}$ form a basis for $\Lin(A_{0},A_{1})$.
Now let $\Theta \in \mbb{L}(A,B)$. Then $\mbf{J}^{AB}_{\Theta} \in \Lin(A_{0}\otimes A_{1} \otimes B_{0} \otimes B_{1})$ is given by
\begin{equation}\label{eqn:supermap-choi}
\begin{aligned}
    \mbf{J}^{AB} :=& \sum_{a_{0},a_{1}} J^{A_{0} \to A_{1}}_{\cE_{a_{0},a_{1}}} \otimes J^{B_{0} \to B_{1}}_{\Theta[\cE_{a_{0}a_{1}}]} \\
    =&(\id_{A_{0}B_{0}} \otimes (\mathds{1}_{A} \otimes \Theta)[\Upsilon^{A\wt{A}}])(\phi^{+}_{A_{0}\wt{A}_{0}} \otimes \phi^{+}_{B_{0}\wt{B}_{0}})\ ,
\end{aligned}
\end{equation}
where $J^{B_{0} \to B_{1}}_{\Theta[\cE_{a_{0}a_{1}}]}$ is the Choi operator of $\Theta[\cE_{a_{0}a_{1}}]$ and $\Upsilon^{A\wt{A}} = \sum_{a_{0},a_{1}} \cE^{A_{0} \to A_{1}}_{a_{0},a_{1}} \otimes \cE^{\wt{A}_{0} \to \wt{A}_{1}}_{a_{0},a_{1}}$, which may be seen as the channel equivalent of the maximally entangled state. Its action is given by $\Upsilon^{A\wt{A}}(\rho_{A_{0}\wt{A_{0}}}) = \Tr(\rho \phi^{+}_{A_{0}\wt{A}_{0}})\phi^{+}_{A_{1}\wt{A}_{1}}$ The equivalence of these two equations just follows from the standard Choi-Jamiolkowski formalism and linearity. The advantage of this second expression is that it shows the map $\Theta \to (\mathds{1}_{A} \otimes \Theta)(\Upsilon^{A\wt{A}})$ is an isomorphism between the space of supermaps and the space of bipartite maps. The other advantage of this Choi-Jamiolkowski isomorphism is that it allows conditions on $\Theta$ to be in terms of $\mbf{J}_{\Theta}$. Specifically,
a supermap $\Theta \in \mbb{L}(A,B)$ is a superchannel if and only if $\mbf{J}_{\Theta}^{AB} \geq 0$ such that $\mbf{J}^{A_{1}B_{0}}_{\Theta} = \I_{A_{1}B_{0}}$ and $\mbf{J}^{AB_{0}}_{\Theta} = \mbf{J}_{\Theta}^{A_{0}B_{0}} \otimes d_{A_{1}}^{-1}\I_{A_{1}}$ \cite{Chiribella-2008a}. Note that all of these constraints fit into the framework of semidefinite programming.

Given this, one may wish to define a min-entropy-like quantity for superchannels.  Indeed, this was done in \cite{Gour-2019a}. We will take one of the equivalent expressions presented in \cite{Gour-2019a} as the definition as it seems more fundamentally motivated for this work. The only change this results in is a change in the SDP's dual program, which will not be relevant here. For a bipartite channel $\Phi \in \Channel(A_{0}B_{0},A_{1}B_{1})$, the extended min-entropy is given by:
\begin{align}\label{eq:superchannel-Choi-overlap}
     d_{B_{0}}\exp(-H^{\ext}_{\min}(B|A)_{\Phi})
    :=& \underset{\Theta \in \mbb{S}^{AB}}{\max} \langle \mbf{J}^{AB}_{\Theta}, J^{AB}_{\Phi} \rangle \ ,
\end{align}
Note that up to the $d_{B_{0}}$, this aligns with \eqref{eqn:min-ent-cptp-support}. Using $\Theta \to (\mathds{1}_{A} \otimes \Theta)(\Upsilon^{A\wt{A}})$ is an isomorphism and their choice of inner product on $\Trans(A,B)$, \cite{Gour-2019a} showed that $d_{B_{0}}\exp(-H^{\ext}_{\min}(B|A)_{\Phi})$ may be expressed as
\begin{align}\label{eqn:entanglement-super-channel}
    d_{B}\underset{\Theta \in \mbb{S}^{AB}}{\max} \langle (\Theta \otimes \mathds{1}_{B})[\Omega^{AB}](\tau_{B_{0}}), \tau_{B_{1}} \rangle \ ,
\end{align}
which we can see as the achievable singlet fraction of passing of a bipartite channel where one half undergoes a super channel. It was already noted in \cite{Gour-2019a} that \eqref{eqn:entanglement-super-channel} is an optimization similar to that of \eqref{eq:min-entropy-singlet-fraction}, but over superchannels.  

Alternatively, using the Choi operator characterization of a superchannel along with \eqref{eq:superchannel-Choi-overlap}, one recovers the SDP for $d_{B_{0}}\exp(-H^{\ext}_{\min}(B|A)_{\Phi})$
\begin{center}
    \begin{miniproblem}{0.6}
      \emph{Primal}\\[-5mm]
      \begin{equation}
      \begin{aligned}\label{eqn:extended-min-entropy-SDP}
        \text{max.:}\quad & \ip{X}{J_{\Phi}^{AB}} \\
        \text{s.t.}\quad & X_{A_{1}B_{0}} = I_{A_{1}B_{0}} \\
        & X_{AB_{0}} = X_{A_{0}B_{0}} \otimes d_{A_{1}}^{-1} I_{A_{1}} \\
        & X \in \Pos(A \otimes B) \ ,
      \end{aligned}
      \end{equation}
    \end{miniproblem}
\end{center}
where we have expressed the marginals of $X$ via labeling the spaces rather than denoting the partial trace maps.
It then seems natural for
$ d_{A_{0}}\exp(-H^{\ext}_{\min}(A|B)_{\Phi})$ to be the above SDP where we keep the same constraints but require $X$ to be restricted to a cone $\cK$. Indeed, we can show this properly parallels the min-entropy case as follows. It has been shown in \cite{Gour-2019a} that the adjoint supermap, referred to as the dual map, satisfies the identity 
\begin{align}\label{eqn:superchannel-dual-map-choi}
\mbf{J}^{BA}_{\Theta^{\ast}} = \overline{\mrm{SWAP}(\mbf{J}^{AB}_{\Theta})}
\end{align}
where the $A$ and $B$ spaces are swapped. Noting the conjugate mapping preserves the positive semidefiniteness, flipping the $A$ and $B$ labels in \eqref{eqn:extended-min-entropy-SDP} means we are optimizing over the set of supermaps $\mbb{L}(A,B)$ that are duals of superchannels $\mbb{S}(B,A)$. We will denote this set as $\wt{\mbb{S}}^{AB}$ in which case we may write
\begin{align}\label{eqn:ext-min-ent-flipped-singlet-frac}
    d_{A_{0}}\exp(-H^{\ext}_{\min}(A|B)_{\Phi}) = \max_{\Theta \in \wt{\mbb{S}}^{AB}} \langle \mbf{J}^{BA}_{\Theta},J^{AB}_{\Phi}\rangle \ ,
\end{align}
which is the equivalent of \eqref{eqn:min-ent-unital-support}. We do note however that this set is not the set of unital-preserving CPP maps in the way that for channels it is the set of CPU maps (See Corollary \ref{corr:set-of-dual-maps} in the appendix). By the same type of argument as for obtaining \eqref{eqn:entanglement-super-channel}, one obtains
\begin{align}\label{eqn:entanglement-super-unital}
    d_{A}\underset{\Theta \in \mbb{S}^{BA}}{\max} \langle \tau_{A_{1}}, (\mathds{1}_{A} \otimes \Theta)[\Phi^{AB}](\tau^{A_{0}}) \rangle \ ,
\end{align}
where we stress the optimization is over the space of superchannels. The proof of this is included in the Appendix.  Note this looks almost identical to \eqref{eq:min-entropy-singlet-fraction}. 

With all of this established, we can discuss the asymptotic equipartition property. To date, an operational interpretation of $\lim_{\ve \to 0} \lim_{n \to \infty} \frac{1}{n} H^{\ext,\ve}_{\min}(B|A)$ is not known, but it is likely to be discovered, so it is worthwhile to show it requires the L\"{o}wner order. It is clear that we could have changed the semidefinite program to be restricted to a cone, so we omit writing this out explicitly. Then all that is necessary is a set of cones which to interpolate between. To do this, we again choose to interpolate in terms of entanglement. For these purposes we first introduce the notion of entanglement rank for CP maps which interpolate between the well-known space of separable CP maps and arbitrary bipartite CP maps.
\begin{definition}\label{Defn:k-ent-gen}
Let $\Phi \in \CP(A_{0}B_{0},A_{1}B_{1})$. We say the map is $k-$Entanglement-Generating $\Ent_{k}\CP(A:B)$ if $J_{\Phi} \in \mrm{Ent}_{k}(A:B)$.
\end{definition}
It is easy to see if $k=1$, then one recovers the space of separable maps as $\Ent_{1}(A:B) = \Sep(A:B)$. Likewise, if $k=r':=\min\{A,B\}$, then one recovers the space of completely positive bipartite maps as $Ent_{r'}(A:B) = \Pos(A \otimes B)$. An operational interpretation for these maps was given in \cite{Johnston-2012a} where it was shown that these maps can only entangle Alice and Bob's environments so much by acting on the target registers if the initial state is separable between Alice and Bob.
\begin{proposition}\label{prop:k-ent-generating}(\cite[Theorem 2.5.7]{Johnston-2012a})
$\Phi \in \Ent_{k}\CP(A:B)$ if and only if $(\id_{\wt{A}\wt{B}} \otimes \Phi^{AB})(P) \in \Ent_{k}(\wt{A}A_{1}:\wt{B}B_{1})$ for all $P \in \Sep(\wt{A}A_{0}:\wt{B}B_{0})$.
\end{proposition}

We now define the set of supermaps that take all joint CP maps to $k-$Entanglement-Generating maps.
\begin{definition}\label{defn:k-ent-break}
A CPP supermap $\Theta \in \mbb{L}(A,B)$ is a $k-$entanglement-breaking supermap if $(\mathds{1}_{R} \otimes \Theta)(\Phi_{RA}) \in \Ent_{k}\CP(R:B)$ for all $\Phi_{RA} \in \CP(R_{0}A_{0},R_{1}A_{1})$. Moreover, if $\Theta$ is a superchannel, we call it a $k-$entanglement-breaking superchannel. We denote the set of $k-$entanglement-breaking superchannels $\EntS_{k}(A,B)$.
\end{definition}
We note that this is a generalization of the definition of entanglement-breaking supermap introduced in \cite{Chen-2020a}. Moreover, we can extend the characterization of EB supermaps in \cite[Theorem 1]{Chen-2020a} to the general case. We note the map $\Delta_{\Theta} \in \Trans(A,B)$ is the map with the same Choi operator as $\mbf{J}^{AB}_{\Theta}$ and $J_{\Theta[\Phi]} = \Delta_{\Theta}(J_{\Phi}^{A})$ (See Appendix for further details).
\begin{theorem}\label{thm:k-ent-breaking-supermaps}
Let $\Theta \in \mbb{L}(A,B)$ be a CPP supermap. The following are equivalent:
\begin{enumerate}
    \item $\Theta_{AB}$ is a $k-$entanglement-breaking supermap
    \item $(\mathds{1}_{\wt{A}} \otimes \Theta^{AB})[\Upsilon_{+}^{\wt{A}A}] \in \Ent_{k}\CP(A:B)$
    \item $\mbf{J}^{AB}_{\Theta} \in \Ent_{k}(A:B)$
    \item $\Delta_{\Theta}$ is $k$-superpositive (Proposition \ref{prop:ent-rank-channel-property})
\end{enumerate}
Moreover, in the special case $\Theta$ is a superchannel, then the supermap may be implemented using $\Gamma_{pre} \in \Channel(B_{0},A_{0}E), \Gamma_{post} \in \Channel(A_{1}E,B_{1})$ whose concatentation is $k-$superpositive.
\end{theorem}
\begin{proof}
We prove the equivalence in a cycle. ($1\to2$) This follows by definition of $k-$entanglement-breaking supermap and the special choice of input map $\Upsilon_{+}^{\wt{A}A}$. ($2\to3$) $(\mathds{1}_{\wt{A}} \otimes \Theta_{AB})[\Upsilon_{+}^{\wt{A}A}] \in \Ent_{k}\CP(A:B)$ implies $J_{(\mathds{1}\otimes\theta)[\Upsilon]} \in \Ent_{k}(A:B)$ by \eqref{eqn:supermap-choi}. ($3\to4$) Using that $\mbf{J}^{AB}_{\Theta}$ is the Choi operator of $J_{\Delta_{\Theta}}$, $J_{\Delta_{\Theta}} \in \Ent_{k}(A:B)$ so $\Delta_{\Theta}$ is $k-$superpositive by Proposition \ref{prop:ent-rank-channel-property}. ($4-1$) As $\Delta_{\Theta}$ is $k-$superpositive, by Proposition \ref{prop:ent-rank-channel-property} and that $\Delta_{\Theta}$ maps Choi matrices to Choi matrices, $J^{RB}_{(\mathds{1}^{R} \otimes \Theta)[\Phi]} = (\id_{R} \otimes \Delta_{\Theta})(J^{RA}_{\Phi}) \in \Ent_{k}(R:B)$ for any Hilbert space $R$ and any $J^{RA}_{\Phi} \geq 0$. Thus by Definitions \ref{Defn:k-ent-gen} and \ref{defn:k-ent-break}, $\Theta$ is a $k-$entanglement-breaking map.

The moreover statement is as follows. It is well-known that a superchannel is always able to be implemented with pre- and post-processing with memory (Proposition \ref{prop:superchannel-equivalence}) and that $\mbf{J}_{\Theta}^{AB} = J_{\Gamma_{post} \circ \Gamma_{pre}}$. One can therefore conclude that the concatentation of these maps must be $k-$superpositive.
\end{proof}
This gives us the immediate corollary which follows from \eqref{eqn:superchannel-dual-map-choi}.
\begin{corollary}\label{corr:ent-break-map-duality}
Let $\Theta \in \mbb{L}(A,B)$. $\Theta$ is $k-$entanglement-breaking if and only if $\Theta^{\ast}$ is $k-$entanglement-breaking.
\end{corollary}
The final point we note is to smooth an extended min-entropy, one considers the diamond norm \cite{Gour-2019a}:
$$ H^{\ext,\ve,\cK}_{\min}(B|A)_{\Phi} := \sup_{\|\Phi - \Psi\|_{\diamond} \leq \ve} H^{\ext,\cK}_{\min}(B|A)_{\Psi} \ , $$
and likewise for the flipped registers.

We are now in a position to prove anti-standard AEPs for $H_{\min}^{\ext}(B|A),H_{\min}^{\ext}(A|B)$. To talk about restricted cones, we will denote 
\begin{align*}
\mbb{S}^{\cK}(A,B) := \{\Theta \in \mbb{S} : \mbf{J}_{\Theta}^{AB}  \in \cK \} \hspace{1cm}
\wt{\mbb{S}}^{\cK}(A,B) := \{\Theta^{\ast} \in \wt{\mbb{S}}: \mbf{J}_{\Theta}^{AB} \in \cK \} \
\end{align*}As in the min-entropy case, the proof methods are distinct at this point, so we introduce the remaining technicalities as we go.
\begin{theorem}\label{thm:anti-AEP-ext-min-ent}
Consider a sequence of cones $\cK^{(n)}((AB)^{\otimes n})$. Then
\begin{align*}
    \lim_{n \to \infty} \left[\frac{1}{n} H^{\ext,\cK,\ve}_{\min}(B|A)_{\id_{AB}^{\otimes n}}\right] \neq \lim_{n \to \infty} \left[\frac{1}{n} H^{\ext,\ve}_{\min}(B|A)_{\id_{AB}^{\otimes n}}\right] 
\end{align*}
if and only if $\supp_{\cK^{(n)} \cap \widehat{D}}(\tau_{A_{0}B_{0}}^{\otimes n}) \leq f(n) = O(c^{n})$ where $c< 1$ and $\widehat{D}((AB)^{\otimes n}) := \{ (d_{A_{1}}d_{B_{0}})^{-n}\mbf{J}_{\Theta}^{A^{n}B^{n}} : \Theta \in \mbb{S}^{\cK} \}$.
\end{theorem}
\begin{proof}
We focus on the case $\ve = 0$ as that is sufficient since smoothing only aggravates the problem. Using \eqref{eq:superchannel-Choi-overlap}, 
\begin{align*}
    &\frac{1}{n}\left[H_{\min}^{\ext,\cK}(B^{n}|A^{n})_{\id_{AB}^{\otimes n}}\right] \\
    =& \frac{1}{n}\left[-\log \left((d_{A}d_{B_0})^{n} \max_{\rho \in \widehat{D}^{\cK}((AB)^{\otimes n})} \langle \rho, \tau_{A_{0}B_{0}}^{\otimes n} \rangle \right) \right]\\
    \geq& -\log(d_{A}d_{B_{0}}) -\frac{1}{n}\log f(n)
\end{align*}
It follows if $f(n) = O(c^{n})$ where $c < 1$, then this will be strictly greater than $-\log(d_{A}d_{B_0})$, which is the limit for $\cK = \Pos$ as $\tau_{A_0 B_0}^{\otimes n} \in \widehat{D}^{\Pos}$. As the support function is upper bounded by unity, the only other case is $c =1$ in which case the AEPs agree. Again, the smoothing only aggravates the problem.
\end{proof}
\begin{corollary}
An AEP cannot hold for $H^{\ext,\cK,\ve}_{\min}(B|A)_{\Phi}$ in the case of a sequence of cones $\cK^{(n)} \subseteq \Ent_{f(n)}(A^{n}:B^{n})$ and there exists $n_0 \in \mbb{N}$ such that $f(n) < d^{n}$ for all $n > n_{0}$ where $d = d_{A_{0}B_{0}}$. In particular, if $f(n) = \lfloor (d-\ve)^{n}\rfloor$ where $\ve > 0$, the AEP does not hold in general.
\end{corollary}
\begin{proof}
If $\cK^{(n)} = \Ent_{f(n)}$, then one is optimizing over $\Theta \in \EntS_{f(n)}$. This results in the support function being upper bounded by $c_{(n)} := \supp_{\Ent_{f(n)}(A:B) \cap \Density}(\tau_{A_0B_0}^{\otimes n})$ which by \eqref{eqn:distance-from-max-ent}, means that $c_{(n)} \leq \frac{f(n)}{d_{A_{0}B_{0}}^{n}}$. It follows long as $f(n) \leq \lceil (d_{A_0 B_0}-\ve)^{n}\rceil$ for all $n \geq n_0$, then Theorem \ref{thm:anti-AEP-ext-min-ent} holds, which completes the proof.
\end{proof}
We now do the same thing for the flipped extended min-entropy, although in this case for simplicity we go directly to the entanglement rank cones.
\begin{theorem}
Consider any sequence of cones $\cK^{(n)}((AB)^{\otimes n})$ such that $\cK^{(n)} \subseteq \Ent_{f(n)}(A:B)$ and $f(n) < d_{A_{1}}^{n}$ for all $n \geq n_0$. Then
\begin{align*}
    \lim_{n \to \infty} \left[\frac{1}{n} H^{\ext,\cK^{(n)},\ve}_{\min}(A|B)_{\id_{AB}^{\otimes n}}\right]
    \neq \lim_{n \to \infty} \left[\frac{1}{n} H^{\ext,\ve}_{\min}(A|B)_{\id_{AB}^{\otimes n}}\right] 
\end{align*}
In particular, if $f(n) \leq \lfloor (d_{A_{1}} - \ve)^{n} \rfloor$ where $\ve > 0$ for $n \geq n_{0}'$, this AEP can't hold in general.
\end{theorem}
\begin{proof}
Using \eqref{eqn:ext-min-ent-flipped-singlet-frac}, we have
\begin{align*}
    & \frac{1}{n}H^{\ext,\cK^{(n)},\ve}_{\min}(A^{n}|B^{n})_{\id_{AB}^{\otimes n}}  \\ 
    =& -\log(d_{A}d_{A_{0}}) -\log(\max_{{\Theta} \in \EntS_{f(n)}(B,A)} \langle \tau_{A_{1}}^{\otimes n},(\mathds{1}^{A} \otimes \Theta)[\id_{AB}](\tau_{A_{0}}^{\otimes n})\rangle)\\
    \geq & -\log(d_{A}d_{A_{0}}) - \frac{1}{n}\log(\max_{\rho \in \Ent_{f(n)}(A^{n}:B^{n})} \langle \tau_{A_{1}}^{\otimes n}, \rho \rangle) \\
    \geq& -\log(d_{A}d_{A_0}) - \log(\sqrt[n]{f(n)}/d_{A_{1}}) \ ,
\end{align*}
where the equality uses Corollary \ref{corr:ent-break-map-duality}, the first inequality uses Theorem \ref{thm:k-ent-breaking-supermaps} combined with Proposition \ref{prop:ent-rank-channel-property}, and the second inequality uses \eqref{eqn:distance-from-max-ent}. It follows if $\limsup \sqrt[n]{f(n)} < d_{A_{1}}$, then the limit will not agree with the positive semidefinite case. As smoothing will only aggravate the issue, this completes the proof.
\end{proof}
In effect with this last proof we have now shown that entropic quantities for bipartite states, point-to-point channels, and bipartite channels, can only recover their asymptotic behaviour in general when we consider comparisons with objects over the positive semidefinite cone. Indeed, even an exponential increase in the allowed entanglement in the cones is not sufficient unless it approximates the positive semidefinite cone in the limit.

\section{Entanglement-Assisted Communication Value with Restricted Measurements}

Having done our due dilligence in understanding the framework of conic entropic quantities, we are now in a position to consider operational interpretations of this framework, as there would be little use for a framework that didn't capture any operational notions. 

In this section we give an operational interpretation of the conic min-entropy evaluated on a Choi state for arbitrary closed convex cone $\cK$ that is closed under local CP maps. As the operational interpretation generalizes the communication value of \cite{Chitambar-2021a}, we define it in these terms.

\begin{figure}
    \centering
    \includegraphics[width=0.7\linewidth]{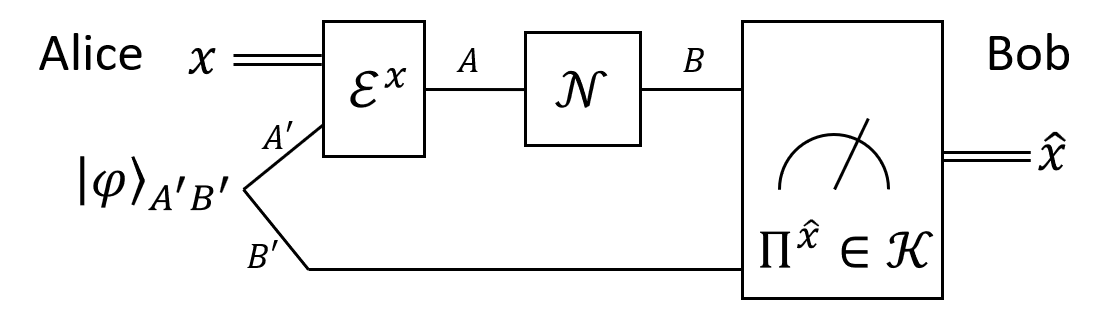}
    \caption{Here we depict the $\cK-$communication value following Definition \ref{def:k-cv}. Note the only restriction is on the power of the measurement in terms of the cone $\cK$ that it lives in.}
    \label{fig:generalized-cv}
\end{figure}

\begin{definition}\label{def:k-cv}
The $\cK-$communication value of a quantum channel $\cN \in \Channel(A,B)$, $\cv^{\cK}(\cN)$, is given by
\begin{align}
    \sup_{\varphi_{A'B'}} \max \left\{\sum_{x \in \Sigma} P(x|x) \right\} 
\end{align}
where $\varphi$ is optimized over all states and all dimension $A',B'$, all finite alphabets $\Sigma$ are optimized over, and 
$$P(x|x) := \Tr\left[\Pi_{BB'}^{x} \left(\cN^{A \to B} \circ \cE^{x}_{A'\to A} \otimes \id_{B'}\right)(\varphi_{A'B'})\right] \ , $$
where $\{\Pi^{x}\}_{x \in \Sigma} \subset \cK$ form a POVM and $\{\cE^{x}\} \subset \Channel(A,A')$ are a set of encoders.
\end{definition}
\begin{remark}
One could more generally allow $x$ and $y$ to index over different sized alphabets, but this was shown to not be uniquely optimal in \cite{Chitambar-2021a}, so we simplify to the minimal error discrimination setting given above.
\end{remark}
We do note that $\cv^{\cK}(\cN)$ does not normalize the minimal error discrimination by the size of the alphabet and so in this sense the $\cK-$communication value is best viewed as a measure of the channel's ability to noiselessly transmit classical information under (restricted) assistance. Regardless, the lack of normalization does not detract from it's ability to capture this zero-error behaviour; it was noted in \cite{Chitambar-2021a} that the $\cv^{\Sep},\cv^{\Pos}$ were related to zero-error channel simulation

We now show the optimal value of $\cv^{\cK}(\cN)$ is proportional to $H_{\min}^{\cK}(A|B)_{J_{\cN}}$ for cones that are invariant under local CP maps. Note that it is straightforward to see that being invariant under local CP maps is equivalent to being invariant under local isometries and local partial trace which arguably are the minimal physical requirements for capturing a notion of locality.
\begin{theorem}\label{thm:generalized-cv}
Let $\cK(A \otimes B)$ be a closed, convex cone such that it is invariant under local CP maps on the $A$ space and $\I \in \relint(\cK)$. Then $\cv^{\cK}(\cN) = \exp(-H_{\min}^{\cK}(A|B)_{J_{\cN}})$.
\end{theorem}
\begin{proof}
The proof is largely identical to that of \cite[Theorem 6]{Chitambar-2021a}. For completeness, we present it in full. Let $\cK(B\otimes B')$ be an arbitrary convex cone such that $\I \in \relint(\cK)$ and it is closed under local CP maps on the $B'$ (i.e.\ $\id_{B} \otimes \Psi_{B' \to A'}(X) \in \cK(A' \otimes B')$ if $X \in \cK(B \otimes B')$ and $\Phi,\Psi \in \CP$). We note the $B'$ space will ultimately correspond to the $A$ space of the cone.

Let $\ket{\varphi}_{A'B'}$ be an arbitrary pure state.  Without loss of generality we can take $\ket{\varphi}$ to be maximally entangled as Nielsen's Theorem guarantees the existence of an LOCC transformation such that $\ket{\tau}_{A'\wt{A'}}\to\ket{\varphi}_{A'B'}$, where we recall $\ket{\tau}_{A'\wt{A'}}$ is the maximally entangled state where $A' \cong \wt{A}'$. Specifically, Nielsen's theorem uses a measurement on Bob's side with Kraus operators $\{M^k\}_k$ and correcting unitaries $\{U^k\}_k$ on Alice's side such that $U_{A'}^k\otimes M_{\wt{A'}\to B'}^k\ket{\tau}_{A'\wt{A}'}=\sqrt{p(k)}\ket{\varphi^{k}}_{A'B'}$. Consider entanglement-assisted minimal error discrimination for recovering $\ket{\varphi}$ under the restriction all POVM elements are contained in $\cK(B \otimes B')$. That is,
\begin{equation}
\label{eq:restricted-cv-1}
\begin{aligned}
\sum_{x} P^{\cK}(x|x) 
=\sum_{x,k}\tr\Big[\Pi^x_{BB'}\left(\cN_{A\to B}\circ\cE^x_{A'\to A}\otimes\id^{B'}\right) \left(p(k)\dyad{\varphi^{k}}\right)\Big] \ ,
\end{aligned}
\end{equation}
where by construction
$$ p(k)\dyad{\varphi^{k}} = [(U^k\otimes M^k)\tau_{A'\wt{A}'}(U^k\otimes M^k)^\ast]  \ , $$
and by assumption $\Pi^{x}_{BB'} \in \cK(B \otimes B')$ for all $x$. Note $\{(\mbb{I}\otimes M^k)^\ast \Pi^x(\mbb{I}\otimes M^k)\}_{k,x}$ forms a POVM on $B\wt{A'}$. This follows from the fact that $\{M^k\}_k$ are Kraus operators for a CPTP map, and so the adjoint map, $\sum_k {M^{k}}^\dagger(\cdot)M^k$ is a unital map. Similarly, note that $\cU^k(\cdot):=U^k(\cdot) {U^{k}}^{\ast}$ denotes a unitary channel, the collection $\{\cE^x \circ \cU^k\}_{x,k}$ is a family of encoders.  Therefore, we can re-express Eq. \eqref{eq:restricted-cv-1} as 
\begin{equation}
\label{eq:restricted-cv-2}
\begin{aligned}
\cv^*(\mbf{P}) 
=\sum_z\tr\!\left[\hat{\Pi}^z_{B\wt{A}'}\!\left(\cN_{A\to B}\circ\hat{\cE}^z_{A'\to A}\otimes\id_{\wt{A'}}[\tau_{A'\wt{A}'}]\right)\right],
\end{aligned}
\end{equation}
where the $\hat{\cE}^z$ and $\hat{\Pi}^z$ are the concatenated encoders and decoder. Note that as $\cK(B \otimes \wt{A}')$ is closed under local CP maps, $\hat{\Pi}^{z}_{B\wt{A}'} \in \cK(B \otimes \wt{A}')$ for all $z$. This shows that we can restrict our attention just to a shared maximally entangled state.  Furthermore, without loss of generality, we can assume that $d_{A'}\geq d_{A}$, because the transformation $\ket{\tau}_{A''\wt{A}''}\to\ket{\varphi}_{A'B'}$ is always possible for any $d_{A''}\geq d_{A'}$ because we could have just as well used the same argument with system $A''$ and arrived at $\tau^{A''\wt{A}''}$ in Eq. \eqref{eq:restricted-cv-2}.

We next take Kraus-operator decompositions $\cE^z(\cdot)=\sum_{k}N^{z,k}(\cdot){N^{z,k}}^\dagger$ with each $N^{z,k}:\mbb{C}^{d_{A'}}\to\mbb{C}^{d_A}$.  Since $d_{A'}\geq d_A$, we can use the transpose trick/ricochet property $N^{z,k}\otimes\mbb{I}\ket{\tau_{d_{A}}}^{A'\wt{A}'}=\mbb{I}\otimes {N^{z,k}}^{\Trans}\ket{\tau_{d_{A}}}^{A\wt{A}}$ to obtain
\begin{align}
\label{Eq:Ea-cv-3}
\cv^*(\mbf{P}) = \frac{1}{d_{A'}}\sum_z\sum_k\tr\left[\wt{P}^{z,k}\left(\cN_{A\to B}\otimes\id_{\wt{A}}\right)(\tau_{d_{A}})\right]
=  \tr[\Omega^{AB} J_\cN^{AB}],
\end{align}
where $\wt{P}_{z,k} := (\mbb{I}^B\otimes {N^{z,k}}^{\ast})\hat{\Pi}^z_{B\wt{A'}}(\mbb{I}^B\otimes {N^{z,k}}^{\Trans})$ and we have swapped the ordering of the systems to match earlier notation, and
\begin{align}
\Omega_{AB}&=\frac{1}{d_{A'}}\sum_z\sum^k({N^{z,k}}^{\ast}\otimes\mbb{I}_B )\hat{\Pi}^z_{A'B}({N^{z,k}}^{\Trans}\otimes\mbb{I}_B)\notag\\
&=\frac{1}{d_{A'}}\sum_z {\cE_{z}^{\ast}}^{A' \to A}\otimes\id^B\left(\hat{\Pi}_z^{A'B}\right),
\end{align}
in which $\cE^*_z(\cdot):=\sum_{k}N^*_{z,k}(\cdot)N_{z,k}^T$. Note, ${\cE_{z}^{\ast}}^{A' \to A} \otimes \id_{B}(\hat{\Pi}^{z}_{A'B})\in \cK(A \otimes B)$ as it is again a local CP map $\wt{A}' \to A$ and the cone structure was invariant under CP maps on $\wt{A}'$ by assumption.\footnote{Note we re-ordered the spaces so we before had $\cK(B \otimes \wt{A}')$, which was then $\cK(\wt{A}' \otimes B)$, which is still invariant under the same space, the ordering merely changed.} As $\cK$ is a cone, so sums of its elements are contained, so it follows that $\Omega_{AB} \in \cK(A \otimes B)$. Moreover, since each $\cE_z^*$ is trace-preserving, we have
\begin{align}
\tr_A\Omega_{AB}&=\frac{1}{d_{A'}}\sum_z\tr_A\left[\cE^{*A'\to A}_z\otimes\id^B\left(\hat{\Pi}^z_{A'B}\right)\right]\notag\\
&=\frac{1}{d_{A'}}\sum_z\tr_{A'}\left(\hat{\Pi}^z_{A'B}\right)\notag\\
&=\frac{1}{d_{A'}}\tr_{A'}\left(\mbb{I}_{A'}\otimes\mbb{I}_B\right)=\mbb{I}_B.
\end{align}
Hence $\tr_A\Omega_{AB}=\mbb{I}_B$ and $\Omega \in \cK$ is a necessary condition on the operator $\Omega_{AB}$ such that $\sum_{x} P(x|x) =\tr[\Omega^{AB}J_{\cN}^{AB}]$. Now we make sure it is also sufficient.

Consider any positive operator $\Omega_{AB}$ such that $\tr_A\Omega_{AB}=\mbb{I}_B$ and $\Omega_{AB} \in \cK(A \otimes B)$. Consider the discrete Weyl operators on $A$, explicitly given by $U_{m,n}=\sum_{k=0}^{d_A-1}e^{i mk 2\pi/d_A}\op{m\oplus n}{m}$, where $m,n=0,\dots,d_A-1$ and addition is taken modulo $d_A$. It's well known that twirling the discrete Weyl operators results in the completely depolarizing map, e.g.\ $$\Delta(X):=\frac{1}{d_A}\sum_{m,n}U_{m,n}(X)U_{m,n}^\dagger = \tr[X]\mbb{I} \ . $$
Hence, 
\[\Delta_A\otimes\id_B[\Omega_{AB}]=\mbb{I}_A\otimes \tr_A\Omega_{AB}=\mbb{I}_A\otimes\mbb{I}_B.\]
This implies that the set $\{ \mu(m,n) := \cU^A_{m,n}\otimes\id^B(\Omega^{AB})\}_{m,n}$ forms a valid POVM on $AB$.  Therefore, we can construct an entanglement-assisted protocol as follows.  Let Alice and Bob share a maximally entangled state $\ket{\tau_{d_A}}_{\wt{A}A}$.  Alice then applies the $\cU_{m,n}^T(\cdot):=U_{m,n}^T(\cdot) U_{m,n}^*$ to the $A$ system and sends it through the channel $\cN$. Bob then performs the POVM $\{\mu(m,n)\}_{m,n}$ just described on systems $\wt{A}B$. The obtained score is
\begin{align}
 \sum_{m,n}P(m,n|m,n)
=& \frac{1}{d_A}\sum_{m,n}\tr\bigg[\left(\cU^{\wt{A}}_{m,n}\otimes\id^B\left[\Omega_{\wt{A}B}\right]\right)(\id_{\wt{A}}\otimes\cN)  \left(\id_{\wt{A}}\otimes\cU^{\Trans}_{m,n}\left[\tau_{d_A}\right]\right)\bigg]\notag\\
=& \frac{1}{(d_A)^2}\sum_{m,n}\tr\left[\Omega^{\wt{A}B}\left(\id^{\wt{A}}\otimes\cN\left[\tau_{d_A}\right]\right)\right]\notag\\
=& \tr[\Omega J_{\cN}].
\end{align}
The key idea in these equalities is that the unitary encoding $U_{m,n}$ performed on Alice's side is canceled by exactly one POVM element on Bob's side.

We therefore have necessary and sufficient conditions for $\cv^{\cK}(\cN)$ to be given by $\max\{ \Tr[\Omega_{AB}J_{\cN}] :\, \Tr_{A}[\Omega] = I_{B} \,,\, \Omega_{AB} \in \cK \}$, which is the conic program given in Eqn. \eqref{eqn:HminKPrimal}. Noting that we assumed $I \in \relint(\cK)$ means that strong duality holds and that the min-entropy corresponds to the conic program. This completes the proof.
\end{proof}

One advantage of the above result is it allows us to refine results of \cite{Chitambar-2021a}. Specifically, it allows us to interpolate between the communication value- captured by the separable cone- and the entanglement-assisted communication value- captured by the positive semidefinite cone- by using the entanglement rank cones. Similarly, it resolves an open problem of \cite{Chitambar-2021a} in determining what the operational interpretation of $\cv^{\PPT}$ was. We summarize these in the following corollary.

\begin{corollary}
Let $\cK \in \{\Ent_{k}(A:B)\}_{k \in \mbb{N}} \cup \PPT(A:B)$. Then $\cv^{\cK}(\cN)$ is the optimal unnormalized minimal error discrimination achieved by $\cN$ using entanglement assistance and any POVM whose elements are restricted to $\cK$.  
\end{corollary}
\begin{proof}
This simply follows from $\rho \in \Ent_{k}(A:B)$ implying $(\Phi \otimes \Psi)(\rho) \in \Ent_{k}(A:B)$ for $\Phi,\Psi \in \Channel$. Likewise for the PPT cone.
\end{proof}

\begin{remark}
Note that Theorem \ref{thm:anti-AEP} establishes that even if we were to consider smoothed versions of the communication value problem and consider using the channel in parallel $n$ times, the asymptotic behaviour would be distinct between using a POVM restricted to separable measurments and a general POVM. In particular, note that up to scaling, Theorem \ref{thm:anti-AEP} is proven using the choice of the identity channel.\footnote{Note that if one wanted to smooth over Choi states, one would want to stay in the space of Choi states, but the proof of Theorem \ref{thm:anti-AEP} holds without even taking the structure of the smoothing into account.} This is good as it aligns with our understanding of (entanglement-assisted) classical capacity of a quantum channel.
\end{remark}

\paragraph*{Relation to the Holevo-Schumacher-Westmoreland Theorem}
Previously it was shown that $\cv^{\Sep}$ may be super-multiplicative \cite{Chitambar-2021a}. This property arose from the fact that $\cv^{\Sep}$ was achieved by preparing a state without entanglement assistance and using a POVM on the output signal. It follows that in such a setting one can gain an advantage using entangled states by sending them over an $n-$fold copy of the channel $\cN$ which results in super-multiplicativity. In contrast, it was shown that $\cv^{\Pos}$ was always multiplicative, which we could have imagined was due to the entanglement assistance already being free. As the $\cK-$communication value allows us to interpolate between these two settings without restricting the shared entanglement, the importance of the entangling power of the measurement is made explicit.

This basic idea mirrors the Holevo-Schumacher-Westmoreland (HSW) theorem \cite{Holevo-1998a,Schumacher-1997a} where the achievability proof uses encodings that are product but relies on the decoder being arbitrary. Furthermore, the HSW theorem requires considering the limit of the $n-$fold channel for the same reason as $\cv^{\Sep}$ but in the entanglement-assisted case becomes multiplicative. However, we note the HSW theorem is recovered from the one-shot setting using
$$ R \approx_{\ve} \sup_{p_{X}} D^{\ve}_{h}(\rho^{XB}||\rho^{X} \otimes \rho^{B}) \ , $$
where $R$ is the rate, $\rho^{XB} = \sum_{x} p(x) \dyad{x} \otimes \rho_{x}^{B}$ and $\cN: \dyad{x} \to \rho^{B}_{x}$ \cite{Wang-2012a}. Noting that both arguments are classical-quantum and this will remain the case when recovering the HSW theorem, by Proposition \ref{prop:sep-one-shot}, one could restrict to the separable-restricted hypothesis testing $D^{\Sep,\ve}_{h}$ and recover the HSW theorem. This difference is in a sense not surprising as $\cv^{\Sep}$
captures (the lack of) entangling power between the signal and an ancillary system, not across the signal itself.

\section{Channel Conversion via Bistochastic Channels}

\begin{figure}[H]
    \centering
    \includegraphics[width=0.7\linewidth]{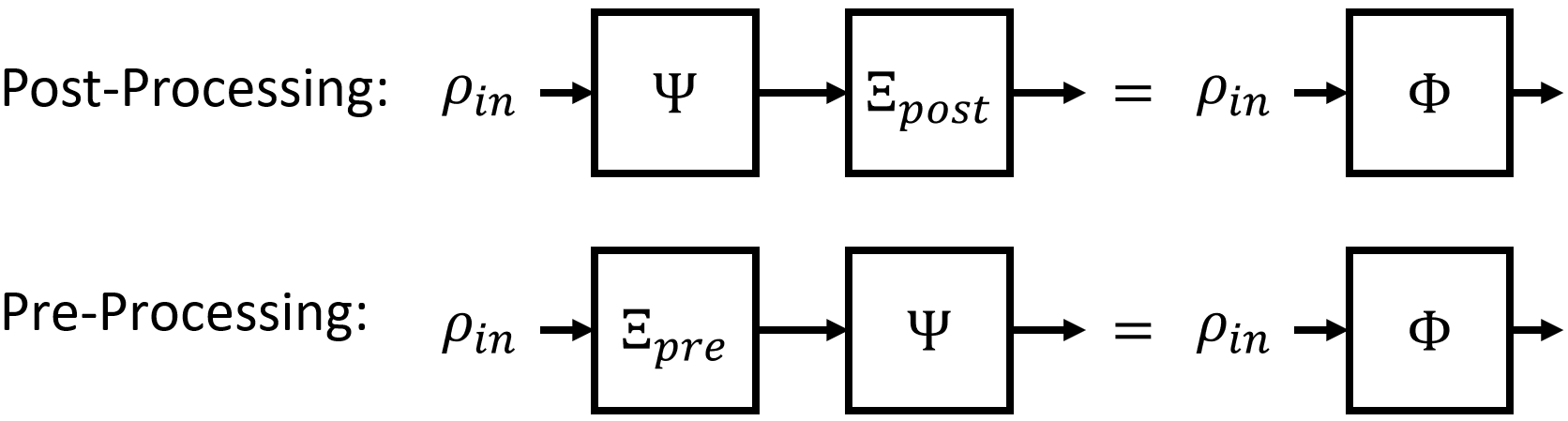}
    \caption{Equivalence of channels $\Phi$ and $\Psi$ via post- and pre-processing. Post-processing convertability has been considered in previous works. In this section we consider bistochastic pre-processing.}
    \label{fig:pre-and-post-proc-pic}
\end{figure}

Min-entropic-like quantities have been fundamental in quantum majorization results, which are partial orders on convertibility between resources \cite{Gour-2018a,Gour-2019a,Gour-2020a,Ji-2021a}. The traditional example of majorization is Blackwell's result on the ability to convert one channel to another using post-processing, known as degradability \cite{Blackwell-1953a,Blackwell-1953b}. The quantum version of degradability was considered in \cite{Buscemi-2016a}. Part of the reason degradability is so important in the quantum regime is because it is useful for bounding capacities of quantum channels \cite{Sutter-2017a}, which are notoriously difficult to determine directly. 

To both exemplify how majorization can be viewed as a function of the cone and stress the relationship between min-entropic terms being support functions of dynamic resources, we study the case of conversion via pre-processing, which has not previously considered (See Fig. \ref{fig:pre-and-post-proc-pic} for distinction). Specifically, we first introduce a new restricted min-entropy which we call the `doubly-restricted min-entropy.' As the definition shows, it in effect measures the support of a Choi matrix in the space of bistochastic channels contained in a cone $\cK$. As such, it in effect optimizes over the intersection of the spaces considered for $H_{\min}^{\cK}(A|B),H_{\min}^{\cK}(B|A)$ via equations \eqref{eqn:min-ent-unital-support} and \eqref{eqn:min-ent-cptp-support} respectively. 
We then show one channel may be converted to another via bistochastic pre-processing in the allowed cone $\cK$ if and only if the doubly restricted min-entropy of the output channel is less than that of the input channel under all post-processings. In other words, we formalize the intuition that a channel $\Phi$ can be simulated by bistochastic pre-processing $\Psi$ if and only if $\Psi$ is a less noisy channel than $\Phi$ under all physical conditions where noisiness is measured in terms of this new min-entropic quantity.

To do this, we introduce a new restricted min-entropy which we call the `doubly'-restricted min-entropy as we restrict both the cone and the set of channels.
\begin{definition}\label{defn:doubly-restricted-min-entropy}
Let $\cK \in (A \otimes B)$ be a closed convex cone with $I \in \relint(\cK)$. Let $A \cong B \cong \mbb{C}^{d}$. Let $P \in \Pos(A \otimes B)$. Then we define the doubly restricted min entropy, $H^{\cK,\uparrow}_{\min}(A|B)_{P}$ by
$$ \exp(-H^{\cK,\uparrow}_{\min}(A|B)_{P}) := \max_{\Psi \in \CPTPU^{\cK}(A,B)} \langle J_{\Psi}, P \rangle  \ , $$
where $\CPTPU^{\cK}(A,B) := \{J_{\Lambda} \in \cK : \Lambda \in \mrm{CPTPU}(A,B) \} $ .
\end{definition}
We remark the need for $A \cong B \cong \mbb{C}^{d}$ simply follows from the fact to be unital, $\Tr_{A}J_{\Psi} = \I_{B}$ and to be trace-preserving, $\Tr_{B}J_{\Psi} = \I_{A}$, which can only be true of $A,B$ have the same dimension. Furthermore, we note this entropy has the following properties of particular interest.
\begin{proposition}\label{prop:doubly-restricted-min-entropy-props}
The doubly restricted min-entropy satisfies the following properties
\begin{enumerate}
    \item (Maximality) 
    \begin{align*}
        H_{\min}^{\cK,\uparrow}(A|B)_{P} \geq \max\{H^{\cK}_{\min}(A|B)_{P},H^{\cK}_{\min}(B|A)_{P}\}
    \end{align*}
    \item (Symmetry) If $X \in \cK(A \otimes B)$ if and only if $\mrm{SWAP}(X)^{\Trans} \in \cK(B \otimes A)$, then $H_{\min}^{\cK,\uparrow}(A|B)_{P} = H_{\min}^{\cK,\uparrow}(B|A)_{P}$.
    \item (Singlet Fraction) $\exp(-H^{\cK,\uparrow}_{\min}(A|B)_{P})$ is equal to
    $$ d_{A} \max_{\substack{\Psi \in \CPTPU(B,A'): \\ J_{\Psi^{\ast}} \in \cK}} \langle (\id_{A} \otimes \Psi)(P), \tau^{AA'} \rangle \ .$$
    \item (Additivity for Positive Cone) 
    \begin{align*}
    & H^{\Pos,\uparrow}_{\min}(A^2_1|B^2_1)_{P \otimes Q} \\
    =& H^{\Pos,\uparrow}_{\min}(A^1|B^1)_{P} + H^{\Pos,\uparrow}_{\min}(A^2|B^2)_{Q} \ .
    \end{align*}
    \item (Data-Processing-Like-Inequality) Let $\cK$ be invariant under local CP maps. Let $\Phi \in \CPTPU(A,A')$, $\Psi \in \CPTPU(B,B')$, then 
    $$H_{\min}^{\cK,\uparrow}(A|B)_{P} \leq H_{\min}^{\cK,\uparrow}(A'|B')_{P}  $$
    \item Does Not Satisfy AEP for Insufficiently Positive Cones
\end{enumerate}
\end{proposition}
\begin{proof}
Item 1 follows from Definition \ref{defn:doubly-restricted-min-entropy} along with the first line in the proof of Proposition \ref{prop:min-entropy-singlet-fraction} and \eqref{eqn:min-ent-cptp-support}. Let $\CU \equiv \CPTPU$. Item 2 follows from 
\begin{align*}
\exp(-H_{\min}^{\cK,\uparrow}(B|A)_{P})
=& \underset{\Psi^{\ast} \in \CU^{\cK}(B,A)}{\max} \langle J_{\Psi^{\ast}}, \mrm{SWAP}(P) \rangle \\
=& \underset{\Psi \in \CU^{\cK}(A,B)}{\max} \langle \mrm{SWAP}(J_{\Psi})^{\Trans} , \mrm{SWAP}(P) \rangle \\
=& \underset{\Psi \in \CU^{\cK}(A,B)}{\max} \langle J_{\Psi}^{\Trans}, P \rangle \\
=& \underset{\Psi \in \CU^{\cK}(A,B)}{\max} \langle J_{\Psi}, P \rangle \\
=& \exp(-H^{\cK,\uparrow}_{\min}(A|B)_{P}) \ ,
\end{align*}
where we used that $\Psi \in \CPTPU^{\cK}(A,B)$ if and only if $\Psi^{\ast} \in \CPTPU^{\cK}(B,A)$. Item 3 follows from the same proof form as for Proposition \ref{prop:min-entropy-singlet-fraction}. Item 4 follows from noting there is a straightforward conic program for $\exp(-H^{\cK,\uparrow}_{\min}(A|B))$ given by the primal
\begin{align*}
\max\{\ip{P}{X}\;:\;X_{A} = \I_{A} \;,\; X_{B} = \I_{B} \;,\; X \in \cK \}
\end{align*}
and the dual
\begin{align*}
\min\{\Tr(Y_1) + \Tr(Y_2):\I_{A} \otimes Y_1 + Y_2 \otimes \I_B \succeq_{\cK^{\ast}} P, 
 Y_1 \in \Herm(B)\;, Y_2 \in \Herm(A) \} \ ,
\end{align*}
which will have strong duality under the assumption $I \in \relint(\cK)$. So using the optimizers in the primal and dual cone for $P,Q$ respectively to construct the optimizers for $P \otimes Q$, one obtains additivity. For Item 5, note we can re-write the dual problem as
\begin{align*}
    \min\{2^{\gamma_{1}} + 2^{\gamma_{2}} : 2^{\gamma_{1}}\I_{A} \otimes \sigma_{1} + \sigma_{2} \otimes 2^{\gamma_{2}}\I_{B}\,, \, \sigma_{i} \in \Density\} \ .
\end{align*}
By the same argument as in the proof of Proposition \ref{prop:DPI}, it follows by assuming $\cK$ is invariant under local CP maps and $\Phi,\Psi \in \CPTPU$, that the optimal choices $(\gamma_{1},\gamma_{2},\sigma_1,\sigma_2)$ for $P$ result in $(\gamma_{1},\gamma_{2},\Psi(\sigma_{1}),\Phi(\sigma_{2}))$ being feasible for $(\Phi \otimes \Psi)(P)$. So $\exp(-H^{\cK,\uparrow}_{\min}(A|B)_{P}) \geq \exp(-H^{\cK,\uparrow}_{\min}(A|B)_{(\Psi \otimes \Phi)(P)})$, which implies Item 5. Item 6 follows from $\exp(-H^{\cK,\uparrow}_{\min}(A|B)_{P}) \leq \exp(-H^{\cK}_{\min}(A|B)_{P})$ along with Theorem \ref{thm:anti-AEP}.
\end{proof}

We now present our pre-processing result. The proof largely follows that of \cite{Gour-2018a} with a few modifications. For notational simplicity, if we talk of (a subset) of linear transformations $\cS \subset \Trans(A,B)$ whose Choi matrix is contained in $\cK(A \otimes B)$, we will denote the set $\cS^{\cK}$. It will also be convenient to restrict ourselves to the set of cones $\cK$ such that $X^{AB} \in \cK(A \otimes B)$ implies $\mrm{SWAP}(X)^{\Trans} \in \cK(B \otimes A)$, as $J_{\Phi} = \mrm{SWAP}(J_{\Phi^{\ast}})^{\Trans}$ (Lemma \ref{lemma:adjoint-choi}). Note however this is not a particularly strong requirement as all standard cones in quantum information theory satisfy this property.
\begin{theorem}
Let $\cK$ be closed, convex, such that $X \in \cK(A \otimes B)$ if and only if $\mrm{SWAP}(X)^{\Trans} \in \cK(B \otimes A)$. Let $D \cong A \cong B$. Given $\Phi \in \Channel(A,C), \Psi \in \Channel(B,C)$ the following are equivalent
\begin{enumerate}
    \item There exists $\Xi \in \CPTPU^{\cK}(A,B)$ such that 
    $$ \Phi = \Psi \circ \Xi \ . $$
    \item $H^{\cK,\uparrow}_{\min}(B|D)_{(\id \otimes \cE)(J_{\Psi})} \leq H^{\cK,\uparrow}_{\min}(A|D)_{(\id \otimes \cE)(J_{\Phi})} $ for all $\cE \in \Channel(C,D)$.
    \item The previous relation holds for any measure-and-prepare channel $\cE(X) := \sum_{j=1}^{d_{A}^{2}} \Tr[M^{C}_{j}X]\omega_{j}$ where $\{\omega_{j}\}$ may be freely varied and $\{M^{C}_{j}\}$ is an arbitrary but fixed informationally complete POVM.
\end{enumerate}
\end{theorem}
\begin{proof}
Throughout this proof, to simplify notation, we will denote $\CU^{\cK}(A,B) \equiv \CPTPU^{\cK}(A,B)$. By the structure of expressing post- and pre-processing of quantum channels (Prop. \ref{prop:pre-and-post-proc-choi}) along with Lemma \ref{lemma:trans-channel-props}, there exists $\Xi \in \CU^{\cK}(A,B)$ such that $\Phi = \Psi \circ \Xi$ if and only if there exists $\cE^{\Trans} \in \CU^{\cK}(B,A)$ such that $(\cE^{\Trans} \otimes \id)J_{\Psi} = J_{\Phi}$. Let $\{Q^{C'}_{k}\}_{k}$ be the dual basis of our arbitrary but fixed informationally complete POVM $\{M^{C'}_{k}\}_{k}$ where $C' \cong A$. We therefore may write
$$ J_{\Psi} = \sum_{k} P_{k}^{B} \otimes Q^{C'}_{k} \quad J_{\Phi} = \sum_{k} R_{k}^{A} \otimes Q^{C'}_{k} \ , $$
where $P^{B}_{k} = \Tr_{C'}[(\I_{A} \otimes M^{C'}_{k})J_{\Psi}]$ and same idea for $R^{A}_{k}$. We may then conclude there exists such a $\Xi$ if and only if there exists $\cE^{\Trans} \in \CU^{\cK}(B,A)$ such that $\cE^{\Trans} (P_{k}^{B}) = R^{A}_{k}$ for all $k$. Note that $\Tr_{A}[J_{\Phi}] = \Tr_{A}[(\cE^{\Trans} \otimes \id)(J_{\Psi})] = \Tr_{A}[J_{\Psi}]$. As $\{Q^{C'}_{k}\}$ is a basis, we can conclude that $c_{k} := \Tr[P_{k}] = \Tr[R_{k}]$ for all $k$ as otherwise the coefficients would differ. Then for arbitrary $\{\omega_{k}\} \subset \Density(C')$, we can define the operator
\begin{align}\label{eqn:defnofOmega}
\Omega^{ABC'} = P_{k} \otimes R_{k} \otimes \omega_{k}^{C'} \ , 
\end{align}
from which it follows 
\begin{align}\label{eqn:omega-marginals}
\Omega^{AC'} = (\id \otimes \mbb{M})(J_{\Phi}) \quad \Omega^{BC'} = (\id \otimes \mbb{M})(J_{\Psi}) \ , 
\end{align}
where $\mbb{M}(X^{C'}) = \sum_{k} \Tr[M^{C'}_{k} X] \omega_{k}$, where $\{M_{k}^{C'}\}$ is the dual basis to $Q^{C}_{k}$, and thus $\mbb{M}$ is a CPTP map.

Now we start our variation of Lemma 1 of \cite{Gour-2018a}. By introducing self-adjoint basis of $A$, $\{X_{j}\}$, we require $\exists \cE^{\Trans} \in \CU^{\cK}(B,A)$ such that $\Tr[\cE^{\Trans}(P_{k}^{B})X_{j}] = \Tr[R^{A}_{k}X_{j}]$ for all $j,k$. Define $r_{\cF}$ by $r(k,j) = \Tr[\cF(P_{k}^{B})X_{j}]$, and the set $\mathcal{S} = \{r_{\cE^{\Trans}} : \cE^{\Trans} \in \CU^{\cK}(B,A)\}$. As $\cK$ is closed and convex, one can show that $\CU^{\cK}(B,A)$ is a compact and convex set, the set $\cS$ is convex and the separation hyperplane theorem may be applied in the same way as \cite{Gour-2018a}, ultimately resulting in our equivalent of \cite[Supplementary Material,Eqn 18]{Gour-2018a}
\begin{equation}\label{eq:pre-proc-proof-step-1}
\begin{aligned}
    \forall \{Z^{A}_{i}\} \in \Herm(A), \underset{\cE^{\Trans} \in \CU^{\cK}(B,A)}{\max} \sum_{i} \Tr[\cE^{\Trans}(P_{i}^{B})Z_{i}^{A}] \geq \sum_{i} \Tr[R^{A}_{i}Z_{i}] \ .
\end{aligned}
\end{equation}
Now, because $\cE^{\Trans}$ is trace-preserving, we may use the mapping $Z_{i} \mapsto \omega_{i} = [Z_{i} + z_{i}\I]/\kappa$ where $z_{i}$ are chosen so these are states to obtain the following equivalent condition to Eq. \eqref{eq:pre-proc-proof-step-1}
\begin{align*}
    \forall \{\omega^{C'}_{i}\} \in \Density(C'),
    \quad \underset{\cE' \in \CU^{\cK}(C',B)}{\max} \sum_{i} \Tr[P_{i}^{B} \cE'(\omega_{i}^{C'})] \geq \sum_{i} \Tr[R^{A}_{i}\omega_{i}^{C'}] \ ,
\end{align*}
where we have used $C' \cong A$ and $\cE'$ is really the adjoint of $\cE^{\Trans}$, so $\cE' = \mbb{C} \circ \cE \circ \mbb{C}$. As we optimize over all $\{\omega_{i}\} \subset \Density(C')$, the optimization is invariant under applying $\mbb{C}$, so $\cE' = \mbb{C} \circ \cE$. Now we use $\Tr[XY] = \Tr[X \otimes Y^{\Trans} \tau_{d}]$ where $X,Y \in \Lin(\mbb{C}^{d})$. This results in $\Trans \circ \cE' = \Trans \circ \mbb{C} \circ \cE$. Noting $\Trans \circ \mbb{C}$ is just complex conjugate and $\cE$ is Hermitian-preserving and any density matrix is Hermitian, we can drop these too, leaving us with just optimizing over $\cE \in \CU^{\cK}(A,B)$:
\begin{align*}
    \forall \{\omega^{C'}_{i}\} \in \Density(C'), 
    \hspace{3mm} \underset{\cE \in \CU^{\cK}(C',B)}{\max} \sum_{i} \Tr[P_{i}^{B} \otimes \cE(\omega_{i}^{C'}) \tau^{BB'}] \geq \sum_{i} \Tr[R^{A}_{i}\otimes \omega_{i}^{C'} \tau^{AA'}] \ .
\end{align*}
Moving the sums inside so that we are considering $\Tr[(\id^{B} \otimes \cE)(\Omega^{BC'})\tau^{BB'}]$ (and likewise for RHS), and multiplying by $d_{A}$ on either side, recalling $d_{A} = d_{B}$ by assumption, by definition of the doubly restricted min-entropy (Definition \ref{defn:doubly-restricted-min-entropy}), we have
\begin{equation}\label{eqn:pre-proc:nec-and-suff}
\begin{aligned}
\forall \{\omega^{C'}_{i}\} \in \Density(C') , 
 \hspace{5mm} \exp(-H_{\min}^{\cK,\uparrow}(B|C')_{\Omega}) \geq \exp(-H^{\cK,\uparrow}_{\min}(A|C')_{\Omega})
\end{aligned}
\end{equation}
is necessary and sufficient. This completes the equivalence of Items 1 and 2. It follows that
\begin{align*}
\forall \{\omega^{C'}_{i}\} \quad H_{\min}^{\cK}(B|C')_{\Omega} \leq H^{\cK}_{\min}(A|C')_{\Omega}
\end{align*}
is a sufficient condition because it implies \eqref{eqn:pre-proc:nec-and-suff} as can be verified using \eqref{eqn:HminKPrimal}. Moreover, if $\cE^{T}(P^{B}_{k}) = R^{A}_{k}$ for all $k$, then it's also necessary because $\Omega^{BC'} = (\cE^{T} \otimes \id)(\Omega^{AC'})$ satisfies data processing for any cone because $\cE^{\Trans}$ and $\id$ are CPTPU maps (Item 5 of Proposition \ref{prop:doubly-restricted-min-entropy-props}), i.e.\ $H_{\min}^{\cK}(A|C')_{\Omega} = H_{\min}^{\cK}(A|C')_{(\cE^{\Trans} \otimes \id)(\Omega^{BC'})} \geq H_{\min}^{\cK}(B|C')_{\Omega}$. Using Eqn. \eqref{eqn:omega-marginals}, we obtain Item 3 and by relaxing that Item 4 we obtain Item 2. This completes the proof.
\end{proof}

\section{Norms, Additivity, and Restricted Entropies}
At this point we have seen the general framework of conic information theory and its ability to capture different operational tasks as a function of the cone. We have also seen that for cones that in a specific sense don't asymptotically approximate the positive semidefinite cone, the asymptotic equipartition properties, and the relevant operational interpretations, cease to hold. A remaining question would then be under what conditions do the conic measures behave like the standard ($\cK = \Pos$) case. Since we have already seen that any pure state outside of the cone will result in $D_{\max}^{\cK}$ behaving differently than $D_{\max}$ (Proposition \ref{prop:Dkmax-pure-state-outside-cone}), the strongest result we could hope for is that $P,Q \in \cK$, $D_{\max}^{\cK}(P||Q)=D_{\max}(P||Q)$. 

Even less demanding, we might hope that $D_{\max}^{\cK}(P||Q)$ is additive. The additive property is not straightforward in general. This can be seen as follows: Consider the `natural' sequence of cones for the $n-$fold space, in the symmetry case where $\cK^{(1)} = \{X \in \Pos : \cG(X) = X\}$ where $\cG$ is the twirling map, then $\cK^{(n)} = \{X \in \Pos: \cG^{\otimes n}(X) = X\}$, and in the entanglement rank cone case $\cK^{(1)} = \Ent_{k}(A:B)$, so $\cK^{(n)} = \Ent_{n\cdot k}(A^{n}:B^{n})$. In both cases it is clear that the primal problems for $\exp(D^{\cK}_{\max}(P||Q))$ and $\exp(-H^{\cK}_{\min}(A|B))$ (Eqns. \eqref{eqn:DmaxKPrimal} and \eqref{eqn:HminKPrimal} respectively) are multiplicative by taking the $n-$fold optimizer. This implies $D_{\max}^{\cK}(P^{\otimes n}||Q^{\otimes n}) \geq n D_{\max}(P||Q)$, $H_{\min}^{\cK}(A^{n}|B^{n})_{P^{\otimes n}} \geq n H_{\min}^{\cK}(A|B)_{P}$. When $\cK = \Pos$, these are known to hold with equality because one can use the $n-$fold dual optimizer as well \cite{Tomamichel-2012a}. However, it is known for a general $P$, $\exp(-H^{\Sep}_{\min}(A|B)_{P})$ can have the optimizer change. This results in the communication value $\cv^{\Sep}$ being non-multiplicative \cite{Chitambar-2021a}.

In this section, we explore both of these questions by connecting the conic program framework to conic norms. We see that in effect for cones that capture symmetries, all of the properties hold so long as $P,Q \in \cK$. In contrast, for entanglement cones we can prove additivity of $D_{\max}^{\Ent_{k}}$ under significantly restricted conditions on $Q$.

\begin{definition}\label{defn:cone-norms}
Let $X \in \Lin(A)$. Let $\cK \subset \Pos$ be a closed, convex cone such that $\mrm{span}(\cK) =\Lin(A)$. Then the $\cK$-operator norm is
\begin{align*}
    \|X\|_{\cK} := \max_{\rho \in \cK} \{ |\langle X , \rho \rangle | : \Tr(\rho) = 1 \} \ ,
\end{align*}
where the maximization follows from $\cK$ being closed and convex guarantees the $\rho \in \Density(A) \cap \cK$ being a compact set.
\end{definition}
We note these norms were studied in \cite{Johnston-2012a,Skowronek-2012a} as well as in a sense in \cite{Jencova-2021a}. It is easy to see that in the proof of Proposition \ref{prop:Dkmax-pure-state-outside-cone}, we were in fact converting problems about $\cK$-restricted entropic quantities into a problem of $\cK$-norms. We now explain why this is the case.
\begin{proposition}(\cite[Proposition 5.2.4]{Johnston-2012a})
Let $X \in \Pos(A)$. Then $\gamma\I^{A}-X \in \cK^{\ast}$ if and only if $\|X\|_{\cK}\leq \gamma$.
\end{proposition}
\begin{proof}
By definition of dual cone, $\gamma\I^{A}-X \in \cK^{\ast}$ if and only if $\ip{\gamma\I^{A}-X}{K} \geq 0$ for all $K \in \cK$. Note since it is a cone we can restrict just to $K \in \cK$ such that $\Tr[K] = 1$, so it is if and only if $\gamma - \ip{X}{\rho}$ for all $\rho \in \cK \cap \Density(A)$, which is true if and only if $\gamma \geq \|X\|_{\cK}$.
\end{proof}
From this it is straightforward to show the conic norms satisfy the following conic program \cite{Johnston-2012a}.
\begin{proposition}[\cite{Johnston-2012a}]
For $X \in Pos(A)$, $\|X\|_{\cK}$ is given by 
\begin{align}
    \emph{Primal:} &\hspace{1mm} \max\{\ip{\rho}{X}\;:\; \Tr(\rho) \leq 1 \;\&\; X \in \cK \} \label{eqn:KNormPrimal} \\
    \emph{Dual:} &\hspace{1mm} \min\{\gamma: \gamma \I \geq  Y + X \;\&\; Y \in \cK^{\ast} \} \label{eqn:KNormDual}
\end{align}
\end{proposition}
It is worth remarking for a general cone $\cK$ it is not clear that it would be multiplicative over tensor products, unlike the standard operator norm. In particular, for $\cK = \Ent_{k}$, these are known as the Operator Schmidt Norms \cite{Johnston-2012a}, denoted $\|\cdot\|_{S(k)}$. For consistency with the literature, we will use this wording throughout this sections. Moreover, in the case that $\cK = \Sep(A:B)$, the primal problem given above is in fact the intermediary term $\Lambda^{2}(X)$ where the geometric measure of entanglement (GME) is given by $G(X) =-\log(\Lambda^{2}(X))$. It is well known that $G(X)$ is not in general additive, or equivalently that $\Lambda^{2}(X)$ is not in general multiplicative, over tensor products. It also follows that one might expect various properties known to hold for the GME might generalize for the operator schmidt norms. Indeed, we can obtain a significant generalization of the additivity of GME for states that are elementwise non-negative in the computational basis from \cite{Zhu-2011a}.
\begin{theorem}\label{thm:multiplicativity-of-norms}
Let $\rho \in \Density(A_{1} \otimes B_{1})$ be elementwise non-negative (in the computational basis). Let $\sigma \in \Density(A_{2} \otimes B_{2})$ be arbitrary. Then for any $r \in \mathbb{N}$,
$$ \|\rho \otimes \sigma\|_{\Ent_{r}(A^{2}_{1}:B^{2}_{1})} = \|\rho\|_{\Sep(A_{1}:B_{1})} \|\sigma\|_{\Ent_{r}(A_{2}:B_{2})} \ . $$
\end{theorem}
The proof is contained in the Appendix. We note that in the original work \cite{Zhu-2011a}, they in effect consider a generalized set of norms that also considers multipartite entanglement. As we are focused on bipartite entanglement throughout, we have don't provide such a generalization, though for the right choice of cone it ought to be straightforward.

It is then also clear by \eqref{eqn:DmaxKDual} that $D_{\max}^{\cK}(P^{A}||I^{A}) = \log \|P\|_{\cK}$ for $P \in \Pos(A)$. It follows that our choice of $\pi$ was fortuitous to make use of this. More generally, we have the following proposition to simplify calculations.
\begin{proposition}\label{prop:K-Dmax-to-cone}
Let $\cK \subset \Pos(A)$ be closed and convex. Let $P,Q \in \Pos(A)$. Consider the CP map $\Phi(\cdot) := Q^{-1/2} \cdot Q^{-1/2}$ where ${\cdot}^{-1/2}$ is the Moore-Penrose psuedo-inverse. If $\Phi: \cK(A) \to \cK(A)$, then $D_{\max}^{\cK}(P||Q) = \log\|Q^{-1/2}PQ^{-1/2}\|_{\cK}$. 
\end{proposition}
\begin{proof}
Note $\Phi$ is self-dual. Therefore, if $\Phi: \cK \to \cK$, then $\Phi^{\ast}: \cK \to \cK$. As such, if $2^{\gamma} Q - P \in \cK$, then $\cK^{\ast} \ni \Phi(\gamma Q - P) = 2^{\gamma} I - Q^{-1/2}PQ^{-1/2}$. By minimizing over $\gamma$, one obtains $\log\|Q^{-1/2}PQ^{-1/2}\|_{\cK}$.
\end{proof}
We note in the case that $\cK = \Pos$, this is a well known result where it is always the case as $Q^{-1/2}\cdot Q^{-1/2}$ is completely positive and $\|\cdot\|_{\Pos} = \|\cdot\|_{\infty}$. The difficulty is when $\cK \subsetneq \Pos$ as the entanglement properties of neither $Q^{-1/2}$ nor $Q^{-1/2}PQ^{-1/2}$ are straightforward. Regardless, it allows us to see certain cases where the problem is (more) tractable. 

\paragraph*{Cones Pertaining to (Abelian) Symmetry}
The first case shows that if the relevant property and cone is a symmetry, then we recover everything for which we could hope. In effect, this is not surprising as satisfying a symmetry implies commuting operators, but it does provide evidence it could be the case more generally.
\begin{theorem}\label{thm:information-theory-with-symmetries}
Let $\cK \subseteq \Pos(A)$ such that $X \in \cK$ if and only if $\cG(X) = X$ where $\cG$ is a twirling map. Let $P,Q \in \cK$. Then
\begin{enumerate}[itemsep=4pt]
    \item $D_{\max}^{\cK}(P||Q) = D_{\max}(P||Q)$.
    \item If $\cK^{(n)}(A^{\otimes n}) := \{X \in \Pos(A^{\otimes n}): \cG^{\otimes n}(X) = X\}$, $D_{\max}^{\cK^{(n)}}(P^{\otimes n}||Q^{\otimes n}) = nD_{\max}(P||Q)$ 
    \item (AEP) For $\ve \in (0,1)$,
    $$\lim_{n \to \infty} \left[\frac{1}{n} D_{\max}^{\cK,\ve}(P^{\otimes n}||Q^{\otimes n})\right] = D(P||Q)\ . $$
    \item If $A = A \otimes B$ and the cone is defined by a local symmetry, then
    $H_{\min}^{\cK}(A|B)_{P} = H_{\min}(\cK)$ and for $\ve \in (0,1),$
    $$\lim_{n \to \infty} \left[\frac{1}{n}H^{\cK,\ve}_{\min}(A^{n}|B^{n})_{P}\right] = H(A|B)_{P} \ . $$
\end{enumerate}
\end{theorem}
\begin{proof}
Any $X \in \cK$ may be written as $X = \sum_{k} \alpha_{k} \Pi_{k}$ where $\{\Pi_{k}\}_{k}$ are the minimal projectors, which are the extreme points of the Hermitian matrices under the twirling. It's then clear that $Q^{-1/2} \cdot Q^{-1/2}: \cK \to \cK$. Moreover, the maximum eigenvalue of $Q^{-1/2}PQ^{-1/2}$ will be an eigenvector of one of the projectors. Therefore, $\|Q^{-1/2}PQ^{-1/2}\|_{\cK} = \|Q^{-1/2}PQ^{-1/2}\|_{\infty}$. This completes the first item. The second item is just noting the $n-$fold case uses the $n-$fold minimal projectors. Lastly, the third item is as follows. Recall $D_{\max}^{\cK,\ve}(\rho||\sigma) = \underset{\wt{\rho} \in \Bve(\rho)}{\min}D^{\cK}_{\max}(\wt{\rho}||\sigma)$.  We can convert $P,Q$ to density matrices by the normalization property (Prop. \ref{prop:normalization}). Then note $\cG^{\otimes n}(\wt{\rho}) \in \Bve(\rho^{\otimes n})$ by data processing. By data-processing (Proposition \ref{prop:DPI}), we can conclude the optimal for smoothed is in this twirled set for each $n$ both in $\cK^{(n)}$ and $\cK^{(n)} = \Pos$, so they must be the same by Item 1. As the regularized smoothed max-relative entropy converges to $D(\rho||\sigma)$ \cite{Tomamichel-2015}, this completes the proof of Item 3. Item 4 follows from Corollary \ref{corr:local-symmetry-restriction}.
\end{proof}

\paragraph*{Entanglement Rank Cones - Operator Schmidt Norms}
In contrast to the previous setting, for the entanglement rank cones $\Ent_{k}$ whose norms are referred to as the operator Schmidt norms and denoted as $S(k)$-norms in \cite{Johnston-2012a}, nothing is well behaved. The first issue is that the structure to guarantee a CP map $\Phi$ satisfies $\Phi: \Ent_{k} \to \Ent_{k}$ for a fixed $k$, is unclear (note that Propositions \ref{prop:ent-rank-channel-property} and \ref{prop:k-ent-generating} are stronger and weaker properties respectively). Therefore, we have to rely on the too-strong property and use the rank of $Q$.
\begin{proposition}
Let $P \in \Pos(A \otimes B)$. Let $\rank(Q) = k$. Then for $n \in \mbb{N}$,
$$D_{\max}^{\Ent_{kn}}(P^{\otimes n}||Q^{\otimes n}) = \log\|(Q^{-1/2}PQ^{-1/2})^{\otimes n}\|_{S(kn)} \ . $$
\end{proposition}
\begin{proof}
If $\rank(Q) = k$, then $\rank({Q^{\otimes n}}^{-1/2}) = nk$. Therefore, by Proposition \ref{prop:ent-rank-channel-property}, ${Q^{\otimes n}}^{-1/2} P^{\otimes n} {Q^{\otimes n}}^{-1/2} = (Q^{-1/2}PQ^{-1/2})^{\otimes n} \in \Ent_{kn}(A:B)$.
\end{proof}
First note that for the above proof to be non-trivial, $P \ll Q$, which means $\rank(P) \leq k$ as well. Therefore, one can see this is quite restrictive, though one might expect it is significantly more general if $Q \in \Ent_{k}$ were to generally imply $Q^{-1/2}PQ^{-1/2} \in \Ent_{k}$. Beyond this, note we have not been able to prove additivity. This is because we would need, for the $n =2$ case, that $(\gamma^{\star},Y^{\star})$ in \eqref{eqn:KNormDual} to be such that ${Y^{\star}}^{\otimes 2} - Y^{\star} \otimes X - X \otimes Y^{\star} \in \cK^{\ast}$ where $X = Q^{-1/2}PQ^{-1/2}$. This seems unlikely in general. We can however prove additivity holds in the special case of $\cK = \PPT$, where we can write it in terms of a simple SDP.
\begin{proposition}\label{prop:PPT-max-relative-entropy-SDP}
$D^{\PPT}_{\max}(P||Q)$ is given by $\log(\alpha)$ where $\alpha$ is given by
\begin{center}
    \begin{miniproblem}{0.45}
      \emph{Primal}\\[-5mm]
      \begin{equation}
      \begin{aligned}\label{eqn:PPTMaxPrimal}
        \text{maximize:}\quad & \ip{P}{X} \\
        \text{subject to:}\quad & \langle Q,X \rangle \leq 1 \\
        & X, X^{\Gamma} \in \Pos(A) \ ,
      \end{aligned}
      \end{equation}
    \end{miniproblem}
    \begin{miniproblem}{0.45}
      \emph{Dual}\\[-5mm]
      \begin{equation}
      \begin{aligned}\label{eqn:PPTMaxKDual}
        \text{minimize:}\quad & \gamma  \\
        \text{subject to:}\quad & \gamma Q \geq P + Y^{\Gamma} \\
        & Y \geq 0 \ ,
      \end{aligned}
      \end{equation}
    \end{miniproblem}
 \end{center}
where $\cdot^{\Gamma}$ is the partial transpose map.
\end{proposition}
We refer the reader to the Appendix for the derivation. This then leads to additivity result in this special case.
\begin{theorem}\label{thm:PPTAdditivity}
Let $P,P',Q,Q' \in \PPT(A:B)$. Then,
$$D_{\max}^{\PPT}(P\otimes P' || Q \otimes Q') = D_{\max}^{\PPT}(P||Q) + D_{\max}^{\PPT}(P'||Q') \ . $$
\end{theorem}
\begin{proof}
Let $\alpha,\gamma,X,(\gamma,Y)$ and $\alpha',X',(\gamma',Y')$ be the optimal value, primal optimizer, and dual optimizer for $D_{\max}^{\PPT}(P||Q)$ and $D_{\max}^{\PPT}(P'||Q')$ respectively. Then $\langle Q \otimes Q', X \otimes X'\rangle = 1$ and $\langle P \otimes P', X \otimes X' \rangle = \alpha \cdot \alpha'$. Note this SDP has strong duality (given we proved the equivalent conic program did earlier), so it suffices to prove we can get the same optimal value for the dual program.
Now note
\begin{align*}
    \gamma Q \otimes \gamma' Q'
    \succeq & (P + Y^{\Gamma}) \otimes (P' + Y'^{\Gamma}) \\
    = & P \otimes P' + P \otimes Y'^{\Gamma} + Y^{\Gamma} \otimes P' + Y^{\Gamma} \otimes Y'^{\Gamma} \\
    = & P \otimes P' + \Gamma^{BB'}(\ol{P} \otimes \ol{P'} + \ol{P} \otimes Y' + Y \otimes \ol{P'} + Y \otimes Y') \ ,
\end{align*}
where the inequality is by assumption of them being optimizers and the last equality is by defining $\ol{P} := P^{\Gamma}, \ol{P'} := P'^{\Gamma}$ and using if you have fixed bases of $B,B'$ for PPT, then the basis for $BB'$ PPT has also been induced. Now note that by assumption $\ol{P},\ol{P'} \succeq 0$, so the argument of $\Gamma^{BB'}$ is positive semidefinite. Therefore it's a feasible operator for the dual program, which means $\gamma \cdot \gamma'$ is a dual optimal value. As $\gamma \cdot \gamma' = \alpha \cdot \alpha'$, we have proven the SDP is multiplicative and so $D_{\max}^{\PPT}$ is additive (as it is the $\log$ of the SDP).
\end{proof}
As a final remark on the framework of entanglement rank cones, we note it is known $\cv^{\Sep}(\Phi \otimes \Psi) = \cv^{\Sep}(\Phi) \cv^{\Sep}(\Psi)$ whenever $\Phi$ is entanglement-breaking \cite[Theorem 3]{Chitambar-2021a}. It is unlikely this generalizes beyond the separable cone and the Choi matrices as in general we have the following proposition.
\begin{proposition}
Let $P \in \Pos(A \otimes B)$. Then 
$$ H^{\Ent_{k}}_{\min}(A|B) = - \min_{\sigma \in \Density(A \otimes B)} \|\sigma_{B}^{-1/2}P_{AB}\sigma_{B}^{-1/2}\|_{\Ent_{k}} \ . $$
\end{proposition}
\begin{proof}
This follows from applying \eqref{eqn:K-min-entropy-simple} along with Proposition \ref{prop:K-Dmax-to-cone} using the fact $(\I_{A} \otimes \sigma_{B})^{-1/2} \cdot (\I_{A} \otimes \sigma_{B})^{-1/2}$ factors across the tensor product into local CP maps.
\end{proof}

\subsection{Special Case: Werner States}
As the operator Schmidt norms are difficult to handle in general, in this section we investigate how they behave for a class of operators that are manageable: the Werner states.

The Werner states of fixed local dimension $d$ may be parameterized in the following two ways
\begin{align}
    \rho_{\lambda} &:= \frac{2\lambda}{d(d+1)}\Pi_{+} + \frac{2(1-\lambda)}{d(d-1)}\Pi_{-} \, , \lambda \in [0,1] \label{eq:WH-defn-lambda} \\
    \rho_{\alpha} &:= \frac{1}{d(d-\alpha)}(\I^{AB} - \alpha \mbb{F})  \, , \alpha \in [-1,1] \label{eq:WH-defn-alpha} \ ,
\end{align}
where $\Pi_{+}, \Pi_{-}$ are the symmetric and anti-symmetric projectors respectively, $\mbb{F}$ is the swap operator, and they may be related by $\Pi_{+} = \frac{1}{2}(\I^{AB} + \mbb{F})$. The two parameterizations may be related by the conversions
\begin{align}\label{eq:WH-param-conversion}
    \alpha = \frac{(1-2\lambda)d+1}{1-2\lambda+d} \quad \quad \lambda = \frac{(1+d)(1-\alpha)}{2(d-\alpha)} \ .
\end{align}
These are known to be entangled for $\lambda \in [0,1/2)$, or equivalently $\alpha \in (1/d,1]$.

As one might expect from their ability to be decomposed into symmetric and asymmetric projectors, these are the states that invariant under the local symmetry $\cU(\cdot) := \int (U \otimes U) \cdot (U \otimes U)^{\ast} \, dU$ where $dU$ is the Haar measure. This is often referred to as being $UU$-invariant. These states are known to be excellent for admitting non-additivity and non-multiplicativity results. In \cite{Zhu-2011a} it was shown their geometric measure of entanglement is super-multiplicative for $n=2$ and the same was shown for the communication value, $\cv^{\Sep}$, of the Werner-Holevo channels $\cW_{\lambda}$ whose Choi operators are given by $d\rho_{\lambda}$ in \cite{Chitambar-2021a}. This is equivalent to the non-additivity of $H_{\min}^{\Sep}$ for the Choi operators of the Werner-Holevo channels by \eqref{thm:generalized-cv}. Here we expand on these properties. 

First we show the non-multiplicativity of the Werner-Holevo channels for $\cv^{\Sep}$ is in fact a property of the non-additivity of the non-multiplicativity of the $S(1)$ norm.
\begin{proposition}\label{prop:WH-min-ent-simplification}
Let $k,n \in \mbb{N}$ and $\cK = \Ent_{k}(A^{n}:B^{n})$. Then \begin{align*}
    H_{\min}^{\cK}(A^{n}|B^{n})_{\rho^{\otimes n}_{\alpha}} = & -D^{\cK}_{\max}(\rho_{\alpha}^{\otimes n}||\I_{A^{n}} \otimes \pi_{B^{n}}) \\
    =& -\log \|\rho_{\alpha}^{\otimes n}\|_{\Ent_{k}(A^{n}:B^{n})} +n\log(d) \ .
\end{align*}
\end{proposition}
\begin{proof}
As $\Ent_{k}(A^{n}:B^{n})$ is invariant under local maps, by the data-processing inequality \eqref{prop:DPI},
\begin{align*}
    D_{\max}^{\Ent_{k}}(\rho^{\otimes n}_{\alpha}||\I \otimes \sigma) 
    \geq & D_{\max}^{\Ent_{k}}(\cU^{\otimes n}(\rho^{\otimes n}_{\alpha})||\, \cU^{\otimes n}(\I \otimes \sigma)) \\
    = & D^{\Ent_{k}}_{\max}(\rho_{\alpha}^{\otimes n}||\I \otimes \int U \sigma U^{\ast}\, dU) \\
    = & D^{\Ent_{k}}_{\max}(\rho_{\alpha}^{\otimes n}||\I \otimes \pi) \ ,
\end{align*}
where the $\sigma$ was arbitrary. Recalling that
$$H^{\Ent_{k}}_{\min}(A^{n}|B^{n})_{\rho_{\alpha}^{\otimes n}} = \max_{\sigma \in \Density(B^{n})} -D_{\max}^{\Ent_{k}}(\rho^{\otimes n}_{\alpha}||\I \otimes \sigma) \ , $$
it follows that
$$ H^{\Ent_{k}}_{\min}(A^{n}|B^{n})_{\rho_{\alpha}^{\otimes n}} = - D^{\Ent_{k}}_{\max}(\rho^{\otimes n}_{\alpha}||\I \otimes \pi) \ . $$
Using that $\I \otimes \pi = \frac{1}{d^{n}}(\I \otimes \I)$ along with the normalization property (Proposition \ref{prop:normalization}) along with $D_{\max}^{\cK}(X||\I) = \|X\|_{\cK}$ gets the second equality.
\end{proof}
This simplifies things as all we need to do is determine $\|\rho^{\otimes n}_{\alpha}\|_{\Ent_{k}(A^{n}:B^{n})}$. In general, this is not known. However, it is known that the space of operators that are $UU-$invariant is the same as the space of PPT operators with the same invariance \cite{Vollbrecht-2001a}. Therefore, using the primal problem \eqref{eqn:KNormPrimal} we ultimately obtain the following linear program:
\begin{proposition}\label{prop:WH-Ent-norm-LP}
$\|\rho_{\lambda}^{\otimes 2}\|_{\Ent_{1}(A^{2}_{1}:B^{2}_{1})}$ is given by the linear program:
\begin{align*}
    \max. \; & \; \lambda^{2}w + \lambda(1-\lambda)(x+y) + (1-\lambda)^{2}z \\
    \mrm{s.t.} \; & \; d_+ (w+y) \geq d_- (x+z) \\
    &\; d_+(w+x) \geq d_- (y+z) \\
    &\; d_+^{2} w + d_{-}^{2}z \geq d_+ d_- (x+y) \\
    &\; \left[d_+^2 w + d_+ d_- (x+y) + d_-^2 z\right] = \frac{4}{d^2} \\
    &\; w,x,y,z \geq 0 \ ,
\end{align*}
where $d_+ \equiv d+1$, $d_- \equiv d-1$, and we note it can be reparameterized to $\rho_{\alpha}$ using \eqref{eq:WH-param-conversion}.
\end{proposition}
Something in effect equivalent was derived in \cite{Chitambar-2021a}, but we provide our derivation in the Appendix for completeness. In the single copy case it was shown in \cite[Proposition 5.2.10]{Johnston-2012a} that the Werner states have a simple expression for their entanglement rank norm:
\begin{equation}\label{eq:single-copy-werner-norm}
\|\rho_{\alpha}\|_{\Ent_{k}(A:B)} = 
\begin{cases}
    \frac{1 + \min\{\alpha,0\}}{d(d-\alpha)} & k = 1 \\
    \frac{1+|\alpha|}{d(d-\alpha)} & 2 \leq k \leq n \ .
\end{cases} 
\end{equation}
Combining these two propositions, we are able to plot tight bounds of the non-multiplicativity of the $S(1)$ norm (equivalently $\Ent_{1} \cong \Sep$ norm) as well as show that there is no positive gap of $\|\rho_{\lambda}^{\otimes 2}\|_{S(1)} - \| \rho_{\lambda} \|^{2}_{S(k\geq2)}$ (See Fig. \ref{fig:Non-Mult-of-Schmidt-Norm}).

\begin{figure}[H]
    \centering
    \includegraphics[width = 0.7\columnwidth]{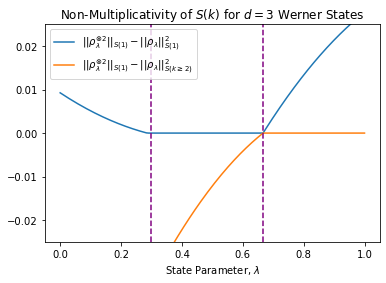}
    \caption{Non-Multiplicativity of the Werner State's Operator Schmidt Norm. The dotted purple lines are at $3/10$ and $2/3$ respectively. The first is when non-multiplicativity dies out the first time. The second is when non-multiplicativity comes back. Note $\rho_{2/3} = \pi$ for $d = 3$.}
    \label{fig:Non-Mult-of-Schmidt-Norm}
\end{figure}

As a final remark, we note that when proving the cone-restricted max-divergence anti-standard AEP, we used that $D_{\max}^{\cK}(\psi||\I) < D(\psi||\I) = D_{\max}(\psi||\I)$. Here we show that such a property holds in great generality if we replace $\psi$ with a Werner state. Specifically, we show that so long as the Werner state is entangled, there exists a dimension such that this misordering for both the $\Sep$-restricted min-entropy and max-relative entropy holds. As the proof is not particularly illuminating, it is presented in the appendix.
\begin{theorem}\label{thm:WH-separation}
(Note that in this setting, the $\PPT$ and $\Sep$ cones are equivalent.)
For $d \in\mathbb{N}$, $\lambda = 0$ is such that $$D^{\PPT}_{\max}(\sigma_{\lambda}||\I_{A} \otimes \pi) < D(\sigma_{\lambda}||\I_{A} \otimes \pi) \ . $$
Moreover this never will hold for separable states, but there always exists sufficiently large $d$ such that 
$D^{\PPT}_{\max}(\sigma_\lambda || \I_{A} \otimes \pi) < D(\sigma_{\lambda}||\I_{A} \otimes \pi)$ for $\lambda < 1/2$ (i.e. any entangled state).
\end{theorem}

\section{Conclusion}
In this work we have investigated the general structure of one-shot information theory measures when we consider the optimization programs for cones other than the positive semidefinite cone. In doing so, we have even further cemented the importance of the positive semidefinite cone and the L\"{o}wner order. We have done this by showing that for restricted cones the asymptotic equipartition properties cease to hold and this effects what tasks they can properly capture, e.g. failing to capture hypothesis testing or successfully capturing minimal-error decoding. Moreover, we have captured this breakdown in terms of resources, in particular coherence and entanglement. We have also shown the separable cone, the minimally cone invariant under local unitaries, is sufficent for information theory over CQ states. This exemplifies the strong distinction between fully quantum and partially quantum tasks that has been observed in the literature.

After we explored the measures behaviour with respect to recovering i.i.d.\ limits, we established operational interpretations of the cone restricted entropic quantities. First, we established that the restricted min-entropy captures a notion of entanglement-assisted  classical communication over a quantum channel under restricted measurements. Second, we introduced a further restricted min-entropic-like quantity that captures the ability to convert one channel to another via pre-processing. This also exemplifies how other majorization results may be generalized to restricted classes of transformations. 

Finally, we introduced the relationship between restricted max-relative entropy and the generalized cone operator norm. In doing so, we showed for an Abelian symmetry, the cone restricted to said Abelian symmetry behaves like the max-relative entropy on states that satisfy the symmetry. We showed for entanglement rank cones, the same property is unclear. It would be interesting and possibly enlightening to know if such a property does hold or not for these entanglement rank cones and why. 

One further way to extend this work would be to extend it into the category-theoretic language of \cite{Jencova-2021a}. It is possible further generality and refinements could be achieved in that framework. Moreover, as noted in the main text, \cite{Jencova-2021a} only considered $H^{\cK}_{\min}(B|A)$, so extending $H_{\min}^{\cK}(A|B)$ into a category-theoretic framework would be new in and of itself.

Beyond this, we note a tangential avenue of research pertaining to superchannels given this work. First, we introduced the flipped extended min-entropy $H^{\ext}_{\min}(A|B)$, but we did not present an operational interpretation beyond it's generalization of singlet fraction. It could be interesting to establish one. We also proved an anti-standard AEP for both forms of extended min-entropy, but what asymptotic measure an extended min-entropy AEP would  converge to has not yet been established. Establishing this would also be of interest. Lastly, in proving our anti-standard AEP for extended min-entropy, we established the notion of $k-$entanglement-breaking supermaps. It is likely that there is more to be said about this general class of supermaps.

\section*{Data Availability} 
The code for generating the data and plot for Fig.\ \ref{fig:Non-Mult-of-Schmidt-Norm} may be found at \href{https://github.com/ChitambarLab/ConeRestrictedInformationTheory}{the Github repository for this project}.

\section*{Acknowledgment}
E.C. acknowledges NSF Award 2112890 for supporting this project. I.G. acknowledges the support of an Illinois Distinguished Fellowship. The authors also thank the anonymous reviewers of BIID10 for making clear previous terminology was confusing.

\section{Appendix}
\subsection{Coherence Resource Theory Proofs}

\begin{proof}[Proof of Prop. \ref{prop:classical-to-classical-Choi}]
In the bipartite case, considering Hilbert spaces, $A \cong \mbb{C}^{d},B \cong \mbb{C}^{d'}$, the preferred basis must be a product basis $\{\ket{ji} = \ket{j}\ket{i}\}_{j \in [d],i \in [d']}$. We are therefore considering $J_{\Phi} = \sum_{i,j} \alpha_{ij} \dyad{i}_{A} \otimes \dyad{j}_{B}$. By the trace preserving condition, $\I^{A} = \Tr_{B}(J_{\Phi}) = \sum_{i,j} \alpha_{i,j} \dyad{i}$. Which implies, $\sum_{j} \alpha_{i,j} = 1$. So we can define $\alpha_{i,j} = p(i|j)$. Using $\vec^{-1}(\sqrt{\alpha_{ij}} \ket{i}\ket{j}) = \sqrt{\alpha_{ij}}\ket{j}\bra{i}$, which is well known to be the structure of the Kraus operators of a classical-to-classical channel. Thus by \eqref{eqn:choi-in-kraus-terms}, we see we have a classical-to-classical channel. The other direction is straightforward from this one. Lastly, note that is also implies the adjoint map's Choi operator is diagonal in the same preferred basis, as swapping the spaces and taking the transpose will preserve this structure.
\end{proof}

\begin{proof}[Proof of Prop. \ref{prop:PIO-Entanglement-Rank}]
First note by the vec mapping identity \eqref{eqn:vec-map-identity},
$$ \vec(U_{i}\Pi_{i}) = \vec(U_{i}\I_{A}\Pi_{i}) = (\Pi_{i} \otimes U_{i})\vec(I_{A}) \, ,$$
where we have assumed the transpose is defined in the incoherent basis so that $\Pi_{i}^{\Trans} = \Pi_{i}$.
By \eqref{eqn:choi-in-kraus-terms} and the definition of a PIO, any $\Phi$ that is a PIO has a Choi operator of the form
\begin{align*}
    J_{\Phi}^{AA} =& \sum_{i} \vec(U_{i}\Pi_{i})\vec(U_{i}\Pi_{i})^{\ast}
    = \sum_{i} (\Pi_{i} \otimes U_{i})\vec(I^{A})\vec(I^{A})^{\ast}(\Pi_{i}^{\ast} \otimes U_{i}^{\ast}) \ .
\end{align*}
Recall that $\vec(I^{A}) = \sqrt{d} \ket{\tau_{d}}$ where we will assume that $\tau_{d}$ is defined in terms of the incoherent basis.\footnote{This is without loss of generality as we are already considering the Choi operator in a basis-dependent fashion.} It follows by the definition of incoherent projectors,
\begin{align*}
(\Pi_{i} \otimes \id^{A})\vec(I^{A})
= &\sqrt{d} \sum_{j \in S_{i}} \op{j}{j} \sum_{j' \in [d]} \ket{j'}\ket{j'}\\ 
=&\sqrt{d} \sum_{j \in S_{i}, j' \in [d]} \delta_{j,j'} \ket{j}\ket{j'} \\
=& \sqrt{d} \sum_{j \in S_{i}} \ket{j}\ket{j} = \sqrt{d} \ket{\tau_{|S_{i}|}} \, .
\end{align*}
As $\mrm{SR}(\ket{\tau_{k}}) = k$, we have that $\mrm{SR}(\Pi_{i} \otimes U_{i})(\vec(\I^{A})) = |S_{i}|$ where we have also used that the local unitary $U_{i}^{A}$ won't change the Schmidt rank. Therefore, the Schmidt rank of $J_{\Phi}$ is upper bounded by $\max_{i} \{\rank(\Pi_{i})\}$. Moreover, note that by definition of the unitary, one has 
$$ (\Pi_{i} \otimes U_{i})(\vec(\I_{A})) = \sum_{j \in S_{i}} e^{i \theta_{j}}\ket{j}\ket{\sigma_{i}(j)} \ . $$
Since $\sigma_{i}$ is a permutation, the product vectors in this sum are mutually orthogonal, so the tensor rank is $|S_{i}|$. Furthermore, as the projectors are mutually orthogonal, $(\Pi_{i} \otimes U_{i})(\vec(\I_{A}))$ and $(\Pi_{i'} \otimes U_{i'})(\vec(\I_{A}))$ share no simple tensors, so the minimal rank needed to construct the Choi matrix is $\max_{i} \{\rank(\Pi_{i})\}$. This proves the equality.

To get the operational interpretation, note that the unitary relabels the basis but does not destroy coherence whereas the projector does. Therefore $\max_{i}\{\rank(\Pi_{i})\}$ is the size of the largest subspace in which coherence may be preserved. This completes the proof.
\end{proof}

\subsection{Extended Min-Entropy Proofs}
\textit{Note: Throughout this section, the spaces of an object are superscripts or subscripts in whatever manner we believe will be cleanest for the reader at the expense of consistency.}\\
We begin with further properties of the space of supermaps, $\mbb{L}(A,B)$ where $A :=(A_0A_1), B:=(B_0B_1)$, that were alluded to or skipped over in the main text but are needed here. Throughout this section, we will drop tensor product symbols when they seem straightforward. Recall from the main text that supermaps satisfy their own Choi-Jamiolkowski isomorphism, $\mbb{L}(A,B) \to \mbf{J}_{\Theta}^{AB} \in \Lin(A_{0} A_{1} B_{0} B_{1})$. Noting the standard Choi-Jamiolkowski isomorphism, this implies there are also isomorphisms from $\mbb{L}(A,B)$ to $\Trans(A,B), \Trans(A_0B_1,A_1B_0),$ and $\Trans(A_1B_0,A_0B_1)$. Define the maps $\Delta_{\Theta} \in \Trans(A,B), \Gamma_{\Theta} \in \Trans(A_1B_0,A_0B_1)$ to be the unique maps such that $J_{\Delta_{\Theta}} = \mbf{J}^{AB}_{\Theta} = J_{\Gamma_{\Theta}}$ where the uniqueness follows from the Choi-Jamiolkowski isomorphism. Moreover, the isomorphism $\Theta \to (\mathds{1}^{A} \otimes \Theta)[\Upsilon^{A\wt{A}}]$ which must also have the same Choi operator, $\mbf{J}^{AB}_{\Theta}$. This leads to the following straightforward proposition recognized in the previous literature.
\begin{proposition}
The following are equivalent
\begin{enumerate}
    \item $\Theta \in \mbb{L}(A,B)$ is CP preserving (CPP)
    \item The dual map (a.k.a.\ adjoint map) $\Theta^{\ast}$ is CPP
    \item $\mbf{J}^{AB}_{\Theta} \geq 0$
    \item $\Delta_{\Theta}$ is CP
    \item $\Gamma_{\Theta}$ is CP
\end{enumerate}
\end{proposition}
\begin{proof}
($1 \to 3$) Let $\Theta$ be CPP. Then $(\mathds{1}_{A} \otimes \Theta)[\Upsilon^{A\wt{A}}]$ is CP as $\Upsilon^{A\wt{A}}$ is CP. Using that \cite[Eqns. 12,13]{Gour-2019a},
$$\mbf{J}^{AB}_{\Theta} = \id^{A_0B_0} \otimes (\mathds{1}^{A} \otimes \Theta)[\Upsilon^{A\wt{A}}](d_{A_{0}}\tau_{d_{A_{0}}} \otimes d_{B_{0}}\tau_{d_{B_{0}}}) \ , $$
we conclude $\mbf{J}^{AB}_{\Theta} \geq 0$. ($3 \to 4 \& 5$) Any map with $\mbf{J}^{AB}_{\Theta}$ as its Choi operator must be CP. These relations worked backwards prove $1 \leftrightarrow 3 \leftrightarrow 4 \leftrightarrow 5$. ($2 \leftrightarrow 3$) Recalling \eqref{eqn:superchannel-dual-map-choi} and noting that the swapping of spaces and the conjugate map are both positivity preserving, $\mbf{J}^{AB}_{\Theta} \geq 0$ if and only if $\mbf{J}^{BA}_{\Theta^{\ast}} \geq 0$. This completes the proof.
\end{proof}

One can also define the action a superchannel on a channel solely in terms of the Choi operators:
\begin{align}\label{eqn:Action-of-Superchannel-Choi}
    J^{B}_{\Theta[\Psi]} = \Tr_{A}(\mbf{J}^{AB}_{\Theta}((J_{\Psi}^{A})^{T} \otimes I_{B}) \ ,
\end{align}
which recalling the action of a map in terms of its Choi operator, implies $\Delta_{\Theta}(J_{\Psi}^{A}) = J^{B}_{\Theta[\Psi]}$. It also allows us to prove the following about the unital preserving maps, which uses the same proof method as is used to prove the TPP conditions of a supermap (e.g.\ embedded in the proof of \cite[Theorem 1]{Gour-2019a}) and was proven in the special case that $\Theta$ is a superchannel in \cite{Gour-2019a}.
\begin{proposition}\label{prop:unital-preserving}
Let $\Theta \in \mbb{L}(A,B)$. $\Theta$ is unital preserving if and only if $\mbf{J}^{AB}_{\Theta}$ has the marginals 
    $$ \mbf{J}^{A_{0}B_{1}}_{\Theta} = I_{A_{0}B_{1}} \quad \mbf{J}^{AB_{1}}_{\Theta} = \mbf{J}^{A_{1}B_{1}}_{\Theta} \otimes d_{A_{0}}^{-1}I_{A_{0}}$$
\end{proposition}
\begin{proof} ($1 \to 2$) Let $\Psi^{\ast} \in \Trans(A_0,A_1)$ be unital. Let $\Phi^{\ast} := \Theta[\Psi^{\ast}]$. Let $Z_{B_{1}} \in \Lin(B_{1})$. Multiplying both sides of \eqref{eqn:Action-of-Superchannel-Choi} by $\I_{B_{0}} \otimes ({Z_{B_{1}}})^{\Trans}$ and taking the trace, 
\begin{align}
    & \Tr[J^{B}_{\Phi^{\ast}}(\I_{B_{0}} \otimes {Z_{B_{1}}})^{\Trans}] 
    =\Tr[\mbf{J}^{AB}_{\Theta}((J^{A}_{\Psi^{\ast}})^{\Trans} \otimes (I_{B_{0}} \otimes {Z_{B_{1}}})^{\Trans}] \nonumber \\
    \Rightarrow & \Tr[Z^{B_{1}}] = \Tr[\mbf{J}^{AB_{1}}_{\Theta^{\ast}}(J^{A}_{\Psi^{\ast}} \otimes {Z_{B_{1}}})^{\Trans}] \label{eqn:unital-equality-1}\ ,
\end{align}
where in the implication we on the L.H.S. we used that $\Tr_{B_{0}}(J^{B}_{\Phi^{\ast}}) = \I_{B_{1}}$, as it is unital by our assumptions on $\Theta$ and $\Psi^{\ast}$, and on the R.H.S. we have used the adjoint map of $(\cdot) \otimes \I_{B_{0}}$ is the partial trace. It follows \eqref{eqn:unital-equality-1} holds for any unital $\Psi^{\ast} \in \Trans(A_{0},A_{1})$ and $Z_{B_{1}} \in \Lin(B_{1})$. It follows that it holds for $J^{A}_{\Psi^{\ast}} = \frac{1}{d_{A_{0}}}I_{A}$, so by linearity
$$ \Tr[\mbf{J}^{AB_{1}}_{\Theta^{\ast}}(W_{A} \otimes {Z_{B_{1}}})^{\Trans}] = 0 \ , $$
for all $W_{A} := J_{\Psi^{\ast}}^{A} - \frac{1}{d_{A_{0}}}I_{A}$ and $Z_{B_{1}} \in \Lin(B_{1})$. Note that $W_{A} \in \Lin(A)$ such that $W_{A_{1}} = 0$ if and only if there exists $\Psi^{\ast}$ such that $W = J_{\Psi^{\ast}}^{A} - \frac{1}{d_{A_{0}}}I_{A}$. That is the, above equation must hold for all $W \in \{W \in \Lin(A) : W_{A_{1}} = 0 \}$ and all $Z_{B} \in \Lin(B_{1})$. We can therefore conclude $\mbf{J}^{AB_{1}}_{\Theta} = \mbf{J}^{A_{1}B_{1}} \otimes \frac{1}{d_{A_{0}}}I_{A_{0}}$. Plugging this into \eqref{eqn:unital-equality-1}, one obtains 
$$ \Tr[Z_{B_{1}}] = \frac{1}{d_{A_{0}}}\Tr[\mbf{J}_{\Theta^{\ast}}^{B_{1}}(Z_{B_{1}})^{\Trans}] \ , $$
where to get the R.H.S we have used the unitality of $\Psi^{\ast}$, i.e.\ that $\Tr_{A_{0}}J_{\Psi^{\ast}}^{A} = \I_{A_{1}}$. As this holds for all $Z_{B_{1}} \in \Lin(B_{1})$, we can conclude $\mbf{J}^{B_{1}} = d_{A_{0}}I_{B_{1}}$, so $\mbf{J}^{A_0 B_1} = \I_{A_0 B_1}$.\\
($2 \to 1$) We start directly from \eqref{eqn:Action-of-Superchannel-Choi} with an extra partial trace:
\begin{equation}\label{eq:converse-unital-supermap}
\begin{aligned}
    \Tr_{B_{0}}J^{B}_{\Phi^{\ast}} =& \Tr_{AB_{0}}[\mbf{J}^{AB}_{\Theta}(J_{\Psi^{\ast}})^{\Trans} \otimes \I_{B}] \\
    =& \Tr_{A}[\mbf{J}^{AB_{1}}_{\Theta}(J^{A}_{\Psi^{\ast}})^{\Trans} \otimes \I_{B_{1}}] \\
    =& d_{A_{0}}^{-1}\Tr_{A}[\mbf{J}^{A_{1}B_{1}}_{\Theta} \otimes \I^{A_{0}} (J_{\Psi^{\ast}})^{\Trans} \otimes \I_{B_{1}}] \\
    =& d_{A_{0}}^{-1}\Tr_{A_{1}}[\mbf{J}^{A_{1}B_{1}}_{\Theta}] \\
    =& d_{A_{0}}^{-1}\mbf{J}^{B_{1}}_{\Theta} \\
    =& \I_{B_{1}} \ , 
\end{aligned}
\end{equation}
where the fourth equality is using that $\Psi^{\ast}$ is unital, and the last equality is using $\mbf{J}^{A_{0}B_{1}}_{\Theta} = I_{A_{0}}I_{B_{1}}$ implies $\mbf{J}^{B_{1}}_{\Theta} = d_{A_{0}}I_{B_{1}}$. As \eqref{eq:converse-unital-supermap} expresses the condition on a Choi operator to correspond with a unital map, this completes the proof.
\end{proof}
Of course by effectively the same proof method one obtains the well-known Choi operator constraints for a supermap to be trace-preserving-preserving.
\begin{proposition}\label{prop:TPP-conditions}
Let $\Theta \in \mbb{L}(A,B)$. $\Theta$ is trace-preserving preserving if and only if $\mbf{J}^{AB}_{\Theta}$ has the marginals
$$ \mbf{J}^{A_{1}B_{0}}_{\Theta} = I_{A_{1}B_{0}} \quad \mbf{J}^{AB_{0}}_{\Theta} = \mbf{J}^{A_{0}B_{0}}_{\Theta} \otimes d_{A_{1}}^{-1}I_{A_{1}} \ . $$
\end{proposition}
From these, we are able to get the following corollary.
\begin{corollary}\label{corr:set-of-dual-maps}
The set of supermaps that are duals of superchannels and the set of CPP unital-preserving supermaps are distinct.
\end{corollary}
\begin{proof}
We construct a superchannel whose dual map is not unital-preserving. Consider the channel that traces out the input and prepares the pure state $\dyad{0}$, $\Phi(X) := \Tr[X_{B_{0}}]\dyad{0}_{B_{1}}$. It is straightforward to check its Choi operator is $J_{\Phi}^{B} = \I_{B_{0}} \otimes \dyad{0}_{B_{1}}$. Now consider the linear map $\Delta_{\Theta}(X) := \frac{1}{d_{A_{0}}}\Tr[X]J_{\Phi}^{B}$. Again by straightforward calculation, the Choi operator of $\Delta_{\Theta}$ is $\frac{1}{d_{A_{0}}} I_{AB_{1}} \otimes \dyad{0}_{B_{1}}$, which is also the Choi operator of $\Theta \in \mbb{L}(A,B)$ that ignores the input channel completely and runs $\Phi$ instead. Then one can verify that $\mbf{J}_{\Phi} \geq 0$ and satisfies the constraints in Proposition \ref{prop:TPP-conditions}. Recalling $\mbf{J}_{\Theta^{\ast}}^{BA} = \overline{\mrm{SWAP}(\mbf{J}^{AB}_{\Theta})}$, $\mbf{J}^{BA}_{\Theta^{\ast}} = \frac{1}{d_{A_{0}}}\dyad{0}^{B_{0}} \otimes \I_{B_{1}A}$, which means in effect it hasn't changed. 

Using Proposition \ref{prop:unital-preserving}, $\Theta^{\ast}$ is unital preserving only if $\mbf{J}^{BA_{1}}_{\Theta^{\ast}} = \mbf{J}^{B_{1}A_{1}}_{\Theta^{\ast}} \otimes d_{B_{0}}^{-1}\I_{B_{0}}$, where we note this is because $\Theta^{\ast} \in \mbb{L}(B,A)$, so the spaces have been swapped from in the proposition. Noting that $\mbf{J}^{B_{1}A_{1}}_{\Theta^{\ast}} = \I_{B_{1}A_{1}}$,
\begin{align*}
\mbf{J}^{BA_{1}}_{\Theta^{\ast}} = \I^{B_{0}} \otimes \dyad{0}^{B_{1}} \otimes \I_{A_{1}} 
= \mbf{J}^{B_{1}A_{1}}_{\Theta^{\ast}} \otimes \I_{B_{0}}
\neq \mbf{J}^{B_{1}A_{1}}_{\Theta^{\ast}} \otimes d_{B_{0}}^{-1} \I_{B_{0}} \ .
\end{align*}
Thus we have a superchannel whose dual map is not unital preserving. This completes the proof.
\end{proof}
As our final property of dual maps we prove \eqref{eqn:entanglement-super-unital}.
\begin{proof}[Proof of \eqref{eqn:entanglement-super-unital}]
This is almost identical to \cite[Eqn. 66]{Gour-2019a}. Let $\Phi \in \Channel(A_{0}A_{1},B_{0}B_{1})$. Then,
\begin{align*}
    d_{A_{0}}\exp(-H^{\ext}_{\min}(A|B)_{\Phi})
    =& \underset{\Theta^{\ast} \in \wt{\mbb{S}}^{AB}}{\max} \langle (\mathds{1}_{A} \otimes \Theta^{\ast})[\Upsilon^{A\wt{A}}], \Phi^{AB} \rangle \\
    =& \underset{\Theta^{\ast} \in \wt{\mbb{S}}^{AB}}{\max} \langle \Upsilon^{A\wt{A}}, (\mathds{1}_{A} \otimes \Theta)[\Phi^{AB}] \rangle \\
    =& \underset{\Theta \in \mbb{S}^{BA}}{\max} \langle \Upsilon^{A\wt{A}}, (\mathds{1}_{A} \otimes \Theta)[\Phi^{AB}] \rangle \\ 
    =& \underset{\Theta \in \mbb{S}^{BA}}{\max} \langle \tau_{A_{1}}, (\mathds{1}_{A} \otimes \Theta)[\Phi^{AB}](\tau^{A_{0}})\rangle \ ,
\end{align*}
where the third equality is because the dual maps $\wt{\mbb{S}}^{AB}$ are the dual maps of $\mbb{S}^{AB}$, and the last equality is using that the inner product for channels is defined by $\langle \Phi, \Psi \rangle = \sum_{a} \langle \Phi[X_{a}],\Psi[X_{a}]\rangle$ for an orthonormal basis $\{X_{a}\}_{a}$ and we used the choice $X_{1} = \frac{1}{d_{A_{0}}}(\tau^{A_{0}})$.
\end{proof}

We now move on to considering superchannels specifically. Beyond the algebraic properties of supermaps encoded in their Choi operators, the space of superchannels specifically admits the property that they can always be written in terms of pre- and post-processing with an ancillary memory.
\begin{proposition}[\cite{Chiribella-2008a}]\label{prop:superchannel-equivalence}
$\Theta$ is a superchannel if and only if there exists Hilbert space $E$ with $d_{E} \leq d_{A_{0}}d_{B_{0}}$ and $\Phi_{pre}^{B_{0} \to A_0E} \in \Channel(B^{0},A_0E)$, $\Phi_{post}^{A_1E \to B_1} \in \Channel(A_1E,B_1)$ such that for all $\Psi \in \Trans(A_{0},A_{1}$,
$$ \Theta[\Psi] = \Phi_{post}^{A_1E \to B_1} \circ (\Psi \otimes \id^{E}) \circ \Phi_{pre}^{A_{0} \to B_0E} \ . $$
\end{proposition}

\subsection{Lemmas for Superchannels without Memory}
In the speical case there is no ancillary space $E$, then the description of supermaps simplifies to the following form.

\begin{proposition}\label{prop:pre-and-post-proc-choi}
Given $\Xi \in \Trans(A,B), \Phi \in \Trans(B,C), \Psi \in \Trans(C,D)$.
\begin{align*}
    J_{\Psi \circ \Phi \circ \Xi} =(T \circ \Xi^{\ast} \circ T \otimes \Psi)J_{\Phi}
\end{align*}
Moreover, if $\Xi$ is Hermitian-preserving, $J_{\Psi \circ \Phi \circ \Xi} =(\Xi^{\Trans} \otimes \Psi)J_{\Phi}$ where $\Xi^{\Trans}$ is just the linear transformation defined by using the transpose of the Kraus operators, i.e.\ $\sum_{k}A_{k} X A_{k}^{\ast} \mapsto \sum_{k}A_{k}^{\Trans} X \ol{A_{k}}$.
\end{proposition}
\begin{proof}
(Note this proof heavily relies on our choice of vec mapping \eqref{eqn:vec-map-identity}.)
Consider $\Xi \in \Trans(A,B), \Phi \in \Trans(B,C), \Psi \in \Trans(C,D)$. Then 
\begin{align*}
 (\Psi \circ \Phi \circ \Xi)(X) &= \sum_{a,b,c} A_{a}\widehat{A}_{b}\wt{A}_{c}XB_{c}^{\ast}\widehat{B}_{b}^{\ast}\wt{B}_{a}^{\ast} \ , 
 \end{align*}
 which implies
 \begin{align*}
 J_{\Psi \circ \Phi \circ \Xi} 
 =& \sum_{(a,b,c)} \vec(A_a \widehat{A}_{b} \wt{A}_{c})\vec(B_{a}\widehat{B}_{b}\wt{B}_{c})^{\ast} \\
=& \sum_{(a,b,c)} (\wt{A}_{c}^{T} \otimes A_{a})\vec(\widehat{A}_{b})\vec(\widehat{B}_{b})^{\ast}( \ol{\wt{B}}_{c} \otimes B_{a}) 
= (T \circ \Xi^{\ast} \circ T \otimes \Psi)J_{\Phi}
\end{align*}
where the first line is to make clear our choice of expressing the Kraus decompositions, the second line is the relation between Kraus decompositions and Choi matrices, the third is Eq. \eqref{eqn:vec-map-identity}, and the fourth is by noting
$ \sum_{c} \wt{A}^{T}_{a} X \ol{\wt{B}}_{a} = (\sum_{c} \wt{B}^{\ast}_{c} X^{T} \wt{A}_{c})^{T} = (T \circ \Xi^{\ast} \circ T)(X) \ . $ 

Now if $\Xi$ is Hermitian preserving, then $\Xi^{\ast}$ is. Recall that $J_{\Phi} = \sum_{a,b} E_{a,b} \otimes \Phi(E_{a,b})$ and $T(E_{a,b}) = E_{b,a} = E_{a,b}^{\ast}$, so $(\Xi^{\ast} \circ T)(E_{a,b}) = (\Xi^{\ast}(E_{a,b}))^{\ast}$, so in this case, we have $J_{\Psi \circ \Phi \circ \Xi} = (C \circ \Xi^{\ast} \otimes \Psi)J_{\Phi}$ where $C$ applies the entry-wise conjugate on $\Lin(A)$. Noting that $\ol{A_{k}^{\ast} E_{a,b} A_{k}} = \ol{A_{k}^{\ast}} \, \ol{E_{a,b}} \, \ol{A_{k}} = A_{k}^{\Trans} E_{a,b} \ol{A_{k}}$, we just define the transpose of a map as in terms of it's Kraus operators which completes the proof. 
\end{proof}
\begin{lemma}\label{lemma:trans-channel-props}
For $\Phi \in \Channel(A,B)$, $\Phi^{\Trans}(B,A) \in \CPU(B,A)$. Moreover, if $\Phi \in \CPTPU(A,B)$, then $\Phi^{\Trans}(B,A) \in \CPTPU(B,A)$.
\end{lemma}
\begin{proof}
As shown in the previous proof, $\Phi^{\Trans} = C \circ \Phi^{\ast}$ where $C \in \Trans(A,A)$ is the entry-wise conjugate. $C$ is clearly completely positive, so $\Phi^{\Trans} \in \CP(B,A)$ if $\Phi \in \CP(A,B)$. $\Phi^{\ast} \in \mrm{Unital}(B,A)$ if $\Phi \in \mrm{TP}(A,B)$. The entry-wise conjugate maps the identity operator to the identity operator, so $\Phi^{\Trans}$ is unital. Similarly, if $\Phi$ is unital, then $\Phi^{\ast}$ is trace-preserving. $C$ is clearly trace-preserving for all $H \in \Herm(A)$ and $\mrm{span}(\Herm(A)) = \Lin(A)$, so, by linearity, $C$ is trace-preserving. As the composition of two trace-preserving $\Phi^{\Trans}$ is trace preserving. Combining these points completes the lemma.
\end{proof}

\begin{lemma}\label{lemma:adjoint-choi}
Let $A \cong \mbb{C}^{\Lambda}, B \cong \mbb{C}^{\Sigma}$, $\Phi \in \Trans(A,B)$. Then $J_{\Phi^{\ast}} =  \mrm{SWAP}(J_{\Phi})^{T}$,
where the SWAP is between the $A$ and $B$ spaces.
\end{lemma}
\begin{proof}
\sloppy First recall that given a Kraus decomposition of the linear map $\Phi \in \Trans(A,B)$, $\{A_{k},B_{k}\}_{k \in \Xi}$, $\Phi(X) = \sum_{k} A_{k} X B_{k}^{\ast}$ and the Choi operator is given by $J_{\Phi}~=~\sum_{k \in \Xi}\vec(A_{k})\vec(B_{k})^{\ast}$. It also follows $J_{\Phi^{\ast}} = \sum_{k} \vec(B_{k}^{\ast})\vec(A_{k}^{\ast})^{\ast}$. Let $A_{k} = \sum_{i,j} \alpha_{i,j}^{k} E_{i,j}$, $B_{k} = \sum_{i',j'} \beta_{i',j'}^{k} E_{i',j'}$, where we will drop the $k$ index for simplicity.
Then, we see
\begin{align*}
     \vec(B^{\ast}_{k})\vec(A_{k}^{\ast})^{\ast}
    = & \sum_{i,j,i',j'} \overline{\beta}_{i',j'} \alpha_{i,j} E_{j',j} \otimes E_{i',i} \\
    \xrightarrow[]{SWAP} & \sum_{i,j,i',j'} \alpha_{i,j}\ol{\beta}_{i',j'} E_{i',i} \otimes E_{j',j} \\
    = & \left(\sum_{i,j,i',j'} \alpha_{i,j}\ol{\beta}_{i',j'} E_{i,i'} \otimes E_{j,j'}\right)^{\Trans} 
    = \left(\vec(A_{k})\vec(B_{k})^{\ast}\right)^{\Trans} \ .
\end{align*}
By linearity of the SWAP and transpose, summing over $k$ completes the proof.
\end{proof}

\subsection{Proofs for Multiplicativity}
\begin{proof}[Proof of Theorem \ref{thm:multiplicativity-of-norms}]
First note that for arbitrary $X \in \Herm(A \otimes B)$, \begin{align*}
\|X\|_{\Ent_{k}} = \sup\{|\bra{w}X\ket{v}| : \mrm{SR}(\ket{w}),\mrm{SR}(\ket{v}) \leq k \} \ , 
\end{align*}
which we will make use of later in the proof.
Now let $\rho \in \Density(A_{1} \otimes B_{1})$ be non-negative (in the computational basis). Let $\sigma \in \Density(A_{2} \otimes B_{2})$ be arbitrary. It follows $\rho \otimes \sigma \in \Pos(A^{2}_{1} \otimes B^{2}_{1})$. Therefore, by the primal problem for the conic norm of a positive semidefinite operator \eqref{eqn:KNormPrimal} and that a convex optimization problem is maximized by an extreme point of the set being optimized over, without loss of generality, let $\ket{\psi^{\star}}$ be the optimizer of $\|\rho \otimes \sigma\|_{\Ent_{r}(A^{2}_{1}:B^{2}_{1})}$. It follows $\mrm{SR}(\ket{\psi^{\star}}) \leq r$. It follows,
\begin{align*}
    \ket{\psi^{\star}} = \sum_{i\in[r]} \alpha_{i} \ket{\varphi_{i}}_{A^{2}_{1}} \otimes \ket{\wt{\varphi}_{i}}_{B^{2}_{1}} 
    = \sum_{j,k} a_{j}b_{k} \ket{j,k}_{A^{2}_{1}} \otimes \left( \sum_{i \in [r]} \alpha_{i} \ket{\psi^{i}_{j}}_{B_{1}}\ket{\phi^{i}_{k}}_{B_{2}}\right) \ ,
\end{align*} 
where $a_{j},b_{k} \geq 0$ for all of their indices and $\sum_{i} |\alpha_{i}|^{2} = \sum_{j} a_{j}^{2} = \sum_{k} b_{k}^{2} = 1$ and $\ket{\psi^{i}_{j}},\ket{\phi^{i}_{k}}$ are normalized states. In the first line we have used the definition of Schmidt rank to decompose the state into product states, and in the second we have decomposed the halves of the product states so that the $A$ spaces are expressed in the computational basis and the coefficients are all real by absorbing any phase into the definitions of $\ket{\psi^{i}_{j}},\ket{\phi^{i}_{k}}$.

Using this decomposition, we obtain
\begin{align*}
    \|\rho \otimes \sigma\|_{\Ent_{r}}
    = & \bra{\psi^{\star}} \rho \otimes \sigma \ket{\psi^{\star}} \\
    = & \Big| \sum_{\substack{j,k,\ol{j},\ol{k}}} a_{j}b_{k} a_{\ol{j}}b_{\ol{j}} \bra{j,k}\rho\ket{\ol{j},\ol{k}} \cdot \sum_{i,\ol{i} \in [r]} \ol{\alpha}_{i}\alpha_{\ol{i}}\bra{\psi^{i}_{j}}\bra{\phi^{i}_{k}}\sigma \ket{\psi^{\ol{i}}_{\ol{j}}}\ket{\phi^{\ol{i}}_{\ol{k}}} \Big| \\
    = & \sum_{\substack{j,k,\ol{j},\ol{k}}} a_{j}b_{k} a_{\ol{j}}b_{\ol{j}} \bra{j,k}\rho\ket{\ol{j},\ol{k}} \cdot \Big|\sum_{i,\ol{i} \in [r]} \ol{\alpha}_{i}\alpha_{\ol{i}}\bra{\psi^{i}_{j}}\bra{\phi^{i}_{k}}\sigma \ket{\psi^{\ol{i}}_{\ol{j}}}\ket{\phi^{\ol{i}}_{\ol{k}}} \Big| \\
    \leq & \left(\sum_{\substack{j,k,\ol{j},\ol{k}}} a_{j}b_{k} a_{\ol{j}}b_{\ol{j}} \bra{j,k}\rho\ket{\ol{j},\ol{k}} \right) \|\sigma\|_{\Ent_{r}(A_{2}:B_{2})} \\
    \leq & \|\rho\|_{\Ent_{1}(A_{1}:B_{1})}\|\sigma\|_{\Ent_{r}(A_{2}:B_{2})} \ ,
\end{align*}
where the third equality uses that $\rho$ is non-negative, the first inequality is noting that $\sum_{i \in [r]} \alpha_{i} \ket{\psi^{i}_{j}}\ket{\phi^{i}_{k}}$ is a quantum state with Schmidt rank $\leq r$, so by the $\|\cdot\|_{\Ent_{r}}$ for Hermitian matrices, it's upper bounded by the norm. The last inequality is noting that $v := \sum_{j,k} a_{j}b_{k} \bra{j,k}$ is also a separable quantum state, so by the conic program for the norm, we have the bound.

Moreover, note 
\begin{align*} 
    \|\rho \otimes \sigma \|_{\Ent_{r}(A^{2}_{1}:B^{2}_{1})}
    \geq& \|\rho\|_{\Ent_{r}(A_{1}:B_{1})} \|\sigma\|_{\Ent_{r}(A_{2}:B_{2})} \\
    \geq& \|\rho\|_{\Ent_{1}(A_{1}:B_{1})} \|\sigma\|_{\Ent_{r}(A_{2}:B_{2})} \ , 
\end{align*}
where the first inequality is because from the primal problem \eqref{eqn:KNormPrimal}, it's clear the entanglement rank norm is always at least multiplicative and the second inequality is because the entanglement rank norms monotonically increase in $r$. Combining the two directions of inequalities forces an inequality and completes the proof.
\end{proof}

\begin{proof}[Proof of Proposition \ref{prop:PPT-max-relative-entropy-SDP}]
We start from (the equality version of) Eqn. \eqref{eqn:conicPrimal} and just use $X \in \PPT$ if and only if $X \in \Pos$ and $X^{\Gamma} \in \Pos$ where $X^{\Gamma}$ is the partial transpose. Then we can trivially relax the equality to inequality but for the consideration of deriving the dual, we stick with the equality. We now derive the dual. Introducing a slack variable, we can write the primal in standard form where 
\begin{align*}
    \wt{X} &= \wt{\text{diag}}(X,Z) \in \Pos(A \oplus A) \\
    B &= 1 \oplus 0 \in \mbb{R} \oplus \Lin(A) \\
    A &= \text{diag}(P,0) \in \Pos(A \oplus A) \\
    \Phi(\wt{X}) &= (\langle Q , X \rangle) \oplus (X^{\Gamma} - Z) \in \mbb{R} \oplus \Lin(A) \ .
\end{align*}
Then it is easy to determine the adjoint map is $\Phi^\ast(\wt{Y}) = (\gamma Q + Y^{\Gamma}) \oplus (-Y)$ where $\wt{Y} = \wt{\text{diag}}(\gamma,Y)$. Plugging all of this into the standard form of the dual, we get 
$$ \min\{\gamma \in \mbb{R} : \gamma Q + Y^{\Gamma} \succeq P \, , \, -Y \succeq 0 \} \ . $$
Redefining $Y := -Y$ and moving it to the other side completes the proof.
\end{proof}

\begin{proof}[Proof of Proposition \ref{prop:WH-Ent-norm-LP}]
First note by \eqref{eqn:KNormPrimal} (relabeling $\rho$ with $\sigma$ to avoid confusion),
$$ \langle \rho_{\lambda}^{\otimes 2} , \sigma \rangle = \langle \rho_{\lambda}^{\otimes 2}, \cU^{\otimes 2}(\sigma) \rangle  \ , $$
where $\cU$ is the twirling map and it's a self-adjoint map. As this twirling can be applied without loss of generality, the optimal $\sigma$ will be invariant under such twirling. Thus we have
$$ \sigma = w \Pi_+^{\otimes 2} + x \Pi_+ \otimes \Pi_- + y \Pi_- \otimes \Pi_+ + z \Pi_{-}^{\otimes 2} \ . $$
For this to be a positive operator, it follows $w,x,y,z \geq 0$. Moreover, to be unit trace, we require
$$ \left[(d+1)^{2}w + (x+y)(d+1)(d-1) + (d-1)^{2}z\right] = \frac{4}{d^2} \ , $$
where we have used $\Tr[\Pi_+] = \binom{n+1}{2}$, $\Tr[\Pi_-] = \binom{n}{2}$. Moreover, for this decomposition and using \eqref{eq:WH-defn-lambda}, 
$$ \ip{\rho_{\lambda}^{\otimes 2}}{\sigma \rangle} = \lambda^{2}w + \lambda(1-\lambda)(x+y) + (1-\lambda)^{2}z \ . $$
As noted in the main text, any such separable operator (i.e. $\Ent_{1}(A^{2}_{1}:B^{2}_{1})$) with this invariance is in fact $\PPT$. Note that 
$$\Pi_+^{\Gamma} = \frac{1}{2d}(\phi^{\perp} + (d+1)\phi^+) \quad \Pi_-^{\Gamma} = \frac{1}{2d}(\phi^{\perp} - (d-1)\phi^+) \ ,$$
where $\phi^{\perp} = d\I - \phi^+$ and $\phi^+$ is the unnormalized maximally entangled state and thus form a basis for the partial transpose space. Then it follows
\begin{align*}
    4d^{2}\sigma^{\Gamma} = & (w+x+y+z) {\phi^{\perp}}^{\otimes 2} \\ 
    & +\left\{(w+y)(d+1)-(x+z)(d-1)\right\}\phi^{\perp}\otimes\phi^{+} \\
    & + \left\{(w+x)(d+1)-(y+z)(d-1)\right\}\phi^{+} \otimes \phi^{\perp} \\
    & + \Big\{w(d+1)^{2} + z(d-1)^{2} - (x+y)(d-1)(d+1)\Big\} {\phi^{+}}^{\otimes 2} \ .
\end{align*}
Noting these are mutually orthogonal subspaces and that the factor of $4d^{2}$ is irrelevant for the positivity of these coefficients results, we have four constraints. Noting that $w+x+y+z \geq 0$ is guaranteed by each being required to be non-negative completes the proof.
\end{proof}

We make use of the following propositions in proving Theorem \ref{thm:WH-separation}.
\begin{proposition}\label{prop:Umegaki-Ent-WH}
 For $\lambda \in [0,1]$, define $\kappa_{\lambda} = \lambda \log\left[ \frac{\lambda}{1-\lambda} \frac{d-1}{d+1} \right] + \log\left[ \frac{2(1-\lambda)}{d(d-1)} \right]$. Then
 $$ \kappa_{\lambda} + \log(d) \geq D(\sigma_{\lambda}||\I_{A} \otimes \omega) \geq \kappa_{\lambda}$$ where the upper bound is saturated if $\omega = \frac{1}{d} \I_{B}$, and the lower bound is saturated if $\omega$ is pure.
\end{proposition}
\begin{proof}
First note $D(\sigma_{\lambda}||\I_{A} \otimes \omega) = \Tr[\sigma_{\lambda} \log \sigma_{\lambda}] - \Tr[\sigma_{\lambda} \I_{A} \otimes \log \omega]$, where we have used that $\log(\I_{A} \otimes P) = \log(I_{A}) \otimes P + \I_{A} \otimes \log P = \I_{A} \otimes \log P$ for positive definite $P$. We focus on each term separately.
\begin{align*}
    \Tr[\sigma_{\lambda} \log \sigma_{\lambda}]
    = & \lambda \alpha_{0} \Tr[\Pi_{0} \log (\lambda \alpha_{0} \Pi_{0})] + (1-\lambda)\alpha_{1} \Tr[\Pi_{1} \log((1-\lambda) \alpha_{1} \Pi_{1})] \\
    = & \lambda \alpha_0 \log(\lambda \alpha_0) \alpha_{0}^{-1} + (1-\lambda)\alpha_1 \log((1-\lambda)\alpha_1) \alpha_1^{-1} \\
    = & \lambda \log(\lambda \alpha_{0}) + \log((1-\lambda)\alpha_{1})-\lambda \log((1-\lambda)\alpha_{1}) \\
    = & \lambda \log(\frac{\lambda}{1-\lambda} \frac{\alpha_{0}}{\alpha_{1}}) + \log((1-\lambda)(\alpha_{1})) \\
    = & \log\left[\frac{\lambda}{1-\lambda}\frac{d-1}{d+1}\right] + \log\left[\frac{2(1-\lambda)}{d(d-1)}\right] \ ,
\end{align*}
where the first equality is just expanding the definition of $\sigma_{\lambda}$ and throwing out the terms which have disjoint supports, the second is that its the same projector up to some factor so the trace is that factor times the rank of the projector, the third equality is just cancelling factors, the fourth inequality is combining terms, the fifth equality is simplifying the constants.

Now we consider the second term:
\begin{align*}
    \Tr[\sigma_{\lambda} \I_{A} \otimes \log \omega] = & 1/d \Tr[\I_{B} \log \omega] = 1/d \Tr[\log \omega] \ ,
\end{align*}
where we have used the partial trace adjoint and that $\Tr_{A}(\sigma_{\lambda}) = \frac{1}{n}\I_{B}$ for the first equality. Now note
\begin{align*}
    \max_{\omega \in \Density(B)} D(\sigma_{\lambda}||\I_{A} \otimes \omega) &= \text{term 1} + \frac{1}{d} \max_{\omega \in \Density(B)} - \Tr[\log \omega] \\
    &= \text{term 1} - \frac{1}{d} \min_{p \in \mathcal{P}(|B|)} \sum_{i} \log(p_{i}) \ ,
\end{align*}
where that last optimization problem is achieved by $p_{i} = |B|^{-1} \, \forall i$, i.e. the maximally mixed state. Note $-\frac{1}{d}\Tr[\log 1/d \I_B] = \log d$. Note if we consider minimizing $D$, then we maximize over that last term, which must be non-positive, so the optimal is achieved by $p_{i} = 1$ for some choice of $i$ which is equivalent to $\omega$ being pure, which completes the proof. 
\end{proof}

\begin{proposition}\label{prop:sep_max_holevo_simple}
For all $\omega \in \Density(B)$,
\begin{align}\label{eqn:sep_max_holevo_simple}
D^{\Sep}_{\max}(\sigma_{\lambda}||\I_{A} \otimes \omega) =
\begin{cases}
    \log\left[\frac{2\lambda}{d+1}\right] & \lambda \geq \frac{d+1}{2d} \\[2mm]
    \log\left[ \frac{d+1-2\lambda}{d^{2}-1} \right]  & \text{ otherwise } \ .
\end{cases}
\end{align}
\end{proposition}
\begin{proof}
Recall that the Werner states are invariant under $\int dU (U \otimes U) \cdot (U \otimes U)^{\ast}$, so using \eqref{eqn:conicPrimal}, we can twirl the optimizer $X$ such that $X = x \I + y \mbb{F} = (x+y)\Pi_0 + (x-y)\Pi_1$. It follows the objective function is 
\begin{equation}
    \lambda \alpha_{0}(x+y) \langle \Pi_0 \rangle + (1-\lambda)\alpha_{1}(x-y) \langle \Pi_{1} \rangle =  x + (2\lambda-1) y \ , \label{DsepWH:obj}
\end{equation} 
where the $\alpha_{0},\alpha_{1}$ cancel due to the trace.
We can simplify the constraint as
\begin{align*}
\langle x\I_{AB} + y \mbb{F}, \I_{A} \otimes \omega \rangle = 1 \Rightarrow dx + y = 1\ ,
\end{align*}
where we have used $\Tr_{A}[\mbb{F}] = \I_{B}$. Therefore, without loss of generality, $y = 1-dx$.
Note this makes the objective function singular variate. 

We also know that the space of $UU$ invariant operators, $\text{Inv}(UU)$, is such that $\Sep \cap \text{Inv}(UU) = \PPT \cap \text{Inv}(UU)$. Therefore we can just get the postiviity and PPT constraints. The positivity constraints are $x \pm y \geq 0, x \geq 0, x+yd \geq 0$ as derived in \cite[Proposition 6]{Chitambar-2021a}. By the same arguments as in that derivation, a bunch of the constraints are redundant and we reduce to
\begin{align*}
    \max \left\{ x + (2\lambda - 1)(1-dx) : \frac{1}{d+1} \leq x \leq \frac{d}{d^{2}-1} \right \}
\end{align*}
As the derivative of the objective function is non-negative for $\lambda \leq \frac{d+1}{2d}$, we have
$$ x^{*} = \begin{cases} d/(d^{2}-1) & \lambda \leq \frac{d-1}{2d} \\
1/(d+1) & \text{ otherwise } \end{cases} \ . $$
Plugging this into the objective function and simplifying completes the proof.
\end{proof}

\begin{proof}[Proof of Theorem \ref{thm:WH-separation}]
Let $\lambda =0$ and $\omega = \pi_{B}$. Then 
\begin{align*}
D(\sigma_{\lambda}||\I_{A} \otimes \pi) =& -\log(d-1) + 1
> -\log(d) = D^{\PPT}_{\max}(\sigma_{\lambda}||\I_{A} \otimes \pi) \, ,
\end{align*}
where the first step follows from Proposition \ref{prop:Umegaki-Ent-WH} as $\kappa_{0} + \log(d) = \log(2) - \log(d-1)$ and the last step is by Proposition \ref{prop:sep_max_holevo_simple} by plugging in $\lambda = 0$ to find $\log(d+1 /(d^{2}-1)) = \log(d+1/((d+1)(d-1))) = -\log(d-1)$. Clearly this holds for all $d \in \mathbb{N}$.

To prove the second point, for $\lambda = 1/2$ note by Proposition \ref{prop:Umegaki-Ent-WH}, $D(\sigma_{\lambda}||\I_{A} \otimes \omega) = -\log(\sqrt{d^{2}-1})$ and by Proposition \ref{prop:sep_max_holevo_simple}, $D^{\Sep}_{\max}(\sigma_{\lambda}||\I_{A} \otimes \omega) = \log(d/(d^{2} -1))$. Now noting that $\frac{1}{\sqrt{x^{2}-1}} = \frac{\sqrt{x^{2}-1}}{x^{2}-1} < \frac{\sqrt{x^{2}}}{x^{2}-1} = \frac{x}{x^{2} - 1}$ so by monotonicity of the logarithm, the two values we are considering will always have an arbitrarily small gap. Now we just want to show that any smaller value can be achieved. Note that we can expand the max-relative entropy and standard relative entropy:
\begin{align*}
 D_{\max}(\sigma_{\lambda}||\I \otimes \pi) &= \log[d+1-2\lambda] - \log[d+1] - \log[d-1] \\
D(\sigma_{\lambda}||\I \otimes \pi)
& = \lambda \log[\lambda/(1-\lambda)] - \log[d+1] + \log[2-2\lambda] \ ,
\end{align*}
where in the first we just used $d^{2}-1 = (d+1)(d-1)$ and the second we expanded and cancelled some terms. Now note both terms have a ``$-\log[d+1]$" term, which we will be able. The other term for the standard relative entropy may be simplified in the following manner
\begin{align*}
     \lambda \log[\lambda/(1-\lambda)] + \log[2-2\lambda]
    =& \lambda \log[\lambda] - \lambda \log[1-\lambda] + \log[2] +\log[1-\lambda] \\
    =&  1 + \lambda \log[\lambda] + (1-\lambda)\log[1-\lambda] \\
    = & 1 - h(\lambda) \ ,
\end{align*}
where $h(\cdot)$ is the binary entropy. Therefore, 
\begin{align*}
 D_{\max}^{\Sep}(\sigma_{\lambda}||\I \otimes \pi) < D(\sigma_{\lambda}||\I \otimes \pi)
\Leftrightarrow  \log[(d+1-2\lambda)/(d-1)] < 1 - h(\lambda) \ . 
\end{align*}
For $\lambda < 1/2$, $1-h(\lambda) = x > 0$ and $0 < \ve := 1-2\lambda$. Therefore, we just need $\frac{d+\ve}{(d-1)} < \exp(x)$. But such a $d$ clearly exists as $\exp(x) > 1$ and $\frac{d + \ve}{(d-1)} \to 1$ as $d \to \infty$ by L'H\^{o}pital's rule.
\end{proof}

\bibliographystyle{IEEEtran}
\bibliography{cone-restricted-bib}{}
\end{document}